%% file: 0_main.tex
\documentclass[11pt]{article}

\input{macro}

\input{math}

\ifodd 1
\newcommand{\wt}[1]{{\color{red}(WT: #1)}}
\else
\newcommand{\wt}[1]{}
\fi
\newcommand{\ignore}[1]{}
\newcommand{\wtr}[1]{{\color{black}  #1}}
\newcommand{\wtrevision}[1]{{\color{black}  #1}}
\newcommand{\cjr}[1]{{\color{black}  #1}}
\newcommand{\yfr}[1]{{\color{black}  #1}}

\newif\ifEC
\ECtrue

\newif\ifArxiv
\Arxivtrue

\title{Rationality-Robust Information Design: \\
Bayesian Persuasion under Quantal Response}

\author{
Yiding Feng\thanks{University of Chicago, 
\texttt{yidingfeng@uchicago.edu}} \and 
Chien-Ju Ho\thanks{Washington University in St. Louis, \texttt{chienju.ho@wustl.edu}} \and
Wei Tang\thanks{Columbia University, \texttt{wt2359@columbia.edu}}
}

\begin{document}

\maketitle

\begin{abstract}

\input{Paper/abstract-new}
\end{abstract}

\newpage

\section{Introduction}
\label{sec:intro}
\input{Paper/intro}

\ifArxiv
\subsection{Related Work}
\label{subsec:related}
\input{Paper/related}

\fi

\section{Preliminaries}
\label{sec:prelim}

\input{Paper/prelim}

\section{State Independent Sender Utility (SISU) Environments}
\label{sec:state independent}
\input{Paper/state-independent}

\section{State Dependent Sender Utility (SDSU) Environments}
\label{sec:state dependent}
\input{Paper/state-dependent}

\section{Rationality-Robust Information Design}
\label{sec:robust}
\input{Paper/robust}
\section{Conclusions and Future Important Directions}
\label{subsec:future}
\input{Paper/future-directions}

\section*{Acknowledgement}
We thank the anonymous reviewers for the helpful comments. 
This work is supported in part by the Office of Naval Research under Grant N00014-20-1-2240 and a J.P. Morgan Faculty Research Award.

\bibliography{mybib}

\appendix

\section{Further Related Work}
\label{apx:further related work}
\input{Paper/further-related}

\section{Motivating Example and Extensions}
\label{apx:discussion}

\input{Paper/apx-discussion-extension}

\ifArxiv
\else
\section{Related Work}
\label{subsec:related}

\input{Paper/related}
\section{Future Important Directions}
\label{subsec:future}
\input{Paper/future-directions}
\fi

\section{Omitted Proofs in Section \ref{sec:prelim}}
\label{apx:prelim proof}
\input{Paper/apx-prelim}

\section{Omitted Proofs in Section \ref{sec:state independent}}
\label{apx:SISU proof}

\input{Paper/apx-independent}

\section{Omitted Proofs in Section \ref{sec:state dependent}}
\label{apx:SDSU proof}

\input{Paper/apx-dependent}

\section{Omitted Proofs in Section \ref{sec:robust}}
\label{apx:robustness proof}

\input{Paper/apx-robustness}

\section{The Complexity on Computing Approximately Optimal 
Signaling Schemes}
\label{apx:complexity}
\input{Paper/apx-complexity.tex}








\end{document}

%% file: macro.tex
\usepackage[parfill]{parskip}
\usepackage[letterpaper, left=1in, right=1in, top=1in, bottom=1in]{geometry}

\usepackage{pgfplots}
\usepackage{tikz}
\usepackage[caption=false]{subfig}
\usetikzlibrary{topaths,calc}
\usepackage{soul}

\usepackage{amsmath,amssymb}
\usepackage{xcolor}
\usepackage{bbm}
\usepackage{multicol}

\usepackage{natbib}
\bibpunct{(}{)}{;}{a}{,}{,}

\usepackage[utf8]{inputenc} 
\usepackage[T1]{fontenc}    
\usepackage[colorlinks=true, linkcolor=blue!70!black, citecolor=blue!70!black,urlcolor=black,breaklinks=true]{hyperref}

\usepackage{hyperref}

\usepackage{comment}

\usepackage{url}            
\usepackage{booktabs}       
\usepackage{amsfonts}       
\usepackage{nicefrac}       
\usepackage{microtype}      

\usepackage{mathtools, amsthm, bm}

\usepackage{array}
\usepackage{makecell}
\newcolumntype{C}[1]{>{\centering\arraybackslash}p{#1}}

\usepackage{multirow}
\usepackage{float}
\usepackage{xspace}

\usepackage{enumitem}

\usepackage{xfrac}

\usepackage{cleveref}

\usepackage{thmtools}
\theoremstyle{plain}
\newtheorem{theorem}{Theorem}[section]
\newtheorem{lemma}[theorem]{Lemma}

\newtheorem{proposition}[theorem]{Proposition}
\newtheorem{corollary}[theorem]{Corollary} 

\usepackage{thm-restate}

\theoremstyle{plain}
\newtheorem{definition}{Definition}[section] 
\newtheorem{example}[definition]{Example}
\newtheorem{remark}[definition]{Remark}

\theoremstyle{plain}

\newcounter{relctr} 
\everydisplay\expandafter{\the\everydisplay\setcounter{relctr}{0}} 

\AtBeginDocument{} 

\newcommand{\prob}[2][]{\text{Pr}\ifthenelse{\not\equal{}{#1}}{_{#1}}{}\!\left[{\def\givenn{\middle|}#2}\right]}
\newcommand{\expect}[2][]{\mathbb{E}\ifthenelse{\not\equal{}{#1}}{_{#1}}{}\!\left[{\def\givenn{\middle|}#2}\right]}

\newcommand{\indicator}[2][]{\mathbbm{1}\ifthenelse{\not\equal{}{#1}}{_{#1}}{}\!\left[{\def\givenn{\middle|}#2}\right]}

\newcommand{\squishlist}{
\begin{list}{{{\small{$\bullet$}}}}
{\setlength{\itemsep}{3pt}      \setlength{\parsep}{1pt}
\setlength{\topsep}{1pt}       \setlength{\partopsep}{0pt}
\setlength{\leftmargin}{1em} \setlength{\labelwidth}{1em}
\setlength{\labelsep}{0.5em} } }
\newcommand{\squishend}{  \end{list}}

\newcommand{\prior}{\lambda}

\newcommand{\stateNum}{m}
\newcommand{\action}{a}

\newcommand{\senderU}{u}
\newcommand{\receiverU}{v}
\newcommand{\state}{i}

\newcommand{\signal}{\sigma}
\newcommand{\signalProb}{\pi}

\newcommand{\posterior}{\mu}

\newcommand{\noiseScale}{\beta}

\newcommand{\noise}{\varepsilon}

\newcommand{\scaledNoiseDiff}{\delta}

\newcommand{\noiseDualVar}{\alpha}
\newcommand{\distDualVar}{\eta}

\newcommand{\approxDualVar}{\tau}

\newcommand{\argmax}{\mathop{\mathrm{arg\,max}}}

\newcommand{\Payoff}[2][]{\text{\bf Payoff}\ifthenelse{\not\equal{}{#1}}{_{#1}}{}\!\left[{\def\givenn{\middle|}#2}\right]}

\newcommand{\zerobf}{\mathbf{0}}

\newcommand{\convexcombin}{f}
\newcommand{\convexcombinbf}{\mathbf{\convexcombin}}

\newcommand{\tangentIntersect}{\kappa}

\newcommand{\censorshipsignalInverse}{{\scaledNoiseDiff}\doubleprimed}

\newcommand{\NOinforSignal}{\scaledNoiseDiff^{\cc{avg}}}

\newcommand{\Quant}{W}

\newcommand{\highstates}{\mathcal{H}}
\newcommand{\lowstates}{\mathcal{L}}

\newcommand{\randomstate}{\theta}
\newcommand{\stochastic}{\Delta}
\newcommand{\signalscheme}{\pi}
\newcommand{\signalspace}{\Sigma}

\newcommand{\betaed}{^{(\noiseScale)}}
\newcommand{\primed}{^\dagger}
\newcommand{\doubleprimed}{^{\ddagger}}

\newcommand{\censorshipprob}{p\primed}
\newcommand{\censorshipstate}{i\primed}
\newcommand{\censorshipsignal}{\delta\primed}
\newcommand{\optscheme}{\signalscheme^*}

%% file: math.tex
\newcommand{\xhdr}[1]{\vspace{-5pt}\paragraph*{
\bf {#1.}}}

\newcommand{\cc}[1]{\ensuremath{\mathsf{#1}}} 

\newcommand{\cA}{\mathcal{A}}

\newcommand{\cS}{\mathcal{S}}

\newcommand{\R}{\mathbb{R}}

\newcommand{\N}{\mathbb{N}}
\newcommand{\naturals}{\N}

\newcommand{\supp}{\cc{supp}}

%% file: Paper/abstract-new.tex
Classic mechanism/information design imposes the assumption that agents are \emph{fully rational}, meaning each of them always selects the action that maximizes her expected utility. Yet many empirical evidence suggests that human decisions may deviate from this full rationality assumption. In this work, we attempt to relax the full rationality assumption with \emph{bounded rationality}. Specifically, we formulate the bounded rationality of an agent by adopting the quantal response model~\citep{MP-95}. 

We develop a theory of rationality-robust information design in the canonical setting of Bayesian persuasion~\citep{KG-11} 
with binary receiver action. We first identify conditions under which the optimal signaling scheme structure for a fully rational receiver remains optimal or approximately optimal for a boundedly rational receiver. In practice, it might be costly for the designer to estimate the degree of the receiver’s bounded rationality level. 
Motivated by this practical consideration, 
we then study the existence and construction of \emph{robust} signaling schemes when there is uncertainty about 
the receiver's bounded rationality level.\footnote{\footnotesize{\emph{A preliminary conference version of this work has appeared in the Proceeding of the ACM-SIAM Symposium on Discrete Algorithm (SODA 2024).}}}

%% file: Paper/intro.tex
In modern computer science, 
an important branch 
of research studies 
computation that involves 
multiple parties.
One fundamental question raises 
in this area:
how to ensure that different parties 
do their computation correctly.
In the field of mechanism design
(which studies protocols 
for strategic agents),
literature imposes the rationality assumption
on the participated agents
-- each agent (a.k.a., party) acts (a.k.a., does computation) in a way 
to {perfectly} maximize their own utility.
In contrast, the system design literature 
favors fault tolerant systems 
\citep{cri-91,lap-92,KK-20},
where the central protocol allows 
faults within some of the parties.
Motivated by the fault tolerance idea 
from the system design literature,
it is interesting to study 
whether and how 
the economic lessons 
derived under the full rationality assumption 
can be extended to more practical scenarios
where agents might make mistakes and thus are 
\emph{boundedly rational}.

In this paper, we tackle this research direction
on relaxing rationality assumption by studying 
a canonical economic model --
persuasion
--
in information design.\footnote{Information design is a field 
closely related to mechanism design. 
Mechanism design builds the rule of the game while 
holding the information structure (i.e., how information is transferred across agents and nature) fixed. 
In contrast, information design builds information structure
while holding the rule of the game fixed.}
In Bayesian persuasion \citep{KG-11}, 
there is a sender and a receiver.
Both players have their own utility functions
which depend on a state drawn from a common prior 
as well as an action selected by the receiver.
Once the state is realized,
the sender 
observes the realized state, while the receiver 
only shares a common prior about the state with the sender.
The sender can commit to an information structure
(a.k.a., signaling scheme) which 
(possibly randomly) maps the realized state to a signal
sent to the receiver.
Given the observed signal, the receiver forms 
a posterior belief about the state and then selects
her action (which impacts both her and the sender's utilities).
We say the receiver is fully rational if she always 
(correctly) selects the best action which maximizes
her expected utility given her posterior belief 
about the state.

To capture the possible mistakes which the receiver can make
in practice, we relax the full rationality assumption 
with the bounded rationality modeled by
\emph{quantal response} \citep[cf.][]{MP-95}.\footnote{The quantal response model is also known
as multinomial logit model \citep[cf.][]{TV-04} and conditional choice probability \citep{rus-87}.
There are other
models which relax the rationality assumption. See related works for more discussions.}
\wtrevision{
To provide informal intuitions, for a fully rational receiver,
she takes an action that maximizes her expected utility. 
When there is no ties in action utility, this action choice is
deterministic. 
On the other hand, the quantal response
accounts for the inherent randomness (and error-proneness) in human decision
making and models the human's decision as a probabilistic process.
Specifically, in quantal response, for each action the receiver can take, a noise is added
into the receiver's utility for taking this action. The receiver then takes an
action that maximizes this noisy version of the utility. This
noise captures several realistic aspects of human decision
making, e.g., when there are additional inherent characteristics in the receiver's utility estimation that we cannot model,
or when receiver is drawn from a population and individual
differences need to be accounted for.}
\wtrevision{Shifting to the above boundedly rational behavior,
two natural questions that our work tries to answer are:}
\begin{quote}
    \emph{Does the structure of optimal signaling scheme 
    for a fully rational receiver 
    preserve or approximately preserve 
    when the receiver is boundedly rational?} 
    
    \emph{Can the sender design robust signaling scheme 
    when he has uncertainty of 
    the receiver's boundedly rational behavior?}
\end{quote}
To answer the above questions, we focus on 
Bayesian persuasion with binary receiver action. Though binary receiver action seems a little restrictive at first glance, it is a canonical persuasion model studied extensively 
in both theoretical computer science and economics literature 
\citep[see, e.g.,][]{KMZL-17,BB-17,GS-19,X-20,BTXZ-21}.
This model, serving as a fundamental cornerstone,
has a wide range of applications in practice, 
including but not limited to  
product advertising, targeting in sponsored search, recommendation letter, and short video recommendation. See \Cref{sec:motivating example} 
for 
detailed descriptions of these examples.
Our results provide both affirmative and negative answers to the above questions, 
\wtrevision{and we underscore that the binary-action setting is sufficiently intricate and challenging enough to establish our main results within our rationality-robust framework.}
At a high-level, a critical condition influencing our findings is the sender's utility structure, specifically whether it is dependent on the state.
\footnote{As we elaborate later, 
for a fully rational receiver, the structure of optimal signaling scheme is well-characterized for the binary-action setting, 
while for the multi-action setting, characterizing a succinct structure that is amenable to theoretical analysis still remains as open question \citep[cf.][]{DX-16,BCGZ-22}. Thus, studying multi-action setting is beyond the focus on this work, and we leave it as an interesting future direction.
}

\subsection{Main Results and Techniques}

Based on the practical applications,
problems in Bayesian persuasion can be further partitioned 
into \emph{state independent sender utility (SISU)} environments
where the sender's utility 
does not depend on the realized state;
and \emph{state dependent sender utility (SDSU)} environments
where the sender's utility depends on both the realized state 
as well as the receiver's action. 
For example, as illustrated in \Cref{sec:motivating example},
the aforementioned 
product advertising and recommendation letter example
fall into SISU environments as 
the seller/advisor only cares whether the buyer 
buys the product/recruiter hires the student, 
while 
short video recommendation 
and targeting in sponsored search example 
fall into SDSU environments
as the platform's/search engine's
revenue also depends on video content/impression attribute. 

\xhdr{Revisiting censorship and direct signaling schemes}
When the receiver is fully rational, the optimal signaling schemes
admit the same structure for both SISU environments 
and SDSU environments.
In a nutshell, the optimal signaling scheme
partitions all states into two subsets -- high states and low states;\footnote{Rigorously speaking, there might exists a threshold state such that
a certain fraction of it belongs to high states and the remaining fraction of it belongs to low states.}
and pools all high states into a single signal.
On the other side, 
the signaling structure for low states can be arbitrary and 
does not affect the optimality of the signaling scheme.
Two representative
subclasses of signaling schemes have been studied extensively
in the literature --
\emph{direct signaling schemes} 
and \emph{censorship signaling schemes}.
Both of them pool all high states, but have different signaling structures for low states. 
Specifically, direct signaling schemes pool all low states,
while censorship signaling schemes reveal 
every low state truthfully.
To persuade a fully rational receiver,
the sender is indifferent between the optimal direct signaling scheme
and the optimal censorship signaling scheme, since both of them 
maximize the sender's expected utility 
over all signaling schemes.

\xhdr{The separation of optimal signaling schemes in 
SISU and SDSU environments}
As the first part of our main contributions,
for a boundedly rational receiver,
we show that 
in SISU environments, censorship signaling schemes 
remain optimal, while direct signaling schemes 
are sub-optimal;
and both of them become
sub-optimal in SDSU environments. 
Nonetheless, we also provide the tight approximation bounds
of censorship and direct signaling schemes
in SDSU environments.
Our results (summarized in \Cref{tabel:summary})
suggest that the structure of optimal signaling schemes 
for a fully rational receiver 
is partially preserved (i.e., censorship remains optimal)
in SISU environments,
and 
approximately preserved (i.e., up to an $\Theta(m)$-approximation factor) in SDSU environments, where $m$ denotes the number of the states.
Moreover, to persuade a boundedly rational receiver,
the sender prefers censorship than 
direct signaling schemes.

\input{Paper/table-result}

In more detail, in SISU environments, we
show that for any boundedly rational 
receiver, 
censorship 
is optimal among all signaling schemes
(\Cref{thm:SISU opt}) and
direct signaling scheme is $\Omega(m)$-approximation
where $m$ is the number of states (\Cref{thm:SISU direct}).
We further provide structural characterizations on 
how to determine the high/low states partition in 
the optimal censorship.
In particular, 
for a receiver with any bounded rationality level, 
including a fully rational receiver,
the subset of high states is {\em nested} 
(i.e., increasing) with respect to the bounded rationality level.
Namely, the optimal signaling scheme reveals 
less information for a more rational receiver.

To show the optimality of censorship signaling schemes 
for a boundedly rational receiver in SISU environments,
we first introduce a linear program~\ref{eq:opt lp},
in which the constraints regulate the set of all feasible signaling schemes, and the objective function 
computes the expected sender utility of a given 
signaling scheme.
This linear program~\ref{eq:opt lp} is inspired by
a connection between
our problem and public Bayesian persuasion for a continuum population
of fully rational receivers with a specific utility structure. 
Given the linear program~\ref{eq:opt lp} and its dual program, 
we characterize the censorship structure in the
optimal signaling scheme by constructing a dual assignment explicitly
and then invoke the strong duality of linear programs.

In SDSU environments, the optimal signaling scheme
no longer admits the censorship nor direct structure.
We start by providing a SDSU example (\Cref{example:SDSU lower bound example})
and showing that the approximation of every 
censorship (resp.\ direct)
signaling scheme 
is $\Omega(m)$
(\Cref{prop: SDSU censorship approx LB}).
En route to proving this lower bound,  
we present a stronger result, namely that any signaling scheme must be
an $\Omega(\sfrac{m}{L})$-approximation where $L$ is
the maximum number of signals
induced by a state in this signaling scheme 
(\Cref{thm:signal size lower bound}).
Next, we provide the matching upper bound
that for any problem instance with $m$ states,
there exists a censorship (resp.\ direct) signaling scheme that is 
an $O(m)$-approximation to the optimal signaling scheme 
(\Cref{coro:SDSU censorship m approx}).

The key step in 
establishing the $O(m)$-approximation upper bounds
for censorship (resp.\ direct) signaling schemes 
(\Cref{coro:SDSU censorship m approx}) is that 
we characterize a 4-approximation
signaling scheme that uses $O(m)$ signals and 
has the following two structural properties (\Cref{thm:SDSU 4 approx}): 
(i) every signal is used to pool at most two states;
(ii) every pair of states is pooled at most one signal.
Intuitively, property~(i)
says that,
in the signaling scheme that we characterize,
a signal either fully reveals
the state or randomizes receiver's uncertainty only on two states, 
and property~(ii) says that there is no need for the sender to 
pool a pair of states at multiple signals in order to have 4-approximation.
We then leverage the structure
of this 4-approximation signaling scheme 
to show the existence of $O(m)$-approximation
censorship (resp.\ direct) signaling schemes.
To prove this technical lemma (\Cref{thm:SDSU 4 approx}),
we build a connection between 
the signaling schemes satisfying properties~(i) (ii)
with fractional solutions in the {\em generalized assignment problem}
\citep{ST-93}.
In particular, 
focusing on signaling schemes that
have properties~(i) (ii), we introduce a linear program
which shares the same format 
as the linear program relaxation of the generalized assignment problem.
For generalized assignment problem,
\cite{ST-93} show that the optimal integral solution 
is a $2$-approximation to the optimal fractional solution.
We argue that the optimal integral solution 
can be converted into a feasible signaling scheme
that has properties~(i) (ii),
uses $O(m)$ signals,
and suffers an additional 
two factor loss in its payoff.

\xhdr{Rationality-robust 
information design}
As the second part of our main contributions,
we introduce \emph{rationality-robust 
information design}
--
a framework
in which a signaling scheme is
designed for a receiver whose bounded rationality 
level is unknown. 
In our previous discussions, designing optimal signaling schemes
in both SISU and SDSU environments 
is rationality-oriented
-- the sender needs to know exactly 
the receiver’s bounded rationality level.
In practice, the sender may not be able to have 
(or require significant cost to learn) such perfect
knowledge. 
Motivated by this concern, 
the goal of rationality-robust information design
is to identify robust
signaling schemes -- ones with good (multiplicative) approximation to the optimal
signaling scheme that is tailored to any possible
bounded rationality level of the receiver.
Similar to our results above, 
we observe that obtaining rationality-robust 
signaling scheme is much {more tractable} 
in SISU environments
than SDSU environments.

In SISU environments, 
by leveraging the structural property 
we mentioned before 
(i.e., the optimal (censorship) signaling scheme
reveals less information
for a more rational
receiver), 
we show that the optimal censorship
for a fully rational receiver achieves a $2$ 
rationality-robust approximation when the sender 
has no knowledge of the receiver's bounded rationality level 
(\Cref{thm:2 robust approx SISU upperbound}).
We also provide an
example to show the tightness of the result 
(\Cref{prop:2 robust approx SISU lower bound}).
Our result suggests that, {\em up to a two factor, 
the knowledge of the receiver's bounded rationality level are
unimportant in SISU environments.}
For the comparison, we also show that 
the optimal direct signaling scheme
for a fully rational receiver 
achieves unbounded rationality-robust approximation 
(\Cref{prop:unbounded robust direct approx SISU lower bound}).
This repeats the takeaway mentioned above -- 
the sender prefers censorship 
than direct signaling schemes in SISU environments
under bounded rationality.

In contrast, in SDSU environments, 
{we show that there exists no signaling
scheme with bounded rationality-robust approximation ratio, 
when the sender has no knowledge of the receiver’s bounded rationality level
(\Cref{thm:imposs robust approx SDSU}), 
and this result holds even if the state space is binary.}
Our result suggests that, {\em there exists a tradeoff between the knowledge of the receiver’s rationality level and the achievable rationality-robustness in SDSU environments.}
To show this impossibility result, we construct 
a binary-state problem instance and 
a set of carefully chosen possible bounded rationality levels. 
The key to our approach is by introducing a 
factor-revealing program to lower bound the optimal 
rationality-robust approximation ratio.
By analyzing its dual program, we show 
that the rationality-robust approximation ratio of any signaling scheme is unbounded.
This impossibility result indicates that there 
exists a tradeoff between the knowledge of the receiver's rationality level and
the achievable rationality-robustness. 
Though it appears challenging to obtain a bounded-factor 
rationality-robust approximation for arbitrary set of rationality levels,
and general problem instances in SDSU environments,
we obtain a preliminary positive result 
under a 
boundedness condition on the receiver's rationality level.
In particular, 
when the state space is binary,
under a {reasonable} 
multiplicative boundedness condition 
(i.e., learning the receiver’s bounded rationality level up to
a multiplicative error) 
on the receiver's bounded rationality levels,
we show that 
the sender is able to design a signaling scheme
whose rationality-robust approximation ratio 
depends linearly on the multiplicative error 
(\Cref{thm:SDSU binary robust}). 




%% file: Paper/table-result.tex
\captionsetup[table]{skip=10pt}

\begin{table}[!ht]
\fontsize{10}{18}\selectfont
\centering
\caption{Approximation ratio of censorship/direct signaling schemes under bounded rationality.
The number of states is $m$.}
\label{tabel:summary}
\begin{tabular}{|c|c|c|}
\cline{2-3}
\multicolumn{1}{c|}{} & censorship signaling schemes & direct signaling schemes
\\
\hline
SISU & $\begin{array}{cc}
     1 \;\;
     {\scriptsize \text{[\Cref{thm:SISU opt}]}}
\end{array}$
&
$\begin{array}{cc}
     \Theta(m) \;\;
     {\scriptsize \text{[\Cref{thm:SISU direct}, \Cref{coro:SDSU censorship m approx}]}}
\end{array}$
\\
\hline
SDSU 
&
\multicolumn{2}{c|}{
$\begin{array}{cc}
     \Theta(m) \;\;
     {\scriptsize \text{[\Cref{prop: SDSU censorship approx LB}, \Cref{coro:SDSU censorship m approx}]}}
\end{array}$
}
\\
\hline
\end{tabular}
\end{table}

%% file: Paper/related.tex
In this section, we discuss the works that are closely related to our work, and we discuss further related work in \Cref{apx:further related work}. 

There has been a growing interest in understanding how to design 
robust signaling schemes in the face of uncertain receiver behavior.
Our work contributes to this line of research by studying robust signaling schemes
when the receiver's bounded rationality level is unknown. 
The approach we take is similar to the approach often used in 
prior-independent mechanism design, examining 
the {\em approximation ratios} of the designed mechanisms.
Our work differs from previous works that either focus on the regret minimization \cite{BTXZ-21,CL-23} or minimax approach \cite{DP-22,K-22,HW-21}.
Notably, \citet{BTXZ-21} present a 
negative result saying that there exists no nontrivial bound of the additive regret if the sender 
has no knowledge about the receiver's utilities. 
While this result shares a similar message to our impossibility result in \Cref{sec:robust}, there are notable differences between the two studies that preclude direct comparison:
(i) our impossibility result 
is under SDSU setting, whereas theirs is under SISU setting (for which we have a positive result);
(ii) the adversary in our setting is limited to 
choosing the receiver's behavior (i.e., the rationality level) 
in the quantal response model, 
whereas theirs considers a worst-case adversary.
Recent work by \citet{CL-23}
also examines the design of robust signaling schemes for non-best-responding receiver, but with a focus on the regret minimization approach.
Our persuasion setting with fully rational receiver
can also be viewed as a Stackelberg game where 
the sender moves first by committing to a signaling scheme, 
and the receiver takes an action that 
best responds to sender's signaling scheme.
With bounded rationality, receiver in our setting is not best-responding to sender's signaling scheme.
This shares similarity to recent work by \cite{GHWX-23} who study 
Stackelberg games with suboptimal follower response. 
However, their work adopts a worst-case perspective and considers the worst possible follower
behavior up to some plausible ranges, while our work adopts a model-based approach and the follower is responding with following a quantal response model.


%% file: Paper/prelim.tex


\subsection{Model and Problem Definition}
\label{sec:model}
In this paper, we study the  
persuasion problem for a receiver with bounded 
rationality.
There are two players, a {sender} and 
a {receiver}. 
There is an unknown state $\randomstate$
drawn from a finite set $[m] \triangleq \{1, 2,\dots, m\}$
according to a prior distribution $\prior \in \stochastic([m])$,
which is common knowledge among both players.
Throughout the paper, we use $\randomstate$ to denote 
the state as a random variable, and $i, j, k\in[m]$
as its possible realization.
We use $\prior_i$ to denote the probability that 
the realized state is $i\in[m]$, i.e., $\prior_i \triangleq\prob{\randomstate = i}$.
The receiver has a binary action set $\cA = \{0, 1\}$.
Given a realized state $i\in[m]$, by taking action $\action\in\cA$, 
the utility of the receiver is $\receiverU_i(\action)$
and the utility of the sender is $\senderU_i(\action)$.
Following the standard convention \citep[e.g.,][]{AIL-20,AC-16b,BTXZ-21,LI-19}, 
Throughout this paper, 
\wtr{
we focus on the setting where
$\senderU_i(1) \geq \senderU_i(0)$ for all $i\in[m]$,
and normalize $\senderU_i(0) \equiv 0$
and denote $\senderU_i \triangleq \senderU_i(1)$.
\footnote{
A discussion of how our results could be extended without this assumption
is provided in \Cref{sec:assump favorable action}.
}
}

The objective of the sender is to maximize
his expected utility.
Before the state $\randomstate$ is realized, the sender commits 
to a signal space $\signalspace$ and 
a signaling scheme $\signalscheme:[m] \rightarrow \stochastic(\signalspace)$,
a mapping from the realized state  
into probability distributions over signals.
We use $\signalscheme_i(\signal)$ to denote the probability that 
signal $\signal\in\signalspace$ is realized when the realized state is state $i$.
Upon seeing signal $\sigma$, the receiver 
performs a Bayesian update and infers 
a posterior belief over the state.
In particular, 
the posterior probability of state $i$ given realized signal $\signal$
is $\posterior_{\state}(\sigma)\triangleq \frac{
\prior_i\signalscheme_i(\signal)}{\sum_{j\in[m]}\prior_j \signalscheme_j(\signal)}$.

In this paper, we assume that the receiver is boundedly rational
by modeling her as a (logit) quantal response player
\citep{MP-95}. 
Specifically, instead of taking the action that maximizes the expected 
utility, 
a quantal player randomly selects an action with probability proportional
to the expected utility.
In our model, given posterior belief $\posterior\in\stochastic([m])$
and its induced expected utility 
$\receiverU(\action\mid \posterior) 
\triangleq \sum_{\state\in[\stateNum]}
\posterior_\state\, \receiverU_\state(\action)$
for action $\action\in\cA$, 
the receiver takes action 1
with probability\footnote{One explanation 
of this quantal response behavior is 
that the receiver
faces a random shock when she is making
the decision. See \Cref{sec:reformulate} for more details.}
\begin{align*}
    \frac{\exp(\noiseScale\cdot \receiverU(1\mid\posterior))}{
    \exp(\noiseScale\cdot \receiverU(1\mid\posterior)) + 
    \exp(\noiseScale\cdot \receiverU(0\mid\posterior))}
    = 
    \frac{1}{ 1+
    \exp(\noiseScale\cdot \left(\receiverU(0\mid\posterior) - \receiverU(1\mid\posterior)\right))}
\end{align*}
Here $\noiseScale\in[0,\infty)$ is the bounded rationality level.
When the bounded rationality level $\noiseScale$ equals zero, the receiver takes each action uniformly at random
regardless of her posterior belief.
When the bounded rationality level $\noiseScale$ equals infinite, 
our model recovers the classic Bayesian persuasion for 
a (fully) rational receiver who takes the action which maximizes her expected utility.

Let function\footnote{We note that many results in 
\Cref{sec:state independent}, \Cref{sec:state dependent} hold
for general function $\Quant$. 
See \Cref{sec:general W curve} for detailed discussions.} $\Quant\betaed(x) \triangleq 
\sfrac{1}{(1+\exp(\noiseScale x))}$.
When the bounded rationality level
$\noiseScale$ is clear from the context,
we simplify  $\Quant\betaed$
with $\Quant$.
Given any posterior belief $\posterior$, 
we have $\receiverU(0\mid\posterior) - \receiverU(1\mid\posterior)
= \sum_{i\in[m]}\posterior_i\,
    \receiverU_i$, 
where 
$\receiverU_i\triangleq \receiverU_i(0) - 
\receiverU_i(1)$
represents how much the 
receiver prefers action $0$ over action $1$ given state $i$.
Without loss of generality, we assume $\{\receiverU_i\}$
is strictly increasing in $i$.
With the above definitions, we can rewrite
the probability that the receiver takes action 1 
as 
$
    \Quant\left
    (
    \sum_{i\in[m]}\posterior_i\,
    \receiverU_i\right)$.
Intuitively speaking, since the probability that receiver
takes action 1 only depends on the expected utility difference $\sum_{i\in[m]}\posterior_i\,
    \receiverU_i$,
it is without loss of generality to restrict to signaling scheme 
where each signal $\scaledNoiseDiff$ represents its 
induced expected utility difference, i.e., 
$\scaledNoiseDiff\equiv \sum_{i\in[m]}\posterior_i(\scaledNoiseDiff)\,
    \receiverU_i$.
\cjr{Recall that $\posterior_i(\scaledNoiseDiff)$ is the posterior
probability of state $i$ given realized signal $\scaledNoiseDiff$}.
We formalize this idea by
writing
our problem  
as the following linear program~\ref{eq:opt lp}
(and its dual program~\ref{eq:opt lp dual}) 
with 
variables $\{\signalscheme_i(\scaledNoiseDiff)\}_{\scaledNoiseDiff \in \R, i\in[m]}$.\footnote{While 
we allow signaling schemes to have continuous signal space
$\signalspace$, i.e., $\signalscheme_i(\cdot)$ can be interpreted as 
a probability density function over $\signalspace$,
all signaling schemes (as well as the optimal signaling schemes)
considered in this paper have finite signal space.
Therefore, we abuse the notation and 
use $\signalscheme_i(\cdot)$ as the probability mass function
when it is clear from the context.}
See \Cref{prop:opt lp}
and its proof in \Cref{apx:prelim proof}.
\begin{align}
    \label{eq:opt lp}
    \arraycolsep=5.4pt\def\arraystretch{1}
    \tag{$\mathcal{P}_\texttt{OPT-Primal}$}
    &\begin{array}{llll}
     \max\limits_{\boldsymbol\signalscheme\geq \zerobf} ~ &
     \displaystyle\sum\nolimits_{i\in[m]} \prior_i \senderU_i 
     \displaystyle\int_{-\infty}^{\infty} \signalscheme_i(\scaledNoiseDiff)
     \Quant(\scaledNoiseDiff)\,d\scaledNoiseDiff 
     \quad& \text{s.t.} &
     \vspace{1mm}
     \\
       & 
       \displaystyle\sum\nolimits_{i\in[m]}  \prior_i 
      \left(\receiverU_i - \scaledNoiseDiff \right)
       \signalscheme_i(\scaledNoiseDiff)
       = 0
       & \scaledNoiseDiff\in(-\infty,\infty) 
       \quad&  
       \langle\noiseDualVar(\scaledNoiseDiff)\rangle
     \vspace{1mm}
       \\
       &
       \displaystyle\int_{-\infty}^{\infty} \signalscheme_i(\scaledNoiseDiff)d\scaledNoiseDiff = 1
       &
       i\in[m] &
       \langle\distDualVar(i)\rangle
     \vspace{1mm}
     \\
    \end{array}
    \vspace{2mm}
    \\
    \label{eq:opt lp dual}
    \tag{$\mathcal{P}_\texttt{OPT-Dual}$}
    &\begin{array}{llll}
     \min\limits_{\boldsymbol\noiseDualVar,
     \boldsymbol\distDualVar} ~ &
     \displaystyle\sum\nolimits_{i\in[m]} \distDualVar(i)
     & \text{s.t.} &
     \vspace{1mm}
     \\
       & 
      \prior_i 
      \left(\receiverU_i - \scaledNoiseDiff \right)
      \noiseDualVar(\scaledNoiseDiff) 
      +
      \distDualVar(i)
      \geq 
      \prior_i \senderU_i \Quant(\scaledNoiseDiff)
      &
      \scaledNoiseDiff\in(-\infty,\infty),
      i\in[m] 
      \quad
      &
      \langle
      \signalscheme_i(\scaledNoiseDiff)
      \rangle
     \\
    \end{array}
\end{align}

\begin{restatable}{proposition}{optimallp}
\label{prop:opt lp}
For every feasible solution
$\{\signalscheme_i(\scaledNoiseDiff)\}$
in program~\ref{eq:opt lp},
there exists a signaling scheme
where for each state $i\in[m]$,
the boundedly rational receiver 
takes action 1 with probability 
$\int_{-\infty}^{\infty}
\signalscheme_i(\scaledNoiseDiff)\Quant(\scaledNoiseDiff)\,d\scaledNoiseDiff$.
Furthermore, the sender's optimal expected utility 
(in the optimal signaling scheme)
is equal to the optimal objective value of program~\ref{eq:opt lp}.
\end{restatable}

The first constraint in the program 
\ref{eq:opt lp} ensures that whenever a signal $\scaledNoiseDiff$
is realized, the probability for the receiver 
for taking action $1$ is exactly $\Quant(\scaledNoiseDiff)$.
Due to \Cref{prop:opt lp},
in the remaining of the paper, 
we  
describe 
signaling schemes by $\{\signalscheme_i(\scaledNoiseDiff)\}$
as the feasible solutions 
in program~\ref{eq:opt lp},
{
and $\{\signalscheme_i(\sigma)\}$
as the original definition interchangeably.}
We use $\Payoff[\noiseScale]{\signalscheme}$
to denote the expected sender utility 
of signaling scheme $\signalscheme$
(i.e., the objective value for feasible solution $\signalscheme$
in program~\ref{eq:opt lp})
for a receiver with bounded rationality level $\noiseScale$.
We drop subscript $\noiseScale$ in $\Payoff[\noiseScale]{\cdot}$
when it is clear from the context.

Our persuasion problem for the
boundedly rationally 
receiver is equivalent 
to a public persuasion problem 
for a continuum population of rational receivers
with a specific utility structure,
and thus program~\ref{eq:opt lp}
can be reinterpreted as the program 
for this public persuasion problem. 
See \Cref{sec:reformulate}
for more details.

\subsection{Optimal Signaling Schemes for A Fully Rational Receiver}
\label{sec:simple mechanisms}

Here we introduce two subclasses of signaling schemes
--
\emph{censorship signaling schemes}
and \emph{direct signaling schemes}, 
that will be discussed throughout this paper.
Briefly speaking, 
a censorship (resp.\ direct) signaling scheme
partitions the state space $[m]$ into 
three disjoint subsets:\footnote{Namely, $\highstates\cup\{\censorshipstate\}\cup\lowstates = [m]$,
$\highstates\cap \lowstates = \emptyset$,
$\censorshipstate\not\in \highstates$,
and $\censorshipstate\not\in\lowstates$.}
high states $\highstates$,
threshold state $\{\censorshipstate\}$, 
and low states $\lowstates$,
and 
specifies
a threshold state probability $\censorshipprob\in[0,1]$.
It pools
all states in $\highstates$ as well as 
a $(\censorshipprob)$-fraction of the threshold state $\censorshipstate$
into a pooling signal~$\censorshipsignal$,
and fully reveals other states
(resp.\ pools all other states into another pooling signal $\scaledNoiseDiff\doubleprimed$).
See \Cref{def:censorship} and \Cref{def:direct scheme}
for the formal definitions.\footnote{Our 
definition 
is equivalent to 
censorship signaling schemes
for persuasion problem with
continuous state space 
\citep[see, e.g.,][]{DM-19}
by considering the quantile space 
of state space $[m]$.}

\begin{definition}
\label{def:censorship}
A \emph{censorship signaling scheme},
parameterized by a
state space partition $\highstates\sqcup\{\censorshipstate\}\sqcup\lowstates$,
and 
threshold state probability $\censorshipprob$
admits the form as follows
\begin{align*}
    i\in\highstates:&\qquad
    \signalscheme_i(\scaledNoiseDiff) =
    \indicator{\scaledNoiseDiff = \censorshipsignal}
    \\
    &\qquad
    \signalscheme_{\censorshipstate}(\scaledNoiseDiff)
    =
    \censorshipprob
    \cdot 
    \indicator{\scaledNoiseDiff=\censorshipsignal}
    +
    (1-\censorshipprob)
    \cdot 
    \indicator{\scaledNoiseDiff=\receiverU_{\censorshipstate}}
    \\
    i\in\lowstates:&\qquad
    \signalscheme_i(\scaledNoiseDiff) =
    \indicator{\scaledNoiseDiff = \receiverU_i}
\end{align*}
where $\censorshipsignal=
\frac{
\censorshipprob\prior_{\censorshipstate} 
\receiverU_{\censorshipstate}
+\sum_{i\in\highstates}
\prior_i \receiverU_i
}
{
\censorshipprob\prior_{\censorshipstate} 
+ \sum_{i\in\highstates}
\prior_i
}$ is the pooling signal.
\end{definition}

\begin{definition}
\label{def:direct scheme}
A \emph{direct signaling scheme},
parameterized by a
state space partition $\highstates\sqcup\{\censorshipstate\}\sqcup\lowstates$,
and 
threshold state probability $\censorshipprob$
admits the form as follows
\begin{align*}
    i\in\highstates:&\qquad
    \signalscheme_i(\scaledNoiseDiff) =
    \indicator{\scaledNoiseDiff = \censorshipsignal}
    \\
    &\qquad
    \signalscheme_{\censorshipstate}(\scaledNoiseDiff)
    =
    \censorshipprob
    \cdot 
    \indicator{\scaledNoiseDiff=\censorshipsignal}
    +
    (1-\censorshipprob)
    \cdot 
    \indicator{\scaledNoiseDiff=\scaledNoiseDiff\doubleprimed}
    \\
    i\in\lowstates:&\qquad
    \signalscheme_i(\scaledNoiseDiff) =
    \indicator{\scaledNoiseDiff = \scaledNoiseDiff\doubleprimed}
\end{align*}
where $\censorshipsignal=
\frac{
\censorshipprob\prior_{\censorshipstate} 
\receiverU_{\censorshipstate}
+\sum_{i\in\highstates}
\prior_i \receiverU_i
}
{
\censorshipprob\prior_{\censorshipstate} 
+ \sum_{i\in\highstates}
\prior_i
}$ 
and $\scaledNoiseDiff\doubleprimed=
\frac{
(1-\censorshipprob)\prior_{\censorshipstate} 
\receiverU_{\censorshipstate}
+\sum_{i\in\lowstates}
\prior_i \receiverU_i
}
{
(1-\censorshipprob)\prior_{\censorshipstate} 
+ \sum_{i\in\lowstates}
\prior_i
}$ 
are the pooling signals.
\end{definition}
The only difference between censorship signaling schemes 
and direct signaling schemes is the signaling structure
for the $(1-\censorshipprob)$-fraction of the threshold state $\censorshipstate$ and every state in $\lowstates$ --
censorship signaling schemes fully reveal them,
while direct signaling schemes pools them all together. 
As a sanity check, note that 
censorship (resp.\ direct) signaling schemes are indeed the feasible solutions
of program~\ref{eq:opt lp}.
We also highlight two standard 
censorship signaling schemes:
the \emph{full-information revealing signaling scheme}
which reveals all states separately,
and the \emph{no-information revealing signaling scheme}
which pools all states at a single signal.

When the receiver is fully rational
(i.e.,
the bounded rationality level $\noiseScale = \infty$), there exists a censorship (resp.\ direct) signaling scheme that is indeed optimal.
\begin{restatable}{lemma}{optfullyrational}
[See for example \citealp{RSV-17}]
\label{lem:opt fully rational}
For a fully rational receiver 
(i.e., 
with bounded rationality level $\noiseScale = \infty$),
\cjr{it is optimal for the sender to adopt a censorship
(resp.\ direct) signaling scheme such that}
\begin{itemize}
    
    \item[(i)] threshold state $\censorshipstate
    =
    \argmax\limits_{i\in[m]} \left\{\frac{\receiverU_i}{\senderU_i}:
    \sum_{j: \frac{\receiverU_j}{\senderU_j}< \frac{\receiverU_i}{\senderU_i}} 
    \prior_j\receiverU_j \leq 0
    \right\}$;
    
    \item[(ii)] high states $\highstates = \left\{i\in[m]:
    \frac{\receiverU_i}{\senderU_i} < 
    \frac{\receiverU_{\censorshipstate}}{
    \senderU_{\censorshipstate}}\right\}$,
    and low states $\lowstates = \left\{i\in[m]:
    \frac{\receiverU_i}{\senderU_i} >
    \frac{\receiverU_{\censorshipstate}}{
    \senderU_{\censorshipstate}}\right\}$;
    
    \item[(iii)] threshold state probability 
    $\censorshipprob = \max\left\{p\in[0, 1]:
    p\prior_{\censorshipstate}\receiverU_{\censorshipstate}
    +
    \sum_{i\in\highstates}
    \prior_i\receiverU_i \leq 0
    \right\}$.
\end{itemize}
\end{restatable}

In fact, to achieve the optimality 
for a fully rational receiver, 
it only requires that all 
 states in $\highstates$
together with $(\censorshipprob)$-fraction of threshold state $\censorshipstate$
are pooled
into signal $\censorshipsignal$
where the assignments of $\highstates,\censorshipstate,\censorshipprob$
and $\censorshipsignal$
are defined in \Cref{lem:opt fully rational}
\citep{RSV-17}.
In other words, no restrictions on 
the signaling structure on the remaining $(1 - \censorshipprob)$-fraction of threshold state $\censorshipstate$
and other states in $\lowstates$
are required.
In this sense, the optimal censorship and 
optimal direct signaling scheme can be thought as two extreme
cases on the signaling structure for those states,
i.e., fully revealing them or pooling them all together,
and both achieve the optimality over all signaling schemes.
Therefore, for a fully rational receiver, 
the sender is indifferent between the optimal censorship 
and the optimal
direct signaling schemes.
However,
as we shown in the later sections,
there exists a 
separation between 
these two types of signaling schemes
when the receiver is boundedly rational.

%% file: Paper/state-independent.tex
\newcommand{\reals}{\R}
\newcommand{\tangentmapping}{\Psi}
\newcommand{\probmapping}{\Phi}
\newcommand{\positivestateindex}{i\doubleprimed}

\newcommand{\cs}{{\small \texttt{complementary-slackness}}}
\newcommand{\dfone}{{\small \texttt{dual-feasibility-1}}}
\newcommand{\dftwo}{{\small \texttt{dual-feasibility-2}}}
\newcommand{\dfthree}{{\small \texttt{dual-feasibility-3}}}
\newcommand{\pf}{{\small \texttt{primal-feasibility}}}

In this section, we 
consider the 
\emph{state independent sender utility (SISU)} environments
where the sender's utility $\{\senderU_i\}_{i\in[m]}$
is independent of the realized state.
Namely, we assume $\senderU_i \equiv 1$
for every state $i\in[m]$.
Furthermore, for ease of presentation, this section assumes
$\receiverU_1 < 0$ and $\receiverU_m > 0$.\footnote{
If $\receiverU_i \leq 0, \forall i\in[m]$,
we can introduce one dummy state $m + 1$ such that 
$\receiverU_{m+1} = 1$ and $\prior_{m + 1} = 0$.
Similarly, we can add one dummy state if 
 $\receiverU_i \geq 0, \forall i\in[m]$.
Thus, $\receiverU_1 < 0$ and 
$\receiverU_m > 0$ is without loss of generality.}

Recall that for a fully rational receiver, 
\Cref{lem:opt fully rational} shows the optimality 
of both censorship signaling schemes and direct signaling schemes.
However, when the receiver is boundedly rational, 
there exists a separation between these two subclasses of signaling schemes.
As the main result of this section,
\Cref{thm:opt state independent}
in \Cref{sec:SISU censorship}
shows that in SISU environments,
for a boundedly rational
receiver, 
it is optimal for the sender to adopt a censorship signaling scheme.
In contrast, \Cref{thm:SISU direct} in
\Cref{sec:SISU direct} shows that 
there exists a SISU problem instance,
where any direct signaling scheme is $\Omega(m)$-approximation. 


\subsection{Censorship as Optimal Signaling Schemes}
In this subsection,
we 
show that in SISU environments,
for a receiver with bounded rationality level $\noiseScale$, 
it is optimal for the sender to adopt a censorship signaling scheme.
Our result recovers 
the optimal censorship signaling scheme 
of a fully rational receiver 
(\Cref{lem:opt fully rational}).
In other words, the optimality of 
the censorship signaling schemes
is preserved
even when the receiver is boundedly rational
in SISU environments. 
\label{sec:SISU censorship}
\begin{theorem}
\label{thm:SISU opt}
\label{thm:opt state independent}
In SISU environments,
for a boundedly rational receiver 
with any bounded rationality level $\noiseScale$,
there exists a censorship signaling scheme $\optscheme$
that is the optimal signaling scheme.
Specifically,
\begin{enumerate}
    \item [(i)]
    the threshold state $\censorshipstate$
     and 
     the threshold state probability 
    $\censorshipprob$,
    together with an auxiliary 
    variable 
    $ 
    \censorshipsignalInverse$,
    solve the following 
    feasibility program~\ref{eq:SISU opt condition}:
    \begin{align}
        \label{eq:SISU opt condition}
        \tag{$\mathcal{P}_{\texttt{SISU-OPT}}$}
        \begin{split}
    (1 - \censorshipprob)
    (\censorshipsignalInverse - 
    \receiverU_{\censorshipstate}) &= 0
    \qquad\qquad(\text{complementary-slackness})
    \\
    \left({\Quant(\censorshipsignalInverse)  - 
    \Quant(\censorshipsignal)}\right)
    -
    \Quant'(\censorshipsignal) ({
    \censorshipsignalInverse
    - \censorshipsignal}) &= 0
    \qquad\qquad (\text{dual-feasibility-1})
    \\
\max_{i\in[m]}\{
    \receiverU_i:\receiverU_i \leq \censorshipsignalInverse\}
    &=
    \receiverU_{\censorshipstate}
    \qquad\,\,\,\,\,\,\,\,(\text{dual-feasibility-2})
    \\
            \censorshipsignal \leq 0,~
            \censorshipsignalInverse &\geq 0
           \qquad\qquad (\text{dual-feasibility-3})
           \\
           0\leq \censorshipprob&\leq 1
           \qquad\qquad(\text{primal-feasibility})
        \end{split}
    \end{align}
where $\censorshipsignal = 
    \frac{\sum_{i:i < \censorshipstate} \prior_i\receiverU_i + \censorshipprob \prior_{\censorshipstate}\receiverU_{\censorshipstate}}{\sum_{i:i < \censorshipstate} \prior_i + \censorshipprob \prior_{\censorshipstate}}$ is the pooling signal;
    \item [(ii)] high states $\highstates = 
    \left\{i\in[m]:i < \censorshipstate\right\}$, 
    and low states $\lowstates = \left\{i\in[m]: i > \censorshipstate\right\}$.
    \end{enumerate}
\end{theorem}

In below, we first provide intuitions behind the constraints 
in the feasibility program \ref{eq:SISU opt condition}, 
and
the properties as well as
the implications of the above characterized
optimal censorship signaling scheme. 
Then, we provide 
the high-level proof idea of \Cref{thm:SISU opt}.

\begin{figure}[ht]
    \centering
    \input{plots/figure-SISU-opt-condition}
    \caption{Graphical illustration for 
    constraints \dfone,
    \dftwo\ and \dfthree\ in
    feasibility program~\ref{eq:SISU opt condition}.
    The gray solid curve is function $\Quant(\cdot)$.
    Fix an arbitrary $\censorshipsignalInverse\in[0,\infty)$.
    Constraint \dfone\ uniquely pins down
    $\censorshipsignal\in(-\infty, 0]$ (and thus
    constraint \dfthree\ is satisfied as well) 
    such that the black solid line through
    point $(\censorshipsignal,\Quant(\censorshipsignal))$
    and point 
    $(\censorshipsignalInverse,\Quant(\censorshipsignalInverse))$
    is tangent to curve $\Quant(\cdot)$
    at point $(\censorshipsignal,\Quant(\censorshipsignal))$.
    Constraint \dftwo\ uniquely pins down
    $\censorshipstate \triangleq 
    \argmax_{i}\{\receiverU_i:
    \receiverU_i \leq \censorshipsignalInverse\}$.
    The tangent line and the curve $\Quant(\scaledNoiseDiff), \forall \scaledNoiseDiff \ge \censorshipsignalInverse$
    forms a upper convex envelop for function $\Quant$.
    \label{fig:SISU opt condition}
    \vspace{-6pt}
    }
\end{figure}
\xhdr{Graphical interpretation 
of optimal signaling scheme}
To 
develop intuition for optimal censorship,
we start with
constraints \dfone\
and \dfthree.
Recall that $\Quant(x) = \sfrac{1}
{(1+\exp(\noiseScale x))}$ is 
concave in $(-\infty, 0]$
and convex in $[0, \infty)$.
Constraint \dfone\ 
has the following graphical interpretation:
the line through
point $(\censorshipsignal,\Quant(\censorshipsignal))$
and point $(\censorshipsignalInverse,
\Quant(\censorshipsignalInverse))$
is tangent to curve $\Quant(\cdot)$
at point $(\censorshipsignal,\Quant(\censorshipsignal))$.
Notably, 
for every $\censorshipsignalInverse \geq 0$,
there exists a unique $\censorshipsignal\leq 0$
which satisfies \dfone.
In particular, the mapping from 
$\censorshipsignalInverse \in[ 0,\infty)$
to $\censorshipsignal\in
(-\infty, 0]$ satisfying \dfone\
is monotone decreasing and is a bijection
(See \Cref{fig:SISU opt condition}
for illustration).
Constraint~\dftwo\ 
means that threshold state $\censorshipstate$  
is the largest state index such that $\receiverU_i \leq \censorshipsignalInverse$,
i.e., $\censorshipstate = \argmax_{i}\{
\receiverU_i:\receiverU_i\leq \censorshipsignalInverse\}$.\footnote{Recall 
we assume $\receiverU_1 < 0$ without loss of generality,
and thus $\censorshipstate = \argmax_{i}\{
\receiverU_i:\receiverU_i\leq \censorshipsignalInverse\}$
is well-defined for $\censorshipsignalInverse \geq 0$.}
Hence, 
starting with an arbitrary $\censorshipsignalInverse \geq 0$,
constraints \dfone\ and \dftwo\ 
pin down a unique tuple $(\censorshipsignal,\censorshipsignalInverse,
\censorshipstate,\censorshipprob)$:
constraint \dfone\ pins down
a unique $\censorshipsignal \leq 0$,
constraint \dftwo\ pins down
a unique $\censorshipstate$,
and then $\censorshipprob$
is uniquely determined as well by the 
relation between $\censorshipsignal, \censorshipstate,
\censorshipprob$.

Essentially, the tangent line segment
from point $(\censorshipsignal,\Quant(\censorshipsignal))$
to point $(\censorshipsignalInverse,
\Quant(\censorshipsignalInverse))$
and the part of curve $\Quant(\scaledNoiseDiff), 
\forall \scaledNoiseDiff \ge \censorshipsignalInverse$
form a {\em upper convex envelop}
for function $\Quant(\cdot)$. 
It is easy to see that there exist infinitely many 
such upper convex envelops for function $\Quant(\cdot)$.
However, the optimal censorship is the unique one 
that the corresponding envelop ensures  \pf\ 
and satisfies \cs\ 
for the assignment on
the threshold state $\censorshipstate$.

\newcommand{\optschemeHat}{\hat{\signalscheme}^*}
\newcommand{\censorshipprobHat}{\hat{p}\primed}
\newcommand{\censorshipsignalHat}{\hat{\scaledNoiseDiff}\primed}
\newcommand{\censorshipstateHat}{\hat{i}\primed}

\xhdr{Less rational, more information revealing}
In SISU environments,
for both fully rational receiver
and boundedly rational receiver,
it is optimal for the sender to adopt
censorship signaling schemes
(\Cref{lem:opt fully rational}, 
\Cref{thm:SISU opt}).\footnote{Recall we assume $\{\receiverU_i\}$ is weakly increasing and $\senderU_i \equiv 1$ for every state $i\in[m]$. 
    Thus, $i < \censorshipstate$ 
    if and only if $\sfrac{\receiverU_i}{\senderU_i} \leq 
    \sfrac{\receiverU_{\censorshipstate}}{\senderU_{\censorshipstate}}$ and the construction of optimal censorship in \Cref{thm:SISU opt} recovers the construction in \Cref{lem:opt fully rational} for a fully rational receiver in the SISU environments.}
    \footnote{
    \ifEC 
    The feasibility program~\ref{eq:SISU opt condition}
    also recovers the structure of the optimal censorship
    for a fully rational receiver
    in SISU environments. See \Cref{apx:feasibility for fully rational} for detailed discussions.
    \fi}
However, 
in the optimal censorship for different 
rationality levels, 
the threshold state $\censorshipstate$
and the threshold state probability $\censorshipprob$
may not be the same.
For example,
consider an instance where 
$\sum_{i\in[m]}\prior_i\receiverU_i < 0$.
\Cref{lem:opt fully rational} suggests that 
the optimal censorship $\optschemeHat$
for a fully rational
receiver
selects the threshold state $\censorshipstateHat$
and threshold state probability $\censorshipprobHat$
such that the pooling signal $\censorshipsignalHat = 0$.
In contrast, \Cref{thm:SISU opt}
suggests that 
the optimal censorship $\optscheme$
selects the threshold state $\censorshipstate$
and threshold state probability $\censorshipprob$
such that the pooling signal $\censorshipsignal \leq 0$.
Thus, in this instance, 
the optimal censorship $\optschemeHat$ 
for a fully rational receiver 
pools more states than 
the optimal censorship $\optscheme$
for a boundedly rational receiver,
i.e., $\hat{\highstates}\supseteq \highstates$.
Here we generalize this observation and show
the monotonicity 
of threshold state $\censorshipstate$
and threshold state probability $\censorshipprob$
with respect to the rationality level $\noiseScale$.
Its proof is based on the analysis for 
the feasibility program~\ref{eq:SISU opt condition},
which we defer to \Cref{apx:SISU proof SISU threshold state monotonicity}.

\begin{restatable}{proposition}{sisustatemonotone}
\label{prop:SISU threshold state monotonicity}
In SISU environments, 
let $\optscheme$ (resp.\ $\optschemeHat$) be the optimal censorship
for a boundedly rational receiver with 
boundedly rational level $\noiseScale$
(resp.\ $\hat\noiseScale$).
If $\noiseScale \leq \hat\noiseScale$,
then
the threshold state $\censorshipstate$
in $\optscheme$
is weakly smaller than the threshold state $\censorshipstateHat$
in $\optschemeHat$,
i.e., $\censorshipstate \leq \censorshipstateHat$;
and 
threshold state probability $\censorshipprob \leq \censorshipprobHat$.
\end{restatable}
One concrete insight behind the above result is that:
The optimal (censorship) signaling scheme requires 
the sender to reveal more information 
(i.e., Blackwell ordering, \citealp{bla-53}) for a less rational receiver. 
This insight 
can be also developed from the curvature of the function $W(\cdot)$. 
When the receiver is less rational, i.e., the rationality 
level $\noiseScale$ becomes smaller, the curve $W(\cdot)$ 
becomes flatter. 
Hence, the tangent point $(\censorshipsignal,
W(\censorshipsignal))$ is farther away from the point
$(\censorshipsignalInverse,
W(\censorshipsignalInverse))$, 
and the pooling probability $\censorshipprob$
has to be smaller to make the pooling 
signal $\censorshipsignal$ relatively smaller. 
Thus, the threshold state and the threshold state probability
must decrease in order to satisfy the feasibility 
constraint in the program ~\ref{eq:SISU opt condition}, 
which leads to more information revealing.

\xhdr{Proof overview of \Cref{thm:SISU opt}}
Now we first provide a proof overview for \Cref{thm:SISU opt}, and in the sequel, we present the detailed proof.
At the heart of proof of \Cref{thm:SISU opt},
we use the strong duality between
the primal program~\ref{eq:opt lp}
and its dual program~\ref{eq:opt lp dual}.
Specifically, given 
a feasible solution in 
feasibility program~\ref{eq:SISU opt condition},
we explicitly construct a feasible primal assignment
in \ref{eq:opt lp}
and a feasible dual assignment 
in \ref{eq:opt lp dual}
and show the complementary slackness holds.\footnote{Here 
the name of each constraint in \ref{eq:SISU opt condition}
indicates its usage in the assignment construction.}
In the formal proof,
for each possible tuple 
$(\censorshipsignal,\censorshipsignalInverse,
\censorshipstate,\censorshipprob)$
described in the graphical interpretation 
(i.e., satisfying constraints
\dfone, \dftwo, \dfthree),
we can construct a feasible assignment 
for dual program~\ref{eq:opt lp dual}.
Notably, each feasible solution to the
dual program~\ref{eq:opt lp dual} forms a upper
convex envelop for the function $W(\cdot)$.
To finish the proof with strong duality, 
we require such tuple to additionally 
satisfy constraint \pf\
to ensure the feasibility of 
the constructed primal assignment,
and constraint \cs\ 
to ensure the complementary slackness 
of the constructed assignment on
the threshold state $\censorshipstate$.
The existence and uniqueness of such tuple 
is shown in \Cref{lem:existence SISU opt condition},
its proof is in \Cref{apx:SISU proof opt condition existence}.

\input{Paper/proof-thm-3-1.tex}

\subsection{Approximation Lower Bounds of Direct Signaling Schemes}
\label{sec:SISU direct}
When the receiver is fully rational, 
the optimality of direct signaling schemes follows
from the standard revelation principle \citep{KG-11}. 
However, when the receiver is boundedly rational,
this standard revelation principle fails. 
In this subsection, we
provide an approximation lower bound 
for direct signaling schemes in SISU environment
for a boundedly rational receiver. 
The proof of \Cref{thm:SISU direct} is straightforward and thus we defer it to 
\Cref{apx:SISU direct proof}. 
\begin{restatable}{theorem}{sisudirectopt}
\label{thm:SISU direct}
In SISU environments, 
there exists a problem instance (\Cref{example:SISU direct lower bound example})
such that for any direct signaling scheme $\signalscheme$,
it is $\Omega(m)$-approximation
to the optimal signaling scheme.
\end{restatable}
In \Cref{coro:SDSU censorship m approx}, we also give an $O(m)$-approximation upper bound for direct signaling schemes, which shows the tightness of our result.

%% file: plots/figure-SISU-opt-condition.tex
\begin{tikzpicture}[scale=0.7, transform shape]
\begin{axis}[
axis line style=gray,
axis lines=middle,
xtick style={draw=none},
ytick style={draw=none},
xticklabels=\empty,
yticklabels=\empty,
xmin=-9,xmax=12,ymin=-0.15,ymax=1.3,
width=0.9\textwidth,
height=0.5\textwidth,
samples=50]

\addplot[domain=-9:12, gray!40!white, line width=2.mm] (x, {1/(1+e^(0.7*x))});

\addplot[mark=*,only marks, fill=white] coordinates {(6.34907,0.0116084)} node[above, pos=1]{};
\addplot[gray, thick] coordinates {(6.34907,0.01)
(6.34907,-0.01)};
\addplot[] coordinates {(6.34907,0.)} node[below, pos=1]{\Large$\censorshipsignalInverse$};

\addplot[mark=*,only marks, fill=white] coordinates {(-2.34798,0.838022)} node[above, pos=1]{};
\addplot[gray, thick] coordinates {(-2.34798,0.01)
(-2.34798,-0.01)};
\addplot[] coordinates {(-2.34798,0.)} node[below, pos=1]{\Large$\censorshipsignal$};

\addplot[gray, dotted] coordinates {(-2.34798,0.838022)
(-2.34798,0.)};

\addplot[gray, thick] coordinates {(5,0.01)
(5,-0.01)};
\addplot[] coordinates {(5,-0.035)} node[below, pos=1]{\Large$\receiverU_{\censorshipstate}$};


\addplot[gray, thick] coordinates {(11,0.01)
(11,-0.01)};
\addplot[] coordinates {(11,-0.035)} node[below, pos=1]{\Large$\receiverU_{\censorshipstate+1}$};

\addplot[domain=-6:6.34907, black,line width=0.5mm] (x, {-0.09502 * x + 0.614897});
\end{axis}

\end{tikzpicture}

%% file: Paper/proof-thm-3-1.tex
\begin{restatable}{lemma}{sisuoptcondition}
\label{lem:existence SISU opt condition}
There exists a unique feasible solution 
in program~\ref{eq:SISU opt condition}.
\end{restatable}


\begin{proof}[Proof of \Cref{thm:SISU opt}]\renewcommand{\qedsymbol}{}
We prove the optimality of the signaling scheme $\optscheme$ defined in \Cref{thm:SISU opt}
by constructing a feasible dual solution to 
the dual program \ref{eq:opt lp dual} 
that satisfies the complementary slackness. 
Let $(\censorshipsignal,\censorshipsignalInverse,
\censorshipstate,\censorshipprob)$
be the unique feasible solution to program~\ref{eq:SISU opt condition}.

\paragraph{Assignment construction.}
To facilitate the analysis, we explicitly write out the 
optimal signaling scheme $\optscheme$ as follows,
\begin{align*}
    i\in[\censorshipstate - 1]: 
     & \qquad\optscheme_i(\censorshipsignal) = 1; 
    \\ 
    & \qquad\optscheme_{\censorshipstate}(\censorshipsignal) 
    = \censorshipprob, 
    \qquad
    \optscheme_{\censorshipstate}\left(\receiverU_{\censorshipstate}\right) 
    = 1 - \censorshipprob; 
    \\ 
    i\in[\censorshipstate + 1:m]: 
     & \qquad\optscheme_i\left(\receiverU_i\right) = 1~
\intertext{
Due to constraint \pf\ in program~\ref{eq:SISU opt condition},
signaling scheme $\optscheme$ is feasible. 
Now, consider the following dual assignment 
$\{\noiseDualVar(\scaledNoiseDiff), \distDualVar(i)\}$
of 
program~\ref{eq:opt lp dual},}
    \scaledNoiseDiff\in(\infty,\censorshipsignalInverse]:
    &\qquad
    \noiseDualVar(\scaledNoiseDiff) = 
    -\Quant'(\censorshipsignal)
    \\
    \scaledNoiseDiff\in
    [\censorshipsignalInverse,
    \receiverU_{\censorshipstate + 1}]:
    &\qquad
    \noiseDualVar(\scaledNoiseDiff) = 
    -\frac{
    \Quant(\scaledNoiseDiff) - 
    \Quant(\receiverU_{\censorshipstate + 1})
    }{
    \scaledNoiseDiff - 
    \receiverU_{\censorshipstate + 1}
    }
    \\
    i\in[\censorshipstate+1: m - 1],~
    \scaledNoiseDiff\in
    [\receiverU_{i}, \receiverU_{i+1}]:
    &\qquad
    \noiseDualVar(\scaledNoiseDiff) = 
    -\frac{
    \Quant(\scaledNoiseDiff) - 
    \Quant(\receiverU_{i})
    }{
    \scaledNoiseDiff - 
    \receiverU_{i}
    }
    \\
    \scaledNoiseDiff\in
    [
    \receiverU_m, \infty):
    &\qquad
    \noiseDualVar(\scaledNoiseDiff) = 
    0
    \\
    i\in[\censorshipstate]:
    &\qquad 
    \distDualVar(i) = 
    \prior_i\left(
    \Quant(\censorshipsignal) + 
    (\censorshipsignal - \receiverU_i)
    \noiseDualVar(\censorshipsignal)
    \right)
    \\
    i\in[\censorshipstate+1:m]:
    &\qquad 
    \distDualVar(i) = 
    \prior_i\Quant(\receiverU_i)
\end{align*}

\paragraph{Complementary slackness.} 
We now argue the complementary slackness 
of the constructed assignment.
Namely, for each state $i$
and $\scaledNoiseDiff\in(-\infty,\infty)$
such that $\optscheme_i(\scaledNoiseDiff) > 0$,
its corresponding dual constraint holds with equality,
i.e., 
\begin{align}
      \label{eq:SISU cs}
      \Quant(\scaledNoiseDiff) + 
      ( \scaledNoiseDiff - \receiverU_i)
      \noiseDualVar(\scaledNoiseDiff) &= 
      \frac{\distDualVar(i)}{\prior_i}
\end{align}
We verify this for each state $i\in[m]$ separately.
\begin{itemize}[label={-}]
    \item Fix an arbitrary
    state $i\in[\censorshipstate - 1]$,
    note that $\optscheme_i(\scaledNoiseDiff) > 0$
    for $\scaledNoiseDiff = \censorshipsignal$ only.
    Here equality~\eqref{eq:SISU cs} holds 
    by construction straightforwardly.
    \item Now consider threshold state $\censorshipstate$, 
    the same argument holds for 
    equality~\eqref{eq:SISU cs}
    with $\optscheme_{\censorshipstate}
    (\censorshipsignal) > 0$.
    It is remaining to verify 
    equality~\eqref{eq:SISU cs} 
    associated with
    $\optscheme_{\censorshipstate}
    (\receiverU_{\censorshipstate})
    = 1 - \censorshipprob > 0$.
    When $\censorshipprob < 1$, 
    constraint~\cs\ in program~\ref{eq:SISU opt condition}
    ensures that $\censorshipsignalInverse =
    \receiverU_{\censorshipstate}$.
    Thus,
    \begin{align*}
      \Quant(\receiverU_{\censorshipstate}) + 
      ( \receiverU_{\censorshipstate} -
      \receiverU_{\censorshipstate})
      \noiseDualVar(\receiverU_{\censorshipstate})
      \overset{(a)}{=}
      \Quant(\censorshipsignalInverse)
      \overset{(b)}{=}
      \Quant(\censorshipsignal) +
      (\censorshipsignal - \censorshipsignalInverse)
      \noiseDualVar(\censorshipsignal)
      \overset{(c)}{=}
      \Quant(\censorshipsignal) +
      (\censorshipsignal - \receiverU_{\censorshipstate})
      \noiseDualVar(\censorshipsignal)
      \overset{(d)}{=}
      \frac{\distDualVar(i)}{\prior_i}
    \end{align*}
where equalities~(a), (c) hold since $\receiverU_{\censorshipstate} = \censorshipsignalInverse$;
equality~(b) holds 
due to constraint \dfone\ in program~\ref{eq:SISU opt condition}
and the construction of $\noiseDualVar(\censorshipsignal)$;
and 
equality~(d) holds since 
equality~\eqref{eq:SISU cs} holds 
for $\optscheme_{\censorshipstate}(\censorshipsignal)$
shown above.
\item Fix an arbitrary state $i\in[\censorshipstate + 1:m]$,
note that $\optscheme_i(\scaledNoiseDiff) > 0$
for $\scaledNoiseDiff = \receiverU_i$ only.
Here equality~\eqref{eq:SISU cs} holds 
by construction straightforwardly.
\end{itemize}

\paragraph{Dual feasibility.}
To verify whether the dual constraints
associated with $\optscheme_i(\scaledNoiseDiff)$
for state $i \in [\censorshipstate]$ hold,
note that 
\begin{align*}
    \frac{\distDualVar(i)}{\prior_i}
      \overset{(a)}{=}
      \Quant(\censorshipsignal) +
      (\censorshipsignal - \receiverU_i)
      \noiseDualVar(\censorshipsignal)
      \overset{(b)}{=}
      \Quant(\censorshipsignal) -
      (\censorshipsignal - \receiverU_i)
      \Quant'(\censorshipsignal)
\end{align*}
where equality~(a) holds by the complementary slackness
of $\optscheme_i(\censorshipsignal)$ verified above;
and equality~(b) holds by the construction of 
$\noiseDualVar(\censorshipsignal)$.
Thus, 
we can rewrite those dual constraints
associated with $\optscheme_i(\scaledNoiseDiff)$
for state $i \in [\censorshipstate]$ 
as
\begin{align}
    \label{eq:SISU dual feasibility pooling state}
    \Quant(\scaledNoiseDiff) + 
      ( \scaledNoiseDiff - \receiverU_i)
      \noiseDualVar(\scaledNoiseDiff) \leq 
      \Quant(\censorshipsignal) -
      (\censorshipsignal - \receiverU_i)
      \Quant'(\censorshipsignal) 
\end{align}
To verify whether the dual constraints
associated with $\optscheme_i(\scaledNoiseDiff)$
for state $i \in [\censorshipstate+1:m]$ hold,  
by the complementary slackness
of $\optscheme_i(\censorshipsignal)$ verified above,
we can rewrite the dual constraints for
state $i \in [\censorshipstate+1:m]$
as
\begin{align}
    \label{eq:SISU dual feasibility revealing state}
    \Quant(\scaledNoiseDiff) + 
      ( \scaledNoiseDiff - \receiverU_i)
      \noiseDualVar(\scaledNoiseDiff) \leq 
      \Quant(\receiverU_i)
\end{align}
We verify both inequality~\eqref{eq:SISU dual feasibility pooling state}
and inequality \eqref{eq:SISU dual feasibility revealing state}
for different values of $\scaledNoiseDiff$ 
in four cases separately: 
$\scaledNoiseDiff \in 
(-\infty,\censorshipsignalInverse]$;
$\scaledNoiseDiff\in 
[\censorshipsignalInverse, \receiverU_{\censorshipstate + 1}]$;
$\scaledNoiseDiff\in[\receiverU_j, \receiverU_{j + 1}]$
for some $j\in[\censorshipstate + 1: m - 1]$;
and $\scaledNoiseDiff\in[\receiverU_m, \infty)$.
The argument mainly uses the curvature of 
function $\Quant(\cdot)$ 
and the constraints in feasibility program~\ref{eq:SISU opt condition},
and the feasibility of the constructed dual assignment
is summarized as follows:
\begin{restatable}{lemma}{sisudualfeasibility}
\label{lem: sisu dual feasibility}
The constructed dual assignment is a feasible 
solution to the dual program \ref{eq:opt lp dual}.
\end{restatable}
In the main text below, 
we present the analysis for state $i\in[\censorshipstate]$
for the first two cases,
together with a graphical illustration
of our argument (see \Cref{fig:SISU dual feasibility}).
The analysis of the third case is similar to the second case,
and the fourth case is trivial.
Therefore, we defer the later two cases and 
the analysis for the state $i\in [\censorshipstate+1, m]$ 
to \Cref{apx:SISU proof SISU dual feasibility}.

\begin{itemize}[label={-}]
    \item Fix an arbitrary $\scaledNoiseDiff \in 
    (-\infty,\censorshipsignalInverse]$.
    We illustrate this case in \Cref{fig:case a}.
    Note that
    \begin{align*}
        \Quant(\scaledNoiseDiff) + 
      ( \scaledNoiseDiff - \receiverU_i)
      \noiseDualVar(\scaledNoiseDiff)
      &\overset{(a)}{=}
      \Quant(\scaledNoiseDiff) -
      (\scaledNoiseDiff - \receiverU_i)
      \Quant'(\censorshipsignal) 
      \\
      &=
       \Quant(\censorshipsignal) -
      (\censorshipsignal - \receiverU_i)
      \Quant'(\censorshipsignal) 
      +
      \left(\Quant(\scaledNoiseDiff) - 
      \Quant(\censorshipsignal)\right)
      -
      (\scaledNoiseDiff - \censorshipsignal)
      \Quant'(\censorshipsignal)
    \end{align*}
    where equality~(a) holds due to 
    the construction of $\noiseDualVar(\scaledNoiseDiff)$.
    Hence, to show inequality~\eqref{eq:SISU dual feasibility pooling state} in this case,
    it is sufficient to argue that 
    \begin{align}
    \label{eq:SISU dual feasibility pooling state case 1}
        (\censorshipsignal - \scaledNoiseDiff)
        \Quant'(\censorshipsignal) \leq 
        \Quant(\censorshipsignal) - 
        \Quant(\scaledNoiseDiff)
    \end{align}
    which is true due to the curvature
    of function $\Quant(\cdot)$.
    Specifically,
    if $\scaledNoiseDiff\in(-\infty,0]$,
    inequality~\eqref{eq:SISU dual feasibility pooling state case 1}
    holds since function $\Quant(\cdot)$
    is concave in $(-\infty, 0]$;
    if $\scaledNoiseDiff\in[0, \censorshipsignalInverse]$,
    inequality~\eqref{eq:SISU dual feasibility pooling state case 1}
    holds since
    \begin{align*}
        \Quant'(\censorshipsignal)
        \overset{(a)}{=}
        \frac{\Quant(\censorshipsignalInverse) 
        -
        \Quant(\censorshipsignal)}
        {\censorshipsignalInverse - \censorshipsignal}
        \overset{(b)}{\geq}
        \frac{\Quant(\scaledNoiseDiff) 
        -
        \Quant(\censorshipsignal)}
        {\scaledNoiseDiff - \censorshipsignal}
    \end{align*}
    where equality~(a) holds due to 
    constraint~\dfone\ in program~\ref{eq:SISU opt condition};
    and inequality~(b) holds due to 
    the convexity of function $\Quant(\cdot)$
    on $[0,\infty)$.
    \item Fix an arbitrary $\scaledNoiseDiff\in 
    [\censorshipsignalInverse, \receiverU_{\censorshipstate + 1}]$.
    We illustrate this case in \Cref{fig:case b}.
    By construction, $\noiseDualVar(\scaledNoiseDiff) 
    = -
    (\Quant(\scaledNoiseDiff) - 
    \Quant(\receiverU_{\censorshipstate + 1}))/(
    \scaledNoiseDiff - 
    \receiverU_{\censorshipstate + 1})$.
    After rearranging the terms,
    inequality~\eqref{eq:SISU dual feasibility pooling state}
    becomes\footnote{Here we use the fact 
    that $\scaledNoiseDiff - \receiverU_i \geq 0$
    for every state $i\in[\censorshipstate]$, 
    since
    $\scaledNoiseDiff \geq 
    \censorshipsignalInverse\geq 
    \receiverU_{\censorshipstate} \geq 
    \receiverU_i$ 
    where the second inequality holds 
    due to 
    constraint~{\small \dftwo}\ 
    in program~\ref{eq:SISU opt condition}.}
    \begin{align}
        \label{eq:SISU dual feasibility pooling state case 2}
        -\frac{
        \Quant(\scaledNoiseDiff) - \Quant(\receiverU_{\censorshipstate+1})}
        {\scaledNoiseDiff - 
        \receiverU_{\censorshipstate+1}}
        \leq 
        - \frac{ 
        \Quant(\scaledNoiseDiff) - 
        \left(\Quant(\censorshipsignal)
        -\Quant'(\censorshipsignal) 
        (\censorshipsignal - \receiverU_i) \right)}
        {\scaledNoiseDiff - \receiverU_i}
    \end{align}
    Here we argue that it is sufficient to show 
    inequality~\eqref{eq:SISU dual feasibility pooling state case 2} holds when we replace $\receiverU_i$
    with $\censorshipsignalInverse \geq \receiverU_i$.
    To see this, note that
    the right-hand side of inequality~\eqref{eq:SISU dual feasibility pooling state case 2}
    is monotone decreasing as a function of 
    $\receiverU_i$.
    In particular, consider function
    $f(x) \triangleq
    - ({ 
    \Quant(\scaledNoiseDiff) - 
    (\Quant(\censorshipsignal)
    -\Quant'(\censorshipsignal) 
    (\censorshipsignal - x) )})/(
    {\scaledNoiseDiff - x})$,
    and compute its derivative
    $
    f'(x) =
    -
    \frac{
    \Quant(\scaledNoiseDiff) - 
    \Quant(\censorshipsignal)
    -
    \Quant'(\censorshipsignal)
    (\scaledNoiseDiff - \censorshipsignal)
    }
    {(\scaledNoiseDiff - x)^2}
    \leq 0
    $
    where the last inequality holds since 
    $\Quant'(\censorshipsignal)
    (\scaledNoiseDiff -
    \censorshipsignal)
    \leq \Quant(\scaledNoiseDiff) - 
    \Quant(\censorshipsignal)$
    if
    $\scaledNoiseDiff \geq \censorshipsignalInverse$,
    which is implied by 
    constraint~\dfone\ 
    and the convexity of function $\Quant(\cdot)$
    on $[0,\infty)$.
    Hence,
    \begin{align*}
    - \frac{ 
    \Quant(\scaledNoiseDiff) - 
    \left(\Quant(\censorshipsignal)
    -\Quant'(\censorshipsignal) 
    (\censorshipsignal - \receiverU_i) \right)}
    {\scaledNoiseDiff - \receiverU_i}
    &\overset{(a)}{\geq} 
    f(\censorshipsignalInverse)
    \overset{}{=}
        - \frac{ 
    \Quant(\scaledNoiseDiff) - 
    \left(\Quant(\censorshipsignal)
    -\Quant'(\censorshipsignal) 
    (\censorshipsignal - \censorshipsignalInverse) \right)}
    {\scaledNoiseDiff - \censorshipsignalInverse}
    \\
    &\overset{(b)}{=}
     -\frac{
    \Quant(\scaledNoiseDiff) - \Quant(\censorshipsignalInverse)}
    {\scaledNoiseDiff - 
    \censorshipsignalInverse}
    \overset{(c)}{\geq }
    -\frac{
    \Quant(\scaledNoiseDiff) - \Quant(\receiverU_{\censorshipstate+1})}
    {\scaledNoiseDiff - 
    \receiverU_{\censorshipstate+1}}
    \end{align*}
    where inequality~(a) holds due to the 
    monotonicity of function $f(\cdot)$;
    equality~(b) holds 
    due to constraint~\dfone\ in program~\ref{eq:SISU opt condition};
    and inequality~(c)
    holds
    due to the convexity of function $\Quant(\cdot)$
    on $[0,\infty)$. \hfill \ensuremath{\Box}
\end{itemize}
\end{proof}

%% file: Paper/state-dependent.tex
\newcommand{\optsignalspace}{\signalspace^*}
\newcommand{\ked}{^{(k)}}
\newcommand{\instance}{\mathcal{I}}
\newcommand{\oneed}{^{(1)}}
\newcommand{\twoed}{^{(2)}}
\newcommand{\aggregateprob}{p_{ij}}
\newcommand{\aggregateproba}{\aggregateprob\oneed}
\newcommand{\aggregateprobb}{\aggregateprob\twoed}

\newcommand{\optschemeSISUnew}{\signalscheme^*}
\newcommand{\optsignalspaceSISUnew}{\widetilde{\signalspace}^*}
\newcommand{\optschemeSISUcensor}{\widetilde{\signalscheme}^*}

In this section,
we consider the 
\emph{state dependent sender utility (SDSU)} environments
where the sender's utility $\{\senderU_i\}_{i\in[m]}$
depends on both the realized state
as well as the action of the receiver.\footnote{Recall that 
the sender's utility is zero as long as 
the receiver takes action 0, and 
$\senderU_i\geq 0$ denotes the sender's utility 
for realized state $i$ and receiver taking action 1.}
Recall that or a fully rational receiver, \Cref{lem:opt fully rational} shows the optimality of both censorship signaling
schemes and direct signaling schemes. However, in SDSU environments,
both censorship signaling schemes 
and direct signaling schemes are sub-optimal
for a boundedly rational receiver.
As the main result of this section,
we first show that both censorship and direct signaling schemes
are 
$\Omega(m)$-approximation 
(\Cref{prop: SDSU censorship approx LB}), 
and then we provide matching 
approximation upper bounds of censorship and direct signaling schemes 
(\Cref{coro:SDSU censorship m approx}).




\subsection{Approximation Lower Bounds of Censorship and Direct Signaling Schemes}
\label{sec:SDSU simple approx lower bound}
\input{Paper/state-dependent-simple-lower-bound}

\subsection{Approximation Upper Bounds of Censorship and Direct Signaling Schemes}
\label{sec:SDSU simple approx}
\input{Paper/state-dependent-simple}

%% file: Paper/state-dependent-simple-lower-bound.tex
\newcommand{\decomposesignal}{\scaledNoiseDiff}
\newcommand{\lowerbounddecomposescheme}{\signalscheme^{(i, \decomposesignal)}}

\newcommand{\lastNOinforSignal}{\bar{\scaledNoiseDiff}^{\texttt{avg}}}
\newcommand{\lowerbounddecomposeschemeavg}{\signalscheme^{(i, \NOinforSignal)}}
\newcommand{\lowerbounddecomposeschemeavglast}{\signalscheme^{(m-1, \lastNOinforSignal)}}
\newcommand{\lowerbounddecomposeschemeell}{\signalscheme^{(i, \scaledNoiseDiff_\ell)}}

\newcommand{\partitionSet}{\cS}

In this subsection, we provide approximation lower bounds
for censorship and direct
signaling schemes. 
In fact, we present a stronger result that
quantifies
the optimal payoff loss of a signaling scheme 
via its maximum number $L$ of signals induced by each state.
\begin{theorem}
\label{thm:monotone partition lower bound}
\label{thm:signal size lower bound}
In SDSU environments, 
there exists a problem instance (\Cref{example:SDSU lower bound example})
such that for any signaling scheme $\signalscheme$
with signal space $\signalspace$,
it is $\Omega(\sfrac{m}{L})$-approximation
to the optimal signaling scheme, where $L \triangleq \max_{i\in[m]}
|\{\scaledNoiseDiff\in \signalspace:\signalscheme_i(\scaledNoiseDiff) > 0\}|$
denotes the maximum number of signals induced by a state
in this signaling scheme $\signalscheme$.
\end{theorem}
The above result implies approximation
lower bounds for censorship and direct signaling schemes.\footnote{
Though it is not our focus, another broader 
class of signaling schemes that 
are studied in the literature is the
monotone partitional signaling scheme \citep{K-18,DM-19,C-19}.
Both censorship signaling schemes and direct
signaling schemes are also monotone partitional signaling schemes.
A notably fact about 
monotone partitional signaling schemes is that 
each state can only induce at most $3$ signals.
Thus,
\Cref{thm:monotone partition lower bound} also implies that 
there exists a problem instance (\Cref{example:SDSU lower bound example})
such that any monotone partitional signaling scheme
is an $\Omega(m)$-approximate.}
\begin{proposition}
\label{prop: SDSU censorship approx LB}
In SDSU environments, there exists a problem instance (\Cref{example:SDSU lower bound example})
such that any censorship and any direct
signaling scheme is 
$\Omega(m)$-approximation to the optimal signaling scheme.
\end{proposition}
\begin{proof}
The above results follow from
the definition of censorship/direct signaling schemes which 
have at most $2$ signals induced from each state, 
namely, $L \le 2$, thus implying the results. 
\end{proof}
Note that the $\Omega(m)$-approximation lower bound for 
censorship signaling schemes
in SDSU environments
(\Cref{prop: SDSU censorship approx LB})
stands in contrast to
the optimality of censorship signaling schemes 
in SISU environments (\Cref{thm:SISU opt}).

\xhdr{Proof outline of \Cref{thm:monotone partition lower bound}}
In the remainder of this subsection we outline the proof of 
\Cref{thm:monotone partition lower bound} in three steps.
All missing proofs from this subsection can be found in
\Cref{apx:SDSU proof}.

\paragraph{Step 1- constructing problem instance and lower bounding
the optimal payoff.}
We first construct a problem instance 
(\Cref{example:SDSU lower bound example}) with $m$ states and 
a carefully chosen bounded rationality level $\noiseScale$
that has the following properties: 
(i) the sender can only obtain utility from state $m$, i.e., $\senderU_i > 0$
only when $i = m$;
(ii) the prior probability for each 
state $i \in [m-1]$ is exponentially increasing 
with respect to the state.  
With the above two properties, we are able to 
lower bound the optimal expected sender utility by $\Omega(K_1K_2m)$
where $K_1, K_2$ are problem-specific normalization terms 
(\Cref{lem:SDSU lower bound example opt}). 
\begin{example}
\label{example:SDSU lower bound example}
Given an arbitrary $m\in \naturals_+$,
consider a problem instance as follows:
There are $m$ states.
The receiver has 
bounded rationality level
$\noiseScale$ such that $\sfrac{\noiseScale}{\log(\noiseScale)} \geq 2m$.
The sender utility $\{\senderU_i\}$,
the receiver utility difference 
$\{\receiverU_i\}$ are 
$ 
\senderU_i = \indicator{i = m},
\receiverU_i = i, \forall i\in[m]$.
Let $K_1 \triangleq \sfrac{1}{\sum_{i\in[m-1]} 
\exp(\noiseScale i)}$.
The prior $\{\prior_i\}$
over state space $[m]$
is 
$
\prior_i 
=
K_1 K_2
\left(m - i - \frac{1}{\noiseScale}\right)
\noiseScale\exp(\noiseScale i), \forall i\in[m - 1]; 
\prior_m =
K_2$
where $K_2$ is the normalization term such that $\sum_{i\in[m]}\prior_i = 1$.
\end{example}
\begin{lemma}
\label{lem:SDSU lower bound example opt}
In \Cref{example:SDSU lower bound example},
the optimal expected sender utility $\Payoff{\optscheme} 
\geq \Omega(K_1 K_2 m)$.
\end{lemma}
\vspace{-4pt}
\paragraph{Step 2- upper bounding the payoff
via censorship signaling schemes.}
In this step, we show that for any signaling scheme,
we can upper bound 
expected sender utility in \Cref{example:SDSU lower bound example}
via the utility from a set of censorship signaling schemes. 
In particular,
for each state $i \in [m-1]$,
given any possible pooling signal $\scaledNoiseDiff\in[\receiverU_i, \receiverU_m]$, 
we define following censorship signaling scheme
where state $i$ and state $m$ are pooled 
on signal $\scaledNoiseDiff$,
and other states are fully revealed. 
Let $\NOinforSignal \triangleq 
\sfrac{(\prior_i i + \prior_m m)}{(\prior_i + \prior_m)}$
be the pooling signal which state $i$ and state $m$
are fully pooled together.
We consider following censorship signaling scheme $\lowerbounddecomposescheme$:
if $\decomposesignal \leq \NOinforSignal$,
signaling scheme $\lowerbounddecomposescheme$ admits the form as follows
\begin{align*}
    \lowerbounddecomposescheme_i(\decomposesignal) = 1;
    \quad
    \lowerbounddecomposescheme_m(\decomposesignal) = 
    \frac{\prior_i}{\prior_m}\frac{
    \decomposesignal - i}{m - \decomposesignal};
    \quad
    \lowerbounddecomposescheme_m(m) = 1 - 
    \frac{\prior_i}{\prior_m}\frac{
    \decomposesignal - i}{m - \decomposesignal};
    \quad
    \lowerbounddecomposescheme_j(j) = 1
    ~ \forall j\neq i, m
\end{align*}
and if $\decomposesignal \geq \NOinforSignal$,
signaling scheme $\lowerbounddecomposescheme$ admits the form as follows
\begin{align*}
    \lowerbounddecomposescheme_i(i) = 1
    -
    \frac{\prior_m}{\prior_i}
    \frac{m - \decomposesignal}{\decomposesignal - i};
    \quad
    \lowerbounddecomposescheme_i(\decomposesignal) = 
    \frac{\prior_m}{\prior_i}
    \frac{m - \decomposesignal}{\decomposesignal - i};
    \quad
    \lowerbounddecomposescheme_m(\decomposesignal) = 1;
    \quad
    \lowerbounddecomposescheme_j(j) = 1 
    \forall j \neq i, m
\end{align*}
Fix any signaling scheme $\signalscheme$ where 
the signals induced by state $m$ are
$\{\scaledNoiseDiff_\ell\}_{\ell\in L}$. By definition, 
\begin{align*}
    \Payoff{\signalscheme} \leq 
    \sum_{\ell\in[L]}
    \sum_{i\in[m - 1]}
    \Payoff{\lowerbounddecomposeschemeell}~.
\end{align*}
Now it remains to upper bound 
$\Payoff{\lowerbounddecomposeschemeell}$ for each state $i\in[m-1]$
and each $\ell\in[L]$.



\paragraph{Step 3- upper bounding $\Payoff{\signalscheme^{(i, \scaledNoiseDiff)}}$.}
In this step, we upper bound 
the expected sender utility under 
the signaling scheme $\signalscheme^{(i, \scaledNoiseDiff)}$. 
We below provide two characterizations on the upper bound
of the expected sender utility 
$\Payoff{\signalscheme^{(i, \scaledNoiseDiff)}}$
(\Cref{lem:single signal payoff lowerbound large}), 
depending on the value of pooling signal $\scaledNoiseDiff$.
The proof of this lemma is deferred to 
\Cref{apx:SDSU proof single signal payoff lowerbound large}.
\ifEC
\begin{restatable}{lemma}{singlesignalpayofflowerboundlarge}
\label{lem:single signal payoff lowerbound large}
In \Cref{example:SDSU lower bound example},
for any state $i\in[m - 1]$,
the expected sender utility  $\Payoff{\lowerbounddecomposescheme} 
= O({K_1K_ 2})$ for any $\decomposesignal\in
[i, i + \sfrac{m\log(\noiseScale)}{\noiseScale}]$;
and $\Payoff{\lowerbounddecomposescheme} 
= o(\sfrac{K_1K_ 2}{m})$ for any $\decomposesignal\in
[i + \sfrac{m\log(\noiseScale)}{\noiseScale}, m]$.
\end{restatable}
\fi

With the above two characterizations on 
$\Payoff{\lowerbounddecomposescheme}$, we are now ready to prove 
\Cref{thm:monotone partition lower bound}.
\begin{proof}
[Proof of \Cref{thm:monotone partition lower bound}]
Let $L' \triangleq|\{\scaledNoiseDiff\in\signalspace:\signalscheme_m(\scaledNoiseDiff) > 0\}|$ 
be the number of signals induced by state $m$,
and denote these $L'$ signals as $\{\scaledNoiseDiff_{\ell}\}_{\ell\in[L']}$.
For each $\ell\in[L']$,
since $\sfrac{\noiseScale}{\log(\noiseScale)}\geq 2m$,
there exists an most one state $j\in[m- 1]$
such that $\scaledNoiseDiff_\ell \in [j, j + \sfrac{m\log(\noiseScale)}{\noiseScale}]$.
Invoking \Cref{lem:single signal payoff lowerbound large},
we know that $\sum_{i\in[m - 1]}\Payoff{\lowerbounddecomposeschemeell}
= O(K_1 K_2)$.
Thus, invoking \Cref{lem:SDSU lower bound example opt},
we have
\begin{align*}
    \frac{\Payoff{\optscheme}}{\Payoff{\signalscheme}}
    = 
    \frac{\Omega(K_1 K_2 m)}{L' \cdot O(K_1 K_2)}
    =
    \Omega\left(\frac{m}{L'}\right),
\end{align*}
which concludes the proof for \Cref{thm:monotone partition lower bound}.
\end{proof}

%% file: Paper/state-dependent-simple.tex
\newcommand{\ied}{^{(i)}}
\newcommand{\jed}{^{(j)}}
\newcommand{\totalsenderU}{U}
\newcommand{\xbf}{\mathbf{x}}
\newcommand{\edgealloc}{x}
\newcommand{\edgealloci}{\edgealloc_{ij}}
\newcommand{\edgereward}{\omega}
\newcommand{\edgerewardi}{\edgereward_{ij}}
\newcommand{\edgeusage}{c}
\newcommand{\edgeusagei}{\edgeusage_{ij}\ied}
\newcommand{\edgeusagej}{\edgeusage_{ij}\jed}
\newcommand{\Edge}{E}

In this subsection, 
we discuss the approximation upper bounds  
of censorship and direct 
signaling schemes.
The approximation upper bounds we provide here 
are indeed tight according to the lower bounds
we established in \Cref{sec:SDSU simple approx lower bound}. 

\begin{theorem}
\label{coro:SDSU censorship m approx}
In SDSU environments, for a boundedly rational receiver,
there exists a 
censorship/direct signaling scheme that is an
$O(m)$-approximation to the optimal signaling scheme.
\end{theorem}
We would like to highlight that 
designing censorship or direct signaling scheme with $O(m)$-approximation is not in-hindsight straightforward.
For example, even for a
fully rational receiver, 
the approximation of the full/no-information revealing 
or the better of the two could be unbounded.
To establish \Cref{coro:SDSU censorship m approx}, 
we start with characterizing a 4-approximation 
signaling scheme that has desired structure properties 
-- 
the sender's signal either reveals the true state, 
or randomizes the receiver's uncertainty on
only two states,
then we utilize the structure of this 4-approximation 
signaling scheme to show the existence 
of $O(m)$-approximation censorship/direct signaling schemes.

\begin{lemma}
\label{thm:SDSU 4 approx}
In SDSU environments, for a boundedly rational receiver,
there exists a $4$-approximation signaling scheme
using at most $2m$ signals, and it 
has the following two properties:
\begin{enumerate}
    \item[(i)] each signal $\signal\in\optsignalspace$
is induced by at most two states,
i.e., $|\supp(\posterior(\signal))|\leq 2$;
    \item[(ii)] each pair of states $(i, j)$
is pooled at most one signal,
i.e., $|\{\signal\in \optsignalspace:
\supp(\posterior(\signal)) = \{i,j\}\}| \leq 1$.
\end{enumerate}
Furthermore, at most $m$ signals are induced 
by two distinct states,
i.e., $|\{\signal:|\supp(\posterior(\signal))| = 2\}| \leq m$.\footnote{Recall property~(i) 
requires that for every $\signal$, $|\supp(\posterior(\signal))|\leq 2$.}
\end{lemma}

We now provide intuitions of the two properties 
of the signal scheme characterized 
in \Cref{thm:SDSU 4 approx}.
Property~(i)
ensures that whenever the receiver sees a signal,
she can infer that the realized state must 
be one of two particular states. 
From a practical perspective, this property is  
beneficial to a boundedly rational receiver as it 
makes the receiver's state inference easier.
From the sender's perspective, 
property~(ii) ensures that, 
for any pair of states, the sender only needs to design at most one pooling signal.
We provide a proof overview of \Cref{thm:SDSU 4 approx} 
in the end of this subsection and defer its formal proof
to \Cref{apx:SDSU proof 4 approx}.

With the results in \Cref{thm:SDSU 4 approx}, 
we are now ready to 
prove the \Cref{coro:SDSU censorship m approx}. 
\vspace{-6pt}
\begin{proof}[Proof 
of \Cref{coro:SDSU censorship m approx}]
We first prove that there always exists a 
censorship signaling scheme that is $O(m)$-approximation.
Let 
$\signalscheme\primed$
with signal space $\signalspace\primed$
be the signaling scheme
stated in
\Cref{thm:SDSU 4 approx}.
We denote $\totalsenderU_{ij}$
as the expected sender utility
induced by each pair of state $(i, j)$,
i.e., $\totalsenderU_{ij} \triangleq \sum_{\signal:
\signalscheme\primed_i(\signal)> 0
\land
\signalscheme\primed_j(\signal)> 0}
(\prior_i\senderU_i\signalscheme\primed_i(\signal)
+
\prior_j\senderU_j\signalscheme\primed_j(\signal))\Quant(\signal)$.
Let $(i^*, j^*) = \argmax_{(i, j)}
\totalsenderU_{ij}$.
Note that by definition, 
and the property~(i) of signaling scheme $\signalscheme\primed$
we have $\Payoff{\signalscheme\primed} \leq m\cdot \totalsenderU_{i^*j^*}$.

Consider a binary-state instance $\instance = (\hat m,
\{\hat\prior_k\}\{\hat\receiverU_k\}, \{\hat\senderU_k\})$
induced by 
pair of states $(i^*, j^*)$,
i.e.,
\begin{align*}
    \hat m \gets 2,\qquad
    \hat \receiverU_1 &\gets \receiverU_{i^*},\qquad
    \hat \receiverU_2 \gets \receiverU_{j^*},\qquad
    \hat \senderU_1 \gets \senderU_{i^*},\qquad
    \hat \senderU_2 \gets \senderU_{j^*},
    \\
    &\hat\prior_1\gets
    \frac{\prior_{i^*}}
    {\prior_{i^*}+\prior_{j^*}},
    \qquad
    \hat\prior_2 \gets
    \frac{\prior_{j^*}}
    {\prior_{i^*}+\prior_{j^*}}
\end{align*}
It can be shown that 
the optimal signaling scheme for this binary-state instance
is a censorship signaling scheme 
(see \Cref{lem:SDSU binary opt} 
and its proof in \Cref{apx:SDSU proof binary opt}).
Let $\signalscheme\doubleprimed$
be the signaling scheme 
which coincides with the optimal signaling scheme 
for this binary-state instance,
and reveals all other states.
By construction, $\signalscheme\doubleprimed$
is again a censorship, and the expected sender utility
\begin{align*}
     m\cdot \Payoff{\signalscheme\doubleprimed}
    \overset{(a)}{\geq}
     m\cdot \totalsenderU_{i^*j^*}
    \overset{(b)}{\geq}
    \Payoff{\signalscheme\primed}
    \overset{(c)}{\geq}
    \frac{1}{4} \cdot \Payoff{\optscheme}
\end{align*}
where $\optscheme$ is the optimal signaling scheme,
(a) holds due to the construction of 
$\signalscheme\doubleprimed$;
(b) holds due to the definition of 
$(i^*,j^*)$;
and
(c) holds since $\signalscheme\primed$
is a 4-approximation to 
the signaling scheme $\optscheme$.

The proof of the $O(m)$-approximation for 
direct signaling scheme follows the similar argument
which utilizes the 
structure of the signaling scheme $\signalscheme\primed$, 
and thus is deferred to 
\Cref{apx:proof SDSU censorship m approx}.
\end{proof}

Before finishing this subsection, 
we provide a proof overview for \Cref{thm:SDSU 4 approx}, and we defer the formal proof to 
\Cref{apx:SDSU proof 4 approx}.
At a high-level, our proof consists of two main steps. 
In the first step, we show that 
within the subclass of signaling schemes satisfying properties (i) (ii) in \Cref{thm:SDSU 4 approx},
there 
exists a signaling scheme $\optscheme$
using at most $O(m^2)$ signals and 
achieving the optimality over all signaling schemes. 
In the second step,
we discuss how to construct the signaling scheme 
stated in \Cref{thm:SDSU 4 approx}
based on the optimal signaling scheme $\optscheme$
identified in the first step.
Specifically,
we establish a connection to 
the fractional generalized assignment problem
\citep{ST-93}.
In particular, by leveraging those two properties (i) (ii),
we construct a linear program~\ref{eq:SDSU opt matching lp} based on the optimal signaling scheme identified 
in the first step.
This linear program upperbounds
the optimal expected sender utility
and has the same formulation as the 
fractional generalized assignment problem. 
\cite{ST-93}
show that the optimal integral solution
of program~\ref{eq:SDSU opt matching lp}
(which has at most $m$ non-zero entries)
is a 2-approximation 
to the optimal fractional solution (which may have at most $m^2$ non-zero entries). 
With this result, we then convert
this optimal integral solution 
to a signaling scheme stated in \Cref{thm:SDSU 4 approx},
which has at most $2m$
signals,
and is a 2-approximation 
to the objective value of the optimal integral solution.

%% file: Paper/robust.tex
\newcommand{\doublestarred}{^{**}}
\newcommand{\rationalityClass}{\mathcal{B}}
\newcommand{\robustRatio}{\Gamma}
\newcommand{\optschemebeta}{\texttt{OPT}(\noiseScale)}
\newcommand{\optschemeInfty}{\hat{\signalscheme}^*}

\newcommand{\censorshipprobbeta}{p^{\dagger}}
\newcommand{\censorshipstatebeta}{i^{\dagger}}
\newcommand{\censorshipsignalbeta}{\scaledNoiseDiff^{\dagger}}

\newcommand{\censorshipprobInfty}{\hat{p}^{\dagger}}
\newcommand{\censorshipstateInfty}{\hat{i}^{\dagger}}
\newcommand{\censorshipsignalInfty}{\hat{\scaledNoiseDiff}^{\dagger}}

\newcommand{\noiseScaleLB}{{\noiseScale_0}}
\newcommand{\noiseScaleUB}{{\overline{\noiseScale}}}

\newcommand{\optschemebetaell}{\texttt{OPT}(\noiseScale_\ell)}
\newcommand{\pen}{\psi}

\newcommand{\schemebinary}{\widetilde{\signalscheme}^*}
\newcommand{\belowpoolingsignal}{\widetilde{\delta}}

\newcommand{\aggregateprobaq}{q_{ij}^{(1)}}
\newcommand{\aggregateprobbq}{q_{ij}^{(2)}}

\newcommand{\adjustedConst}{K}

In practice, the sender may not be able to
have 
(or require significant cost to learn)
the perfect
knowledge of a receiver's 
bounded rationality level.
Motivated by this concern, we introduce 
\emph{rationality-robust information design},
in which a signaling scheme (a.k.a., information structure) is designed for a receiver
whose bounded rationality level is unknown.
The goal is to identify \emph{robust}
signaling schemes
--
ones with good (multiplicative) 
rationality-robust 
approximation to the optimal signaling scheme
that is tailored to the receiver's bounded
rationality level.

\begin{definition}
Fixing any problem instance $\instance 
= (m, \{\prior_i\}, \{\receiverU_i\}, \{\senderU_i\})$,  
the \emph{rationality-robust approximation ratio}
$\robustRatio(\signalscheme,\rationalityClass)$
of a given signaling scheme $\signalscheme$
and a set of possible bounded rationality levels
$\rationalityClass\subseteq[0,\infty)$
is 
\begin{align*}
    \robustRatio(\signalscheme, \rationalityClass)
    \triangleq \max_{\noiseScale\in \rationalityClass} \frac{\Payoff[\noiseScale]{\texttt{OPT}(\noiseScale)} }{\Payoff[\noiseScale]{\signalscheme}}
\end{align*}
where $\texttt{OPT}(\noiseScale)$
is the optimal signaling scheme\footnote{
Here we write the optimal signaling scheme
with bounded rationality level $\noiseScale$ as
$\texttt{OPT}(\noiseScale)$, instead of $\optscheme$ in previous sections,
to emphasize its dependency on 
the rationality level $\noiseScale$.
} for a 
receiver with bounded rationality level $\noiseScale$
(characterized in \Cref{lem:opt fully rational},
\Cref{thm:SISU opt}, \Cref{thm:SDSU opt});
and $\Payoff[\noiseScale]{\texttt{OPT}(\noiseScale)}$
(resp.\ 
$\Payoff[\noiseScale]{\signalscheme}$)
is the expected sender utility of signaling scheme 
$\texttt{OPT}(\noiseScale)$
(resp.\ $\signalscheme$)
for 
bounded rationality level $\noiseScale$.
\end{definition}
In the above definition, the rationality-robust
approximation ratio is defined in worst-case over
the set $\rationalityClass$ of possible bounded rationality levels.
Ideally, one would like to have a signaling scheme 
that is approximately optimal under 
any bounded rationality level,
i.e., $\rationalityClass = [0,\infty)$.
This is the scenario illustrated in \Cref{sec:SISU robust},
in which we show that 
in SISU environments, 
the optimal censorship 
signaling scheme for a fully rational receiver 
can achieve 
$2$ rationality-robust approximation
for any receiver's bounded rationality level 
(\Cref{thm:2 robust approx SISU upperbound}).
This suggests that, up to a two factor,
the knowledge of the bounded rationality level 
are unimportant in SISU environments;
and directly optimizing under fully rational receiver 
model is robust enough.
In contrast, as we show in \Cref{sec:SDSU robust},
there exists no signaling scheme with 
bounded rationality-robust approximation ratio in SDSU environments,
when the sender has no knowledge of the 
receiver's bounded rationality 
level (\Cref{thm:imposs robust approx SDSU}).
This impossibility result indicates that (a) there 
exists a tradeoff between
the knowledge of the receiver's rationality level 
and the achievable rationality-robustness;
and (b) even if the adversary is restricted to pick receiver's behavior in the quantal response model, 
designing robust signaling scheme 
still requires additional knowledge.
Finally, we show a preliminary positive result in SDSU environments:
for problem instances with binary state,
when the actual rationality robust level is sufficiently large,
learning the bounded rationality level up 
to a multiplicative error 
enables the sender to design signaling schemes
with good rationality-robust approximation guarantee.

\subsection{Rationality-Robust Signaling Schemes in SISU Environments}
\label{sec:SISU robust}
In SISU environments, we show that for any problem instance, 
the optimal censorship signaling scheme 
(defined in \Cref{lem:opt fully rational})
for a fully rational receiver 
achieves 
a 
$2$ rationality-robust approximation 
when the sender has no knowledge of 
the receiver's bounded rationality level.
We also provide an example to show the tightness of the result.

\begin{restatable}{theorem}{robustapproxSISUupperbound}
\label{thm:2 robust approx SISU upperbound}
In SISU environments, 
for any problem instance, 
the optimal censorship signaling scheme $\optschemeInfty$ 
(defined in \Cref{lem:opt fully rational})
for a fully rational receiver 
has rationality-robust approximation ratio
$\robustRatio(\optschemeInfty, [0,\infty)) \leq 2$.
\end{restatable}
To understand the intuition behind the above theorem,
recall that the structural property of optimal censorship signaling scheme
established in \Cref{prop:SISU threshold state monotonicity}:
The optimal censorship for a less rational receiver 
requires the sender to reveal more information. 
As a result, the optimal censorship
for receiver with 
$\noiseScale = \infty$ reveals least information and 
pools most states compared to 
other optimal censorship for
receiver with $\noiseScale < \infty$. 
Meanwhile, the pooling signal $\hat\scaledNoiseDiff\primed 
 \equiv 0$ 
in
$\optschemeInfty$ 
ensures that the utility
contributed from those pooled states is at least 
half of the utility
contributed from those states in the optimal censorship
with less rational receiver.
\begin{proof}[Proof of \Cref{thm:2 robust approx SISU upperbound}]
Fix any bounded rationality level $\noiseScale\in\rationalityClass$.
For signaling scheme $\optschemeInfty$, 
\begin{align*}
    \Payoff[\noiseScale]{\optschemeInfty}
    &= 
     \sum_{i\in[\censorshipstateInfty-1]} \prior_i \Quant(0) + \prior_{\censorshipstateInfty} \left(\censorshipprobInfty \Quant(0) + \left(1 - \censorshipprobInfty\right)
    \Quant\left(\receiverU_{\censorshipstateInfty}\right)\right) + 
    \sum_{i\in [\censorshipstateInfty + 1:m]}\prior_i \Quant(\receiverU_i)
    \\
    &\geq
    \sum_{i\in[\censorshipstateInfty-1]} \prior_i \Quant(0) 
    + 
    \sum_{i\in [\censorshipstateInfty:m]}\prior_i \Quant(\receiverU_i)
\end{align*}
where $\censorshipstateInfty, \censorshipprobInfty$ 
is the threshold state, the threshold state probability of $\optschemeInfty$.\footnote{Here we use the superscript $\hat{}$
\ 
to denote the concepts in signaling scheme $\optschemeInfty$.}
Moreover, by \Cref{thm:SISU opt},
the optimal expected sender utility
under the bounded rationality level $\noiseScale$ is 
\begin{align*}
    \Payoff[\noiseScale]{\optschemebeta}
    &= \sum_{i\in[\censorshipstatebeta - 1]} \prior_i \Quant(\censorshipsignalbeta) + \prior_{\censorshipstatebeta} \left(\censorshipprobbeta \Quant(\censorshipsignalbeta) + \left(1 - \censorshipprobbeta\right) \Quant(\receiverU_{\censorshipstatebeta})\right) + 
    \sum_{i\in[\censorshipstatebeta + 1:m]} \prior_i\Quant(\receiverU_i)
    \\
    &\leq \sum_{i\in[\censorshipstatebeta]} \prior_i \Quant(\censorshipsignalbeta) + 
    \sum_{i\in[\censorshipstatebeta + 1:m]} \prior_i\Quant(\receiverU_i)
\end{align*}
Recall \Cref{prop:SISU threshold state monotonicity}
implies that $\censorshipstateInfty \ge \censorshipstatebeta$.
Hence,
\begin{align*}
    \frac{\Payoff[\noiseScale]{\optschemebeta}}{\Payoff[\noiseScale]{\optschemeInfty}}
    \leq
    \max\left\{
    \max_{i\in[\censorshipstatebeta]}
    \frac{\prior_i \Quant(\censorshipsignalbeta)}{\prior_i\Quant(0)},
    \max_{i\in[\censorshipstatebeta + 1:\censorshipstateInfty - 1]}
    \frac{\prior_i \Quant(\receiverU_i)}{\prior_i\Quant(0)},
    \max_{i\in[\censorshipstateInfty:m]}
    \frac{\prior_i \Quant(\receiverU_i)}{\prior_i\Quant(\receiverU_i)}
    \right\}
    =
    \frac{ \Quant(\censorshipsignalbeta)}{\Quant(0)}
    \overset{(a)}{\leq} 2
\end{align*}
where 
inequality (a) holds since $\Quant(0) = \sfrac{1}{2} \geq 
\sfrac{\Quant(\scaledNoiseDiff)}{2}$ for all $\scaledNoiseDiff\in(-\infty,\infty)$.
\end{proof}

The below result (its proof is deferred to 
\Cref{apx:proof 2 robust approx SISU lower bound}) 
shows the tightness 
of the robust-rationality approximation ratio established in 
\Cref{thm:2 robust approx SISU upperbound}.

\begin{proposition}
\label{prop:2 robust approx SISU lower bound}
In SISU environments, 
for any $\varepsilon > 0$,
there exists a problem instance 
such that the optimal censorship $\optschemeInfty$ 
(defined in \Cref{lem:opt fully rational})
for a fully rational receiver 
has rationality-robust approximation ratio
$\robustRatio(\optschemeInfty, [0,\infty)) \geq 2 - \varepsilon$.
\end{proposition}

\newcommand{\optdirectInfy}{\widetilde{\signalscheme}^*}

We conclude this subsection by noting that
the robust signaling scheme $\optschemeInfty$ 
used in \Cref{thm:2 robust approx SISU upperbound}
is the optimal censorship
for a fully rational receiver.
However, as we show in  
\Cref{prop:unbounded robust direct approx SISU lower bound}
below
(its proof is straightforward and is deferred to 
\Cref{apx:unbounded robust proof direct approx SISU lower bound}),
the optimal direct 
signaling scheme $\optdirectInfy$ 
(defined in \Cref{lem:opt fully rational})
for a fully rational receiver cannot achieve any meaningful
rationality-robust approximation guarantee. 
This again mirrors the analogous separation results
on the censorship and direct signaling schemes we show in 
previous sections.
\begin{proposition}
\label{prop:unbounded robust direct approx SISU lower bound}
In SISU environments, 
there exists a problem instance 
such that the optimal direct signaling scheme 
$\optdirectInfy$ 
for a fully rational receiver 
has rationality-robust approximation ratio
$\robustRatio(\optdirectInfy, [0,\infty)) = \infty$.
\end{proposition}

\subsection{Rationality-Robust Signaling Schemes in SDSU Environments}
\label{sec:SDSU robust}
Unlike SISU environments where
the knowledge of the rationality level 
is unimportant up to a two factor (\Cref{thm:2 robust approx SISU upperbound}),
in this subsection, we first present the following negative result 
that without the knowledge of the rationality level,
there exists no signaling scheme with 
bounded rationality-robust approximation ratio,
even if the state space is binary (\Cref{thm:imposs robust approx SDSU}).
Nonetheless, we also provide a positive result 
for binary-state problem instances
under a reasonable condition 
of receiver's bounded rationality levels (\Cref{thm:SDSU binary robust}).

\begin{restatable}{theorem}{impossrobustapproxSDSU}
\label{thm:imposs robust approx SDSU}
In SDSU environments, 
there exists a problem instance (\Cref{example: SDSU rationality impossibility})
with binary state
such that
for any signaling scheme $\signalscheme$
and any $\noiseScaleLB \geq 0$,
the rationality-robust approximation ratio 
with respect to $\rationalityClass=[\noiseScaleLB,\infty)$
is unbounded,
i.e., 
$\robustRatio(\signalscheme, [\noiseScaleLB,\infty)) = \infty$.
\end{restatable}
\xhdr{Proof overview of \Cref{thm:imposs robust approx SDSU}}
The formal proof of \Cref{thm:imposs robust approx SDSU} is 
deferred to \Cref{apx:imposs robust proof CHECK}. 
Here we sketch the high-level idea 
behind the proof.
Our proof proceeds with two steps as follows.
In the first step, we construct a binary-state problem instance 
in \Cref{example: SDSU rationality impossibility}.
We further provide a 
finite set $\rationalityClass\triangleq 
\{\noiseScale_\ell\}_{\ell\in[L]}$
where\footnote{Here $L$ is a sufficiently large constant, which goes to infinite in the end of the analysis.}
$\noiseScale_\ell \triangleq L^\ell$.
Recall that when the state space is binary,
the optimal signaling scheme
is a censorship signaling scheme (\Cref{lem:SDSU binary opt}).
The construction in \Cref{example: SDSU rationality impossibility}
ensures that
the contribution in the optimal (censorship) 
signaling scheme 
mainly comes from the pooling signal $\censorshipsignal$.
Similar to the analysis in \Cref{thm:monotone partition lower bound},
the value of $\censorshipsignal$ 
is quite sensitive to the rationality level $\noiseScale_\ell$.
As a consequence,
for any $\noiseScale_1, \noiseScale_2\in\rationalityClass$
such that $\noiseScale_1 \not=\noiseScale_2$,
it satisfies that 
$\Payoff[\noiseScale_2]{\texttt{OPT}(\noiseScale_1)} \ll 
\Payoff[\noiseScale_2]{\texttt{OPT}(\noiseScale_2)}$, 
which says that the optimal signaling scheme under a specific rationality level must have a very bad performance 
if sender implements such optimal signaling scheme 
with a receiver who has a different bounded rationality level.
\vspace{-6pt}
\begin{example}
\label{example: SDSU rationality impossibility}
Consider the following problem instance
with binary state (i.e., $m = 2$),
\begin{align*}
    \prior_1 =  \frac{1}{2}, \qquad 
    \prior_2 =  \frac{1}{2}, \qquad 
    \receiverU_1 = 1, \qquad
    \receiverU_2 = 2, \qquad
    \senderU_1 = 0, \qquad
    \senderU_2 = 1.
\end{align*} 
\end{example}
In the second step, given the above constructed 
binary-state problem instance
and the set $\rationalityClass$ of rationality levels,
we introduce the following \emph{factor-revealing program}
to lower bound 
the optimal rationality-robust approximation ratio,
i.e., $\min_{\signalscheme}\robustRatio(\signalscheme,\rationalityClass)$.
\begin{align*}
    &\begin{array}{lll}
     \min\limits_{\boldsymbol{\signalscheme} \geq \zerobf, \robustRatio \geq 0} & \robustRatio
       & \text{s.t.} 
     \\
       & 
       \prior_1\signalscheme_1(\scaledNoiseDiff)\cdot \left(\scaledNoiseDiff - \receiverU_1\right)
        +\prior_2 \signalscheme_2(\scaledNoiseDiff)\cdot \left(\scaledNoiseDiff - \receiverU_2\right) \ge 0
       & \scaledNoiseDiff\in(-\infty, \infty)
     \\
       & 
       \displaystyle\int_{-\infty}^{\infty} \signalscheme_i(\scaledNoiseDiff)d\scaledNoiseDiff = 1
        &  i\in[2]
     \\
       & 
       \signalscheme_i(\scaledNoiseDiff) \ge 0
         & \scaledNoiseDiff\in(-\infty, \infty), ~  i\in[2]
     \\
       & 
    \Payoff[\noiseScale_\ell]{\signalscheme}
        \ge \frac{1}{\robustRatio} \frac{1}{\noiseScale_\ell \exp(\noiseScale_\ell)}, 
          &  \ell\in[L]
    \end{array}
\end{align*}
In this program, the variables $\boldsymbol\signalscheme$
can be interpreted as a signaling scheme,
and $\robustRatio$ can be interpreted as its 
rationality-robust approximation ratio.
In particular, the last constraint requires 
the expected sender utility of signaling scheme $\signalscheme$
for a receiver with bounded rationality level $\noiseScale_\ell$
is at least a $\robustRatio$-approximation to 
$\sfrac{1}{\noiseScale_\ell\exp(\noiseScale_\ell)}$,
which, as we show in the proof, is a lower bound 
of the optimal expected sender utility $\Payoff[\noiseScale_\ell]{\optschemebetaell}$.
Notably, this program
is essentially a linear program.
Hence, by 
explicitly constructing a dual assignment in 
its dual program and then invoking the weak duality,
we can lower bound its optimal objective value by $\Omega(L)$.
Finally, setting $L$ to be infinite finishes the proof of \Cref{thm:imposs robust approx SDSU}.

\xhdr{Positive result for binary-state instances in SDSU environments}
\Cref{thm:imposs robust approx SDSU}
highlights the importance of the knowledge of the receiver's bounded rationality level in SDSU environments.
Namely, even there are {\em only} two states, if the sender
does not have any knowledge about receiver's bounded rationality level, 
then it is impossible to hope for a robust signaling scheme 
that would have bounded rationality-robust approximation ratio.

In \Cref{apx:SDSU binary robust proof}, we 
present a positive result for problem instances 
with binary states 
(see \Cref{thm:SDSU binary robust} and its the proof in \Cref{apx:SDSU binary robust proof}), which shows that 
if the sender learns
the receiver's bounded rationality level up 
to a multiplicative error $K \ge 1$, i.e., $\rationalityClass = [\noiseScaleLB,K\noiseScaleLB]$,
and $\noiseScaleLB$ is larger than an instance-dependent bound,\footnote{Recall our 
negative result (\Cref{thm:imposs robust approx SDSU})
shows that 
there exists no signaling scheme with finite rationality-robust approximation ratio with respect to $\rationalityClass = [\noiseScaleLB,\infty)$ for any $\noiseScaleLB\geq 0$.
It remains as an open question whether 
similar rationality-robust signaling schemes 
exists for small $\noiseScaleLB$.
}
then there exists (censorship) signaling schemes
whose rationality-robust approximation ratio 
depends linearly on
multiplicative error $K$.

%% file: Paper/future-directions.tex
In this work, we develop a theory of rationality-robust information design in the canonical setting of Bayesian persuasion with binary receiver action. 
We first identify conditions under which the optimal signaling scheme structure for a fully rational receiver remains optimal or approximately optimal for a boundedly rational receiver. 
We then study the existence and construction of robust signaling schemes when there is uncertainty about the receiver's bounded rationality level.
Below we highlight the following natural and important directions 
of future research.

\yfr{The most general direction from this paper is to develop a theory 
of information design or mechanism design 
for agents with bounded rationality. 
Most existing results on this direction
restrict attention to specific problems
(see \Cref{apx:further related work} for more details).
An interesting question is whether there exist
conditions under which the optimal/approximately optimal 
results for fully rational agents extend
to boundedly rational agents under a broad class of 
information/mechanism design problems.
For agents with bounded rationality, 
the standard revelation principle fails,
and it is no longer without loss of generality
to impose incentive compatibility. 
In this sense, the bounded rationality 
also provides a motivation and new perspective on 
the recent literature on non-truthful mechanism design
\citep[e.g.,][]{FH-18,CTW-19,DFLSV-20,AKSW-22}.}

The bounded rationality specifies 
how agents select their actions.
\wtrevision{
Therefore, similar to our findings, mechanisms that are equivalent under fully rationality
(e.g., second-price auction and English auction)
may lead to different outcomes under bounded rationality. 
Exploring our first question in mechanism design context may provide an alternative justification on 
practical preference of certain mechanisms format \citep[cf.][]{AL-20}.
For information design problems, action sets of agents are 
given exogenously.
In contrast, for mechanism design problems, 
action sets for agents are usually designed endogenously.
Thus, it is also interesting 
to systematically develop theory to understand how to 
design action sets (a.k.a., mechanism formats) and preference over 
classic format.
}

For our Bayesian persuasion problem, 
there are also several interesting open questions.
In SISU environments, 
a natural question is whether there exists 
a signaling scheme that can beat 
the 2 rationality-robust approximation ratio
achieved by the optimal censorship for 
the fully rational receiver. 
In SDSU environments, one immediate question is whether 
there exists a robust signaling scheme for problem instances with multiple states under a reasonable boundedness condition 
on the receiver's bounded rationality levels.
More importantly, what is the fine-grained tradeoff 
between the knowledge on receiver's behavior 
and the achievable
rationality-robustness of signaling schemes?
Conceptually, these questions
share similar flavor with
the prior-independent mechanism design framework
\citep[e.g.,][]{DRY-15,FILS-15,AB-20,HJL-20}.

Finally, another direction of interest is to characterize the computational complexity of computing an optimal (or approximately optimal) signaling scheme in different environments. Note that when the receiver is fully rational, the optimal signaling scheme can be computed in polynomial time. 
When the receiver is boundedly rational, in \Cref{apx:complexity}, we present some preliminary results on characterizing the complexity of computing the (approximately) optimal signaling scheme in both SISU/SDSU environments. 
Whether our results can be strengthened is an interesting and important future direction. 

%% file: Paper/further-related.tex

Our work 
on relaxing rationality assumption in information design 
is conceptually similar to
a large literature in mechanism design without/relaxing
rationality assumption. 
For example, \citet{BMSW-18} 
and follow-up works 
\citet{DSS-19a,DSS-19b}
study revenue-maximization 
for a single buyer who uses no-regret algorithm in a repeated game with a seller.
\citet{CHJ-20}
study a repeated Stackelberg game 
where both players use no-regret learning algorithm.
Behavioral mechanism design \citep{EG-15}
study how departures from standard economic models of 
agent behavior affect mechanism design.
\citet{CGMP-18} study the revenue-maximization when the buyer's 
behavioral model is beyond expected utility theory
and characterize mechanism that is robust to the buyer's risk attitude.
Other related works in mechanism design include
\cite{FHH-13,CDKS-22,DP-12}.

Our work relates to a rich literature on information design. 
Since the seminal work \citep{KG-11}
that setup the Bayesian persuasion problem that studies the game 
on strategic communication between a sender and a receiver, 
the framework has inspired an active line of
research in information design games 
\citep[e.g., see the surveys by][]{D-17, Kamenica-19,BM-19}.
In addition to applications mentioned in introduction,
Bayesian persuasion has also been studied 
in other different applications like 
online ad auction 
\citep{EFGPT-14,CDHW-20, AB-19, BHMSW-22}, 
recommendation \cite{MSSW-22,FTX-22},
and voting \citep{AC-16,AC-16b}. 
Our work extends this line of research by
relaxing the standard rationality assumption.  
In particular, we consider 
a boundedly rational receiver by modeling 
her as a (logit) quantal response player, while 
standard framework usually assumes that the receiver
is fully rational, i.e., an expected utility maximizer. 
Relaxing rationality assumptions has been studied
in other information design literature. 
For example, \citet{CZ-22} study how receiver's 
mistakes in probabilistic inference impact optimal persuasion, 
\citet{AIL-20} study a persuasion problem where 
the receiver's utility may be nonlinear in her belief, 
and \citet{TH-21,YTNH-23} run behavioral experiments and relax 
the Bayesian rational assumption in a simple persuasion setting.
\citet{CCMG-20,CMCG-21} also relax traditional assumptions 
in an online setting.
Our work also conceptually relates to 
recent papers that focus on settings 
where the receiver has limited attention to 
process and utilize the information \citep{LMW-20,BS-20}.
Since it has been shown that 
the optimal stochastic response of a
rationally inattentive receiver takes a ``logit'' form~\citep{MM-15},
similar to our results in \Cref{sec:state independent}, 
\citet{BS-20} show that the optimal information policy 
in SISU environments for inattentive receivers has a censorship structure. 
Our work differs from their work as 
we consider a more general sender payoff structure
while the payoff in \cite{BS-20} depends 
linearly on the state. 
Moreover, 
in addition to characterizing the optimal information policy, we 
also study the design of approximately optimal and
rationality-robust information polices 
for boundedly rational receivers 
in both SISU and SDSU environments.
We also mention that
our persuasion problem for the boundedly rational receiver 
is equivalent to a public persuasion
problem \citep{DX-17,X-20} for a continuum population 
of {\em rational} receivers with a specific utility structure
(see \Cref{sec:reformulate} for more detailed discussions).

Our work has utilized and compared with 
censorship and direct signaling schemes.
As a general class of signaling schemes, 
censorship has been studied in the recent literature. 
\citet{KMZ-2022}
consider the setting where the sender's utility depends only on the expected state.
They show that a censorship is optimal 
if and only if the sender’s marginal utility is quasi-concave.
\citet{K-18} and \citet{AC-16} provide sufficient conditions
for the optimality of censorship in different contexts.
Our paper departs from these works by not only considering
the optimality of censorship signaling schemes 
under a different context
(i.e., with boundedly rational receiver), but also
studying its approximation guarantees when it is not optimal.
Direct signaling scheme has also been studied in 
persuasion setting with binary action
\citep{DX-17,BB-17,X-20,FTX-22}.
On the other hand, signaling schemes 
like censorship in finite state space
use at most $m$ signals, 
and direct signaling schemes use at most $2$ signals, 
\citet{GHHS-22} analyze optimal persuasion 
subject to limited signals constraint.
However, neither of the two specific classes of problems 
they consider 
-- symmetric instances and independent instances -- 
is applicable to our problem, and thus cannot inform
any approximation guarantees in our setting.
Other related works on persuasion with 
limited communication constraint include 
\cite{DKQ-16, LT-19, AT-19}.

%% file: Paper/apx-discussion-extension.tex
\subsection{Motivating Examples in \Cref{sec:intro}}
\label{sec:motivating example}
In the example of product advertising \citep{EFGPT-14,AB-19}, 
a grocery store (i.e., sender), 
who observes the true product quality (i.e., state)
with exogenous prices,
performs advertising (i.e., signaling scheme) 
to a consumer (i.e., receiver)
who makes binary purchasing decisions.
In recommendation letter \citep{D-17},
an advisor, who observes the 
true ability of students,
writes recommendation letter 
to a recruiter
who makes binary hiring decision.
In short video recommendation \citep{MSSW-22,FTX-22}, 
a short video platform (e.g., TikTok, Reels)
who observes the content of short videos,
makes recommendation to a user
who decide to either watch or skip the video.
In targeting in sponsored search \citep{BBX-18,BHMSW-22}, 
a search engine (e.g., Bing, Google)
who observes the attribute of an impression (i.e., 
a match between advertiser and user),
does targeting to an advertiser 
who decides bid or not bid for this impression.

In the examples of product advertising and recommendation letter, when the grocery store only cares about whether the buyer buys the product,
and the advisor only cares whether the student is hired, the sender's utility is independent of the realized state. Under these scenarios, both examples can be formulated as SISU environments in our work.
In the examples of short video recommendation and 
targeting in sponsored search,
the sender's utility could depend on the realized state.
For example, the short videos could
be sponsored by some companies, and these sponsored videos might bring different revenue to the platform if the user 
chooses to watch the videos.
Similarly, in targeting of sponsored search, 
different impressions could lead to different click-through rates,
the revenue will be generated to the the search engine
if the displayed advertising is clicked. 
Under these scenarios, both examples can be formulated as SDSU environments in our work.

\subsection{Extensions on without Assuming $\senderU_i(1) \geq \senderU_i(0)$}
\label{sec:assump favorable action}
By the definition of SISU environments,
we would like to first note that this assumption (i.e., $\senderU_i(1) \geq \senderU_i(0)$ for all states $i\in[m]$) trivially holds in SISU environments. 
In SDSU environments, all our lower bound results (including the impossibility result in rationality-robust information design) also hold without this assumption. 
The preliminary positive result \Cref{thm:SDSU binary robust} also holds via a similar duality argument.   
It would be an interesting future direction to explore whether 
\Cref{coro:SDSU censorship m approx} still holds without this assumption.

\subsection{Extensions to General Quantal Response Curve $\Quant$}
\label{sec:general W curve}
Our results in \Cref{sec:state independent}
and \Cref{sec:state dependent} on characterizing the
optimal signaling schemes can be readily extended
to a more general quantal response behavior $\Quant$. 
For example, the characterization on the optimality 
of censorship signaling scheme for SISU environments
(i.e., \Cref{thm:SISU opt}), the structure characterization
of the optimal signaling scheme for SDSU environments
(i.e., \Cref{thm:SDSU opt})
hold as long as 
the function $\Quant$ is S-shaped.


\subsection{Reinterpretation via Public Persuasion}
\label{sec:reformulate}
One explanation of a quantal response receiver is that 
she faces a 
action-specific random shock 
$\{\noise(\action)\}_{\action\in\cA}$ 
when she is making the decision
\citep[see][]{rus-87,MP-95}.
In particular, given posterior belief $\posterior\in\stochastic([m])$,
the receiver takes the best action $\action^*$
which maximizes her expected utility (after 
the normalization by 
the bounded rationality level $\noiseScale$) plus the action-specific random shock,
i.e.,
    $\action^* = \displaystyle \argmax\nolimits_{\action\in\cA}
    \noiseScale\cdot \receiverU(\action\mid\posterior)
    +
    \noise(\action)$.
Under the standard assumption that 
the action-specific random shock 
$\{\noise(\action)\}_{\action\in\cA}$ 
is drawn i.i.d.\ from 
the Type I extreme value distribution,\footnote{The 
cumulative function of the Type I extreme value distribution
is $G(\noise) = \exp(-\exp(-\noise))$.}
the probability that action $\action\in\cA$ is selected
over the randomness of 
$\{\noise(\action)\}_{\action\in\cA}$ 
is exactly 
${\exp(\noiseScale\cdot \receiverU(\action\mid\posterior))}
/{\left(
    \exp(\noiseScale\cdot \receiverU(0\mid\posterior)) + 
    \exp(\noiseScale\cdot \receiverU(1\mid\posterior))
    \right)}$.

Recall $\receiverU_i\triangleq \receiverU_i(0) - 
\receiverU_i(1)$.
Let 
$\scaledNoiseDiff\triangleq 
\sfrac{(\noise(1) -\noise(0))}{\noiseScale}$.
By definition, the cumulative function and density function 
of random variable $\scaledNoiseDiff$
is $F(\scaledNoiseDiff) \triangleq
\sfrac{\exp(\noiseScale\scaledNoiseDiff)}{
(1 + 
\exp(\noiseScale\scaledNoiseDiff))}$
and 
$f(\scaledNoiseDiff) \triangleq
\sfrac{\noiseScale\exp(\noiseScale\scaledNoiseDiff)}{
(1 + 
\exp(\noiseScale\scaledNoiseDiff))^2
}$,
respectively.
Given posterior belief $\posterior$,
the receiver with the action-specific random shock 
$\{\noise(\action)\}_{\action\in\cA}$ 
takes action 1 if and only if 
    $\sum_{i\in[m]} \posterior_i\cdot \left(
    \scaledNoiseDiff - 
    \receiverU_i 
    \right) \geq 0$.
Thus, our problem can 
be interpreted as the public persuasion problem 
\citep{X-20,DX-17}
for a continuum population of rational receivers.
Specifically, 
$f(\scaledNoiseDiff)$ fraction of receiver population 
is associated with type $\scaledNoiseDiff\in(-\infty,\infty)$,
who has 
utility $\scaledNoiseDiff - \receiverU_i$
for action 1 and utility zero for action 0
for each state $i\in[m]$.

The linear program~\ref{eq:opt lp}
can then be interpreted for 
the aforementioned public persuasion problem.
Specifically, variables $\{\signalscheme_i(\scaledNoiseDiff)\}$
specify a public signaling scheme where
each variable $\signalscheme_i(\scaledNoiseDiff)$
corresponds to the probability that 
receivers with type greater or equal to $\scaledNoiseDiff$
take action 1 
while receivers with type less than $\scaledNoiseDiff$
take action 0. The first (resp.\ second) constraint
in \ref{eq:opt lp}
guarantees the persuasiveness (resp.\ feasibility)
of the public signaling scheme.

%% file: Paper/apx-prelim.tex
In this section, we 
present the omitted proof
of \Cref{prop:opt lp}
in \Cref{sec:prelim}.
\optimallp*
\begin{proof}
Fix an arbitrary feasible solution 
$\{\signalscheme_i(\scaledNoiseDiff)\}$
in program~\ref{eq:opt lp}.
We construct a signaling scheme $\signalscheme\primed$ 
as follows.\footnote{Here we use the superscript $\dagger$
to denote the constructed signaling scheme.}
Let the signal space $\signalspace\primed \gets  
\{\scaledNoiseDiff:\exists i\in[m],
\signalscheme_i(\scaledNoiseDiff) > 0\}$.
For each realized state $i$, let $\signalscheme\primed_i(\scaledNoiseDiff)
\gets \signalscheme_i(\scaledNoiseDiff)$ 
for each $\scaledNoiseDiff\in \signalspace\primed$.
Due to the second constraint 
in program~\ref{eq:opt lp},
the constructed signaling scheme 
$\signalscheme\primed$ is valid.
When signal $\scaledNoiseDiff\in\signalspace\primed$
is realized, the posterior belief 
$\posterior_i\primed(\scaledNoiseDiff)$
equals 
$\frac{
\prior_i\signalscheme_i(\scaledNoiseDiff)}
{\sum_{j\in[m]}\prior_j \signalscheme_j(\scaledNoiseDiff)}$.
Due to the first constraint 
in program~\ref{eq:opt lp},
we know that 
$\sum_{i\in[m]}\posterior_i\primed(\scaledNoiseDiff)\cdot 
(\receiverU_i - \scaledNoiseDiff) = 0$,
which implies $\scaledNoiseDiff = 
\sum_{i\in[m]}
\posterior_i\primed(\scaledNoiseDiff)\cdot 
\receiverU_i$.
Thus,
given realized signal $\scaledNoiseDiff$,
the receiver takes action 1 with probability 
$
\Quant(\scaledNoiseDiff)$.

So far, we have shown that the sender's optimal expected utility
is weakly higher than the optimal objective value of program~\ref{eq:opt lp}.
We finish our proof by converting 
the optimal signaling scheme $\signalscheme^*$
(with signal space $\signalspace^*$)
into a feasible solution 
$\{\signalscheme_i\doubleprimed(\scaledNoiseDiff)\}$
of program~\ref{eq:opt lp}
whose objective value equals to the sender's optimal expected utility 
in $\signalscheme^*$.\footnote{Here we use 
the superscript $\ddagger$
to denote the constructed feasible solution.}
The construction works as follow.
First, we initialize $\signalscheme\doubleprimed_i(\scaledNoiseDiff)\gets 0$
for all $i\in[m]$, $\scaledNoiseDiff\in(-\infty,\infty)$.
Next, we enumerate each signal $\signal \in \signalspace^*$,
let $\posterior^*(\signal)$ be the induced posterior belief and 
$\scaledNoiseDiff \triangleq 
\sum_{i\in[m]}\posterior^*_i(\signal)\receiverU_i$.
By definition, given posterior belief $\posterior^*(\signal)$,
the boundedly rational receiver takes action 1 
with probability $\Quant(\scaledNoiseDiff)$.
We update $\signalscheme\doubleprimed_i(\scaledNoiseDiff) \gets 
\signalscheme\doubleprimed_i(\scaledNoiseDiff) + \signalscheme^*_i(\signal)
$
for every state $i\in[m]$.
After enumerating every signal $\signal \in\signalspace^*$,
it can be verified that the constructed solution 
$\signalscheme\doubleprimed$ is feasible 
and its objective value equals the sender's expected utility.
\end{proof}

%% file: Paper/apx-independent.tex
In this section,
we present the omitted proofs in \Cref{sec:state independent}.

\ifEC
\subsection{Interpreting 
Program~\ref{eq:SISU opt condition}
for Optimal Censorship of Fully Rational Receiver.}
\label{apx:feasibility for fully rational}
The feasibility program~\ref{eq:SISU opt condition}
recovers the structure of the optimal censorship
for a fully rational receiver
in SISU environments.
For a fully rational receiver 
(whose bounded rationality level $\noiseScale=\infty$),
function $\Quant(\cdot)$ becomes $\Quant(x) = \indicator{x \leq 0}$.
In this case, 
there is no longer a bijection between 
$\censorshipsignalInverse\in[0,\infty)$
and $\censorshipsignal\in(-\infty,0]$
satisfying constraint \dfone.
Instead, the feasible solutions of constraint
\dfone\ admit one of the two forms:
either (i) $\{\censorshipsignalInverse \in [0,\infty), \censorshipsignal = 0\}$;
or (ii) $\{\censorshipsignalInverse = \infty, 
\censorshipsignal\in(-\infty, 0]\}$.
Note that (i) $\{\censorshipsignalInverse \in [0,\infty), \censorshipsignal = 0\}$
corresponds to instances where $\sum_{i\in[m]}\prior_i\receiverU_i < 0$ and 
thus the optimal censorship 
in \Cref{lem:opt fully rational}
selects the threshold state and 
the threshold state probability such that 
$\censorshipsignal=0$, i.e.,
the fully rational receiver is indifferent between
action 0 and action 1 
when the pooling
signal $\censorshipsignal$ is realized.
On the other side, 
(ii) $\{\censorshipsignalInverse = \infty, 
\censorshipsignal\in(-\infty, 0]\}$
corresponds to instances where $\sum_{i\in[m]}\prior_i\receiverU_i \geq 0$ and 
thus the optimal censorship
in \Cref{lem:opt fully rational}
sets the threshold state $\censorshipstate = \argmax_{i}\{
\receiverU_i:\receiverU_i\leq \censorshipsignalInverse\} = m$,
i.e.,
pools all state together
and reveals no information.
\fi

\subsection{Omitted Proof of 
\texorpdfstring{\Cref{prop:SISU threshold state monotonicity}}{}}
\label{apx:SISU proof SISU threshold state monotonicity}
Below we present the omitted proof of 
\Cref{prop:SISU threshold state monotonicity}.
To simplify the analysis, we first introduce the following definition.
For any $\scaledNoiseDiff\in[0, +\infty)$, we define
$\tangentIntersect(\scaledNoiseDiff)\in(- \infty, 0]$ such that 
\begin{align}
    \label{eq: tangent line}
    W'(\tangentIntersect(\scaledNoiseDiff)) 
    = \frac{ W(\scaledNoiseDiff)  - W(\tangentIntersect(\scaledNoiseDiff))}{\scaledNoiseDiff - \tangentIntersect(\scaledNoiseDiff)}~.
\end{align}
Clearly, we have $\censorshipsignal = \tangentIntersect(\censorshipsignalInverse)$
where $\censorshipsignal$ and $\censorshipsignalInverse$ are defined
in \Cref{thm:opt state independent}.
By the curvature of the function $W$, namely, $W$ is 
concave over $(-\infty, 0]$ and convex over $[0, +\infty)$, 
we have the following property about 
$\tangentIntersect(\cdot)$:
\begin{lemma}
\label{lem:kappa property}
$\tangentIntersect(\cdot)$ is a bijection
function from $[0, +\infty)$ to $ (- \infty, 0]$,
i.e, for any $\scaledNoiseDiff \in[0, +\infty)$, 
there exists a unique $\tangentIntersect(\scaledNoiseDiff) \in (- \infty, 0]$ 
that \eqref{eq: tangent line} holds.
Moreover, $\tangentIntersect(\cdot)$ is decreasing 
as $\scaledNoiseDiff \in[0, +\infty)$ increases. 
\end{lemma}
\begin{proof}
Recall that $W(x) = \frac{1}{1+\exp(\noiseScale x)}$,
and $W'(x)
= -\noiseScale \cdot \frac{\exp(\noiseScale x)}{(1+\exp(\noiseScale x))^2}$,
from \eqref{eq: tangent line}, for a fixed $\scaledNoiseDiff \ge 0$,
$\tangentIntersect(\scaledNoiseDiff)$
is the root of the following function
\begin{align*}
    f(x, \scaledNoiseDiff) \triangleq 
    -\noiseScale \cdot \frac{\exp(\noiseScale x)}{(1+\exp(\noiseScale x))^2} \cdot(x - \scaledNoiseDiff) - \frac{1}{1+\exp(\noiseScale x)} + \frac{1}{1+\exp(\noiseScale \scaledNoiseDiff)}~.
\end{align*}
Inspecting its first-order partial derivatives, we can see that
$\frac{f(x, \scaledNoiseDiff)}{\partial x}
> 0, \forall x < 0; 
\left.\frac{f(x, \scaledNoiseDiff)}{\partial \scaledNoiseDiff}\right|_{x = \tangentIntersect(\scaledNoiseDiff)} > 0$.
As a consequence, given $\scaledNoiseDiff_2 > \scaledNoiseDiff_1 > 0$,
we have
$
0 = f\left(\tangentIntersect(\scaledNoiseDiff_1), \scaledNoiseDiff_1\right)
< 
f\left(\tangentIntersect(\scaledNoiseDiff_1), \scaledNoiseDiff_2\right)$.
From $f\left(\tangentIntersect(\scaledNoiseDiff_2), 
\scaledNoiseDiff_2\right) = 0$, we know $\tangentIntersect(\scaledNoiseDiff_2) < \tangentIntersect(\scaledNoiseDiff_1)$, which proves the statement.
\end{proof}
With the above definition \eqref{eq: tangent line},
we also define
\begin{align}
    \label{eq: pooled prob defn}
    p_i \triangleq
    -\frac{\sum_{j:j<i} \prior_j  \cdot (\receiverU_j - \tangentIntersect(\receiverU_i) )}{\prior_i 
    \cdot( \receiverU_i  - \tangentIntersect(\receiverU_i))}~.
\end{align}

\begin{lemma}
\label{lem: unique valid p_i}
For any state $i\in[m]$ with $\receiverU_i \ge  0$, 
if the corresponding $p_i < 0$, then it must have
$p_j < 0, \forall j > i$.
\end{lemma}
\begin{proof}
Consider following two states $i_1, i_2$ where 
$i_1 < i_2, \receiverU_{i_1} \ge 0$
\begin{align*}
    p_{i_1} = -\frac{\sum_{j: j < i_1} \prior_j  (\receiverU_j - \tangentIntersect(\receiverU_{i_1}))}{\prior_{i_1} \cdot(\receiverU_{i_1} - \tangentIntersect(\receiverU_{i_1}))}, \quad
    p_{i_2} = -\frac{\sum_{j: j < i_2} \prior_j (\receiverU_j - \tangentIntersect(\receiverU_{i_2}))}{\prior_{i_2} \cdot(\receiverU_{i_2} - \tangentIntersect(\receiverU_{i_2}))}~.
\end{align*}
Suppose $p_{i_1} < 0$. 
Since $\receiverU_{i_1} - \tangentIntersect(\receiverU_{i_1}) > 0$, it must imply 
that $\sum_{j: j < i_1} \prior_j  (\receiverU_j - \tangentIntersect(\receiverU_{i_1})) > 0$.
Observe that 
\begin{align*}
    \sum_{j: j < i_2} \prior_j (\receiverU_j - \tangentIntersect(\receiverU_{i_2}))
    = \sum_{j: j < i_1} \prior_j (\receiverU_j - \tangentIntersect(\receiverU_{i_2}))
    + \sum_{j: i_1 \le j < i_2} \prior_j (\receiverU_j - \tangentIntersect(\receiverU_{i_2}))
    > \sum_{j: i_1 \le j < i_2} \prior_j (\receiverU_j - \tangentIntersect(\receiverU_{i_2})) > 0~,
\end{align*}
where the last inequality follows from 
the fact that $\receiverU_j > 0, \forall i_1 \le j < i_2$
and $\tangentIntersect(\receiverU_{i_2}) < 0$.
Hence, with the fact that $\receiverU_{i_2} - \tangentIntersect(\receiverU_{i_2}) > 0$, 
one must have $p_{i_2} < 0$.
\end{proof}

We are now ready to prove \Cref{prop:SISU threshold state monotonicity}. 
\sisustatemonotone*
\begin{proof}
We begin the analysis 
with showing the following observation: 
Fix a $\scaledNoiseDiff \in [0, +\infty)$, 
let $\tangentIntersect(\scaledNoiseDiff)$ (resp. $\hat{\tangentIntersect}(\scaledNoiseDiff)$)
be the value that satisfies \eqref{eq: tangent line}
for the bounded rationality level $\noiseScale$ 
(resp. $\hat{\noiseScale}$).
Then we have 
\begin{align}
    \tangentIntersect(\scaledNoiseDiff)
    \le 
    \hat{\tangentIntersect}(\scaledNoiseDiff)
    < 0, 
    \quad
    \text{ if }\noiseScale \le \hat{\noiseScale}~.
    \label{eq: larger beta larger tangent}
\end{align}
To see this, 
recall that $W(x) = \frac{1}{1+\exp(\noiseScale x)}$,
and $\frac{\partial W(x)}{\partial x}
= -\noiseScale \cdot \frac{\exp(\noiseScale x)}{(1+\exp(\noiseScale x))^2}$,
from \eqref{eq: tangent line}, 
$\tangentIntersect(\scaledNoiseDiff)$
is the root of the following function
\begin{align*}
    f(x, \noiseScale) \triangleq 
    -\noiseScale \cdot \frac{\exp(\noiseScale x)}{(1+\exp(\noiseScale x))^2} \cdot(x - \scaledNoiseDiff) - \frac{1}{1+\exp(\noiseScale x)} + \frac{1}{1+\exp(\noiseScale \scaledNoiseDiff)}~.
\end{align*}
Inspecting its first-order partial derivatives, we can see that
$\frac{f(x, \noiseScale)}{\partial x}
\ge 0, \forall x\le 0;
\left.\frac{f(x, \noiseScale)}{\partial \noiseScale}\right|_{x = \tangentIntersect(\scaledNoiseDiff)} \le 0$.
As a consequence, given $\hat{\noiseScale} \ge \noiseScale$,
we have
$f\left(\tangentIntersect(\scaledNoiseDiff), \hat{\noiseScale}\right)
\le 
f\left(\tangentIntersect(\scaledNoiseDiff), \noiseScale\right)
= 0$.
From $f\left(\hat{\tangentIntersect}(\scaledNoiseDiff), 
\hat{\noiseScale}\right) = 0$, we know $\hat{\tangentIntersect}(\scaledNoiseDiff) \ge \tangentIntersect(\scaledNoiseDiff)$.

Now given a state $i$ where $\receiverU_i \ge 0$,
consider the 
bounded rationality level $\noiseScale, \hat{\noiseScale}$,
from \eqref{eq: larger beta larger tangent},
we have
$\tangentIntersect(\receiverU_i)
\le 
\hat{\tangentIntersect}(\receiverU_i)$, 
implying
\begin{align}
    & \sum_{j: j < i} \prior_j \cdot \left(\receiverU_j - \tangentIntersect(\receiverU_i)\right)
    \ge
    \sum_{j: j < i} \prior_j \cdot \left(\receiverU_j - \hat{\tangentIntersect}(\receiverU_i)\right) 
    \text{ and }
    0 > 
    - \frac{1}{\prior_i\cdot(\receiverU_i - \tangentIntersect(\receiverU_i))} 
    \ge - \frac{1}{\prior_i\cdot(\receiverU_i - \hat{\tangentIntersect}(\receiverU_i))}; \nonumber \\
    \Rightarrow ~
    &
    p_i
    = -\frac{\sum_{j: j < i} \prior_j \cdot \left(\receiverU_j - \tangentIntersect(\receiverU_i)\right)}{\prior_i\cdot(\receiverU_i - \tangentIntersect(\receiverU_i))}
    \le 
    -\frac{\sum_{j: j < i} \prior_j \cdot \left(\receiverU_j - \hat{\tangentIntersect}(\receiverU_i)\right)}{\prior_i\cdot(\receiverU_i - \hat{\tangentIntersect}(\receiverU_i))}
    = \hat{p}_i~.
    \label{ineq: p_i monotoncity}
\end{align}
The above inequality ensures that 
the threshold state $\censorshipstateHat$ for a larger
bounded rationality level $\hat{\noiseScale}$
is no smaller than the threshold state $\censorshipstate$ for a smaller
bounded rationality level $\noiseScale$. 
If $\censorshipstateHat = \censorshipstate$, one still has
$\censorshipprob = \min\{p_{\censorshipstate}, 1\}
\le \censorshipprobHat = \min\{p_{\censorshipstateHat}, 1\}$.
\end{proof}

\ifEC
\else
\fi

\subsection{Omitted Proof of 
\texorpdfstring{\Cref{lem:existence SISU opt condition}}{}}
\label{apx:SISU proof opt condition existence}
Here we present the omitted proof of 
\Cref{lem:existence SISU opt condition}.
With the two properties (see \Cref{lem:kappa property} and \Cref{lem: unique valid p_i}) 
for $\tangentIntersect(\cdot)$ and $p_i$
we established in \Cref{apx:SISU proof SISU threshold state monotonicity}, we 
prove \Cref{lem:existence SISU opt condition} as follows:
\sisuoptcondition*
\begin{proof}
When the set $\{i\in[m]: \receiverU_i \ge 0\}$ is not empty,
then from the definition \eqref{eq: tangent line}
and the definition \eqref{eq: pooled prob defn}, we know
that every feasible solution
to the program \ref{eq:SISU opt condition} 
must be that  
the threshold state 
$\censorshipstate = \argmax_{i\in[m]: \receiverU_i \ge 0} \{p_i: p_i\ge0\}$ 
and $\censorshipprob = \min\{p_{\censorshipstate}, 1\}$. 
From \Cref{lem: unique valid p_i}, 
we know that such threshold state
and the pooling probability is unique. 
On the other hand,
if the set $\{i\in[m]: \receiverU_i \ge 0\}$
is empty, then \Cref{thm:opt state independent}, 
together with the definition \eqref{eq: tangent line}
and the definition \eqref{eq: pooled prob defn},
say that 
$\censorshipstate 
= \argmax_{i\in[m]: \receiverU_i \le 0} \{\receiverU_i\}$
and $\censorshipprob = 1$, 
which also guarantees uniqueness of the feasible 
solution to the program \ref{eq:SISU opt condition}. 
\end{proof}

\subsection{Omitted Proof of \texorpdfstring{\Cref{lem: sisu dual feasibility}}{}}
\label{apx:SISU proof SISU dual feasibility}
Here we present the omitted proof of \Cref{lem: sisu dual feasibility}.
\begin{figure}[H]
\centering
\subfloat[]{
\input{plots/figure-SISU-opt-dual-feasibility-1-resize}
\label{fig:case a}
}
\quad
\subfloat[]{
\input{plots/figure-SISU-opt-dual-feasibility-2-resize}
\label{fig:case b}
}
\caption{\label{fig:SISU dual feasibility}
Graphical illustration of 
the dual assignment feasibility 
for inequality~\eqref{eq:SISU dual feasibility pooling state}
in the proof of \Cref{thm:SISU opt}.
The gray solid curve is function $\Quant(\cdot)$.
(a):
For $\scaledNoiseDiff\in 
(-\infty,\censorshipsignalInverse\text{]}$, 
the dual assignment $\noiseDualVar(\scaledNoiseDiff)$
is the absolute value of the slope 
of the black solid line.
We prove that inequality~\eqref{eq:SISU dual feasibility pooling state} holds by 
showing that
the slope of the black dotted (resp.\ dashed)
line 
is larger (resp.\ smaller) than 
the slope of the black solid line 
if $\scaledNoiseDiff\in (-\infty,\censorshipsignal\text{]}$
(resp.\ 
$\scaledNoiseDiff\in\text{[}\censorshipsignal,\censorshipsignalInverse\text{]}$).
(b):
For $\scaledNoiseDiff\in 
\text{[}\censorshipsignalInverse,
\receiverU_{\censorshipstate+1}\text{]}$, 
the dual assignment $\noiseDualVar(\scaledNoiseDiff)$
is the absolute value of the slope 
of the black dotted line.
We prove that inequality~\eqref{eq:SISU dual feasibility pooling state} holds by 
showing that
the slope of the black dash-dotted 
line 
is smaller than 
the slope of the black dotted line.
More specifically, 
we rewrite inequality~\eqref{eq:SISU dual feasibility pooling state}
as inequality~\eqref{eq:SISU dual feasibility pooling state case 2}.
The right-hand side of inequality~\eqref{eq:SISU dual feasibility pooling state case 2}
is the absolute value of the slope of the black dash-dotted line.
We lower bound this term
by the absolute value of the slope of the black dashed line,
which is due to convexity of the function $\Quant(\cdot)$,
is larger than the absolute value of the slope of the black dotted line,
i.e., the left-hand side of inequality~\eqref{eq:SISU dual feasibility pooling state case 2}.
}
\end{figure}

\sisudualfeasibility*
\begin{proof}
We first present the analysis for the state $i\in[\censorshipstate]$
for the latter two cases: 
\begin{itemize}[label={-}]
        \item Fix an arbitrary $j\in
        [\censorshipstate + 1:
        m - 1]$
        and an arbitrary $\scaledNoiseDiff\in[\receiverU_j,
        \receiverU_{j + 1}]$.
        Similar to the previous case,
        after rearranging the terms,
        inequality~\eqref{eq:SISU dual feasibility pooling state} becomes 
        \begin{align*}
                   -\frac{
        \Quant(\scaledNoiseDiff) - \Quant(\receiverU_{j+1})}
        {\scaledNoiseDiff - 
        \receiverU_{j+1}}
        \leq 
        - \frac{ 
        \Quant(\scaledNoiseDiff) - 
        \left(\Quant(\censorshipsignal)
        -\Quant'(\censorshipsignal) 
        (\censorshipsignal - \receiverU_i) \right)}
        {\scaledNoiseDiff - \receiverU_i}
        \end{align*}
        Due to the convexity of function $\Quant(\cdot)$
        on $[0, \infty)$,
        \begin{align*}
            -\frac{
        \Quant(\scaledNoiseDiff) - \Quant(\receiverU_{j+1})}
        {\scaledNoiseDiff - 
        \receiverU_{j+1}}
        \leq 
        -\frac{
        \Quant(\scaledNoiseDiff) - \Quant(\receiverU_{
        \censorshipstate+1})}
        {\scaledNoiseDiff - 
        \receiverU_{\censorshipstate+1}}
        \end{align*}
        and thus the analysis in the previous case 
        can be carried over directly.
        
        \item Fix an arbitrary 
        $\scaledNoiseDiff\in[\receiverU_m, \infty)$.
        By construction $\noiseDualVar(\scaledNoiseDiff) 
        = 0$.
        Inequality~\eqref{eq:SISU dual feasibility pooling state}
        becomes 
        \begin{align}
        \label{eq:SISU dual feasibility pooling state case 4}
            \Quant(\scaledNoiseDiff) \leq 
      \Quant(\censorshipsignal) -
      (\censorshipsignal - \receiverU_i)
      \Quant'(\censorshipsignal) 
        \end{align}
        If $\receiverU_i \leq \censorshipsignal$,
        inequality~\eqref{eq:SISU dual feasibility pooling state case 4} holds since function $\Quant(\cdot)$
        is monotone decreasing.
        Otherwise, i.e., 
        if $\censorshipsignal\leq \receiverU_i \leq
        \censorshipsignalInverse$,
        inequality~\eqref{eq:SISU dual feasibility pooling state case 4} holds since
        \begin{align*}
            \Quant(\censorshipsignal) -
      (\censorshipsignal - \receiverU_i)
      \Quant'(\censorshipsignal) 
      &\geq 
            \Quant(\censorshipsignal) -
      (\censorshipsignal - \censorshipsignalInverse)
      \Quant'(\censorshipsignal) 
      \overset{(a)}{=}\Quant(\censorshipsignalInverse)
      \overset{(b)}{\geq} 
      \Quant(\scaledNoiseDiff)
        \end{align*}
        where inequality~(a) holds due to 
        constraint~\dfone\ in program~\ref{eq:SISU opt condition};
        and inequality~(b) holds due to the monotonicity
        of function $\Quant(\cdot)$.
\end{itemize}
Finally, we verify dual constraints
associated with $\optscheme_i(\scaledNoiseDiff)$
for state $i \in [\censorshipstate+1:m]$. 
\begin{itemize}[label={-}]
\item Fix an arbitrary $\scaledNoiseDiff
\in(-\infty,\censorshipsignalInverse]$.
By construction, $\noiseDualVar(\scaledNoiseDiff) 
= - \Quant(\censorshipsignal)$.
By rearranging the terms, inequality~\eqref{eq:SISU dual feasibility revealing state} becomes 
\begin{align*}
    \Quant'(\censorshipsignal) \leq 
    \frac{
    \Quant(\receiverU_i) - \Quant(\scaledNoiseDiff)}{
    \receiverU_i - \scaledNoiseDiff}
\end{align*}
which holds since
\begin{align*}
    \frac{
    \Quant(\receiverU_i) - \Quant(\scaledNoiseDiff)}{
    \receiverU_i - \scaledNoiseDiff}
    \overset{(a)}{\geq}
    \frac{
    \Quant(\censorshipsignalInverse) - \Quant(\scaledNoiseDiff)}{
    \censorshipsignalInverse - \scaledNoiseDiff}
    \overset{(b)}{\geq}
    \frac{
    \Quant(\censorshipsignalInverse) - \Quant(\censorshipsignal)}{
    \censorshipsignalInverse - \censorshipsignal}
    \overset{(c)}{=}
    \Quant'(\censorshipsignal)
\end{align*}
where inequality~(a) holds due to 
the convexity of function $\Quant(\cdot)$ on $[0,\infty)$
and $\receiverU_i > \censorshipsignalInverse$;
and 
inequality~(b) and equality~(c) hold due to 
the concavity of function $\Quant(\cdot)$ on $(-\infty, 0]$
and constraint~\dfone\ in program~\ref{eq:SISU opt condition}.

\item Fix an arbitrary $\scaledNoiseDiff\in
[\censorshipsignalInverse, \receiverU_{\censorshipstate+1}]$.
By construction, $\noiseDualVar(\scaledNoiseDiff) 
= - (\Quant(\scaledNoiseDiff) - \Quant(\receiverU_{\censorshipstate + 1})
/
(\scaledNoiseDiff - \receiverU_{\censorshipstate + 1})$.
By rearranging the terms, inequality~\eqref{eq:SISU dual feasibility revealing state} becomes 
\begin{align*}
    \frac{
    \Quant(\receiverU_{\censorshipstate + 1}) - \Quant(\scaledNoiseDiff)}{
    \receiverU_{\censorshipstate + 1} - \scaledNoiseDiff}
     \leq 
    \frac{
    \Quant(\receiverU_i) - \Quant(\scaledNoiseDiff)}{
    \receiverU_i - \scaledNoiseDiff}
\end{align*}
which holds due to the convexity of function $\Quant(\cdot)$
on $[0,\infty)$ and 
$\scaledNoiseDiff \leq\receiverU_{\censorshipstate + 1}
\leq \receiverU_i$.

\item Fix an arbitrary $j\in
        [\censorshipstate + 1:
        m - 1]$
        and an arbitrary $\scaledNoiseDiff\in[\receiverU_j,
        \receiverU_{j + 1}]$.
By construction, $\noiseDualVar(\scaledNoiseDiff) 
= - (\Quant(\scaledNoiseDiff) - \Quant(\receiverU_{j + 1})
/
(\scaledNoiseDiff - \receiverU_{j + 1})$.
For state $i\in[\censorshipstate + 1:j]$,
by rearranging the terms, 
inequality~\eqref{eq:SISU dual feasibility revealing state} becomes 
\begin{align*}
    \frac{
    \Quant(\receiverU_{j + 1}) - \Quant(\scaledNoiseDiff)}{
    \receiverU_{j + 1} - \scaledNoiseDiff}
     \geq 
    \frac{
    \Quant(\receiverU_i) - \Quant(\scaledNoiseDiff)}{
    \receiverU_i - \scaledNoiseDiff}
\end{align*}
which holds due to the convexity of function $\Quant(\cdot)$
on $[0,\infty)$ and 
$\receiverU_i \leq \scaledNoiseDiff 
\leq \receiverU_{j+1}$.
Similarly, for state $i\in[j+1:m]$,
by rearranging the terms, 
inequality~\eqref{eq:SISU dual feasibility revealing state} becomes 
\begin{align*}
    \frac{
    \Quant(\receiverU_{j + 1}) - \Quant(\scaledNoiseDiff)}{
    \receiverU_{j + 1} - \scaledNoiseDiff}
     \leq 
    \frac{
    \Quant(\receiverU_i) - \Quant(\scaledNoiseDiff)}{
    \receiverU_i - \scaledNoiseDiff}
\end{align*}
which holds due to the convexity of function $\Quant(\cdot)$
on $[0,\infty)$ and 
$  \scaledNoiseDiff 
\leq \receiverU_{j+1}\leq \receiverU_i$.

\item Fix an arbitrary $\scaledNoiseDiff\in[\receiverU_m,\infty)$.
By construction, $\noiseDualVar(\scaledNoiseDiff) = 0$.
Here inequality~\eqref{eq:SISU dual feasibility revealing state} holds 
by the monotonicity of function $\Quant(\cdot)$ straightforwardly.
\qedhere
\end{itemize}
\end{proof}

\subsection{Omitted Proof of \texorpdfstring{\Cref{thm:SISU direct}}{}}
\label{apx:SISU direct proof}
Here we present the omitted proof of \Cref{thm:SISU direct}.

\sisudirectopt*
\begin{example}
\label{example:SISU direct lower bound example}
Given an arbitrary $m\in \naturals_+$,
consider a problem instance as follows:
There are $m$ states.
The receiver has 
bounded rationality level
$\noiseScale$ such that $\noiseScale \geq \exp(m)$.
The sender utility $\{\senderU_i\}$,
the receiver utility difference 
$\{\receiverU_i\}$, and prior $\{\prior_i\}$
over state space $[m]$ are 
\begin{align*}
     i\in[m]:\qquad 
    &\senderU_i = 1,
    \qquad 
    \receiverU_i = i,
    \qquad 
    \prior_i = 
    K\left(
    \exp(\noiseScale i) + 1
    \right)
\end{align*}
where $K  \triangleq 
\sfrac{1}{
(m + \sum_{j\in[m]} \exp(\noiseScale j))}$.
\end{example}
\begin{proof}
First, we lowerbound 
the expected sender utility 
in the optimal signaling scheme $\optscheme$
by computing the expected sender utility 
in the full-information revealing signaling scheme,
\begin{align*}
    \Payoff{\optscheme}
    \geq 
    \sum_{i = 1} \prior_i \Quant(\receiverU_i)
    = 
    m\cdot K
\end{align*}
Next, we upperbound the expected sender utility 
in the optimal direct signaling scheme $\hat\signalscheme$.
Suppose the optimal direct signaling scheme $\hat\signalscheme$
partitions the state space into $\highstates\sqcup\{\censorshipstate\}\sqcup\lowstates$.
Due to the convexity of function $\Quant(\cdot)$
on $[0,\infty)$,
the expected sender utility in the optimal direct signaling scheme $\hat\signalscheme$ is upperbounded by 
the expected sender utility
in signaling scheme $\tilde\signalscheme$
defined as follows,
\begin{align*}
    i\in\highstates:&\qquad
    \tilde\signalscheme_i(\scaledNoiseDiff) =
    \indicator{\scaledNoiseDiff = \frac{\sum_{i\in\highstates}
\prior_i \receiverU_i
}
{\sum_{i\in\highstates}
\prior_i
}}
    \\
    &\qquad
    \tilde\signalscheme_{\censorshipstate}(\scaledNoiseDiff)
    = \indicator{\scaledNoiseDiff = \receiverU_{\censorshipstate}}
    \\
    i\in\lowstates:&\qquad
    \tilde\signalscheme_i(\scaledNoiseDiff) =
    \indicator{\scaledNoiseDiff = \frac{\sum_{i\in\lowstates}
\prior_i \receiverU_i
}
{\sum_{i\in\lowstates}
\prior_i
}}
\end{align*}
Let $k_h = \max \highstates$ and $k_l = \max\lowstates$. We have 
\begin{align*}
    \Payoff{\hat\signalscheme} &\leq 
    \Payoff{\tilde\signalscheme} \\
    &=
    \left(\sum_{i\in\highstates}\prior_i\right) \Quant\left(
    \frac{\sum_{i\in\highstates}
\prior_j \receiverU_j
}
{\sum_{j\in\highstates}
\prior_j
}
    \right)
    +
    \prior_{\censorshipstate}\Quant\left(\receiverU_{\censorshipstate}\right)
    +
    \left(\sum_{i\in\lowstates}\prior_i\right) \Quant\left(
    \frac{\sum_{j\in\lowstates}
\prior_j \receiverU_j
}
{\sum_{j\in\lowstates}
\prior_j
}
    \right)
    \\
    &\leq 
    2\prior_{k_h} \Quant\left(
\receiverU_{k_h} - \frac{1}{\noiseScale}
    \right)
    +
    \prior_{\censorshipstate}\Quant\left(\receiverU_{\censorshipstate}\right)
    +
    2\prior_{k_l} \Quant\left(
\receiverU_{k_l} - \frac{1}{\noiseScale}
    \right)
    \\
    &\leq (4e + 1) \cdot K
\end{align*}
where the second inequality holds since $\noiseScale \geq \exp(m)$.

Finally, combining the lower bound (i.e., $m\cdot K$) of $\Payoff{\optscheme}$ and 
the upper bound (i.e., $(4e + 1)\cdot K$) of $\Payoff{\hat\signalscheme}$ finishes the proof.
\end{proof}

%% file: plots/figure-SISU-opt-dual-feasibility-1-resize.tex
\begin{tikzpicture}[scale=0.55, transform shape]
\begin{axis}[
axis line style=gray,
axis lines=middle,
xtick style={draw=none},
ytick style={draw=none},
xticklabels=\empty,
yticklabels=\empty,
xmin=-9,xmax=9,ymin=-0.15,ymax=1.3,
width=0.9\textwidth,
height=0.5\textwidth,
samples=50]

\addplot[domain=-9:9, gray!40!white, line width=2.mm] (x, {1/(1+e^(0.7*x))});

\addplot[mark=*,only marks, fill=white] coordinates {(6.34907,0.0116084)} node[above, pos=1]{};
\addplot[gray, thick] coordinates {(6.34907,0.01)
(6.34907,-0.01)};
\addplot[] coordinates {(6.34907,0.)} node[below, pos=1]{\Large$\censorshipsignalInverse$};

\addplot[mark=*,only marks, fill=white] coordinates {(-2.34798,0.838022)} node[above, pos=1]{};
\addplot[gray, thick] coordinates {(-2.34798,0.01)
(-2.34798,-0.01)};

\addplot[gray, dotted] coordinates {(-2.34798,0.838022)
(-2.34798,0.)};
\addplot[] coordinates {(-2.34798,0.)} node[below, pos=1]{\Large$\censorshipsignal$};

\addplot[mark=*,only marks, fill=white] coordinates {(-8,0.99631576)} node[above, pos=1]{};
\addplot[gray, thick] coordinates {(-8,0.01)
(-8,-0.01)};
\addplot[] coordinates {(-8,-0.015)} node[below, pos=1]{\Large$\scaledNoiseDiff$};
\addplot[gray, dotted] coordinates {(-8,0.99631576)
(-8,0.)};

\addplot[mark=*,only marks, fill=white] coordinates {(4,0.0573242)} node[above, pos=1]{};
\addplot[gray, thick] coordinates {(4,0.01)
(4,-0.01)};
\addplot[] coordinates {(4,-0.015)} node[below, pos=1]{\Large$\scaledNoiseDiff$};
\addplot[gray, dotted] coordinates {(4,0.0573242)
(4,0.)};

\addplot[domain=-2.34798:6.34907, black,line width=0.5mm] (x, {-0.09502 * x + 0.614897});

\addplot[domain=-8:-2.34798, dotted, black,line width=0.5mm] (x, {-0.0280066 * x + 0.772263});

\addplot[domain=-2.34798:4, dash pattern={on 8pt off 2pt}, black,line width=0.5mm] (x, {-0.122984 * x + 0.549259});

\end{axis}

\end{tikzpicture}

%% file: plots/figure-SISU-opt-dual-feasibility-2-resize.tex
\begin{tikzpicture}[scale=0.55, transform shape]
\begin{axis}[
axis line style=gray,
axis lines=middle,
xtick style={draw=none},
ytick style={draw=none},
xticklabels=\empty,
yticklabels=\empty,
xmin=-6,xmax=12,ymin=-0.15,ymax=1.3,
width=0.9\textwidth,
height=0.5\textwidth,
samples=50]

\addplot[domain=-6:12, gray!40!white, line width=2.mm] (x, {1/(1+e^(0.7*x))});

\addplot[mark=*,only marks, fill=white] coordinates {(6.34907,0.0116084)} node[above, pos=1]{};
\addplot[gray, thick] coordinates {(6.34907,0.01)
(6.34907,-0.01)};
\addplot[] coordinates {(6.34907,0.)} node[below, pos=1]{\Large$\censorshipsignalInverse$};

\addplot[mark=*,only marks, fill=white] coordinates {(-2.34798,0.838022)} node[above, pos=1]{};
\addplot[gray, thick] coordinates {(-2.34798,0.01)
(-2.34798,-0.01)};
\addplot[] coordinates {(-2.34798,0.)} node[below, pos=1]{\Large$\censorshipsignal$};

\addplot[gray, dotted] coordinates {(-2.34798,0.838022)
(-2.34798,0.)};

\addplot[mark=*,only marks, fill=white] coordinates {(11,0.000452622)} node[above, pos=1]{};
\addplot[gray, thick] coordinates {(11,0.01)
(11,-0.01)};
\addplot[] coordinates {(11,-0.035)} node[below, pos=1]{\Large$\receiverU_{\censorshipstate+1}$};

\addplot[mark=*,only marks, fill=white] coordinates {(9,0.00183294)} node[above, pos=1]{};
\addplot[gray, thick] coordinates {(9,0.01)
(9,-0.01)};
\addplot[] coordinates {(9,-0.015)} node[below, pos=1]{\Large$\scaledNoiseDiff$};
\addplot[gray, dotted] coordinates {(9,0.00183294)
(9,0.)};

\addplot[mark=*,only marks, fill=white] coordinates {(1.2,0.500873)} node[above, pos=1]{};
\addplot[gray, thick] coordinates {(1.2,0.01)
(1.2,-0.01)};
\addplot[] coordinates {(1.2,-0.035)} node[below, pos=1]{\Large$\receiverU_i$};
\addplot[gray, dotted] coordinates {(1.2,0.500873)
(1.2,0.)};

\addplot[domain=-2.34798:6.34907, black,line width=0.5mm] (x, {-0.09502 * x + 0.614897});

\addplot[domain=6.34907:9, dash pattern={on 8pt off 2pt}, black,line width=0.5mm] (x, {-0.00368756 * x + 0.035021});

\addplot[domain=9:11, dotted, black,line width=0.5mm] (x, {-0.000690159 * x + 0.00804437});

\addplot[domain=1.2:9, dash pattern={on 7pt off 2pt on 1pt off 3pt}, black,line width=0.5mm] (x, {-0.0639795  * x + 0.577648});

\end{axis}

\end{tikzpicture}

%% file: Paper/apx-dependent.tex
\subsection{Omitted Proof of \texorpdfstring{\Cref{lem:SDSU lower bound example opt}}{}}
\label{apx:SDSU proof lower bound example opt}
\begin{proof}
We prove the lemma statement by 
constructing a feasible signaling scheme 
$\signalscheme$ with $\Payoff{\signalscheme}
= \Theta(K_1 K_2m)$.
In particular, consider the following 
construction of signaling scheme $\signalscheme$:
\begin{align*}
    i\in[m - 1]:\qquad 
    \signalscheme_i
    \left(i + \frac{1}{\noiseScale}\right) 
    = 1;
    \qquad
    \signalscheme_m
    \left(i + \frac{1}{\noiseScale}\right)
    = 
    K_1\exp(\noiseScale i)
\end{align*}
It is straightforward to verify by algebra that 
signaling scheme $\signalscheme$ constructed above
is a feasible solution of program~\ref{eq:opt lp}.\footnote{Specifically, 
in signaling scheme $\signalscheme$,
each state $i\in[m - 1]$
is 
pooled fully (i.e., $\signalscheme_i(i + \sfrac{1}{\noiseScale}) = 1$) on signal $i + \sfrac{1}{\noiseScale}$
with the last state $m$
(with probability $\signalscheme_m(1+\sfrac{1}{\noiseScale}) = \sfrac{\exp(\noiseScale i)}{K_1}$).}
The expected sender utility $\Payoff{\signalscheme}$
of signaling scheme $\signalscheme$ is 
\begin{align*}
    \Payoff{\signalscheme} &= 
    \sum_{i \in [m - 1]} \prior_m \senderU_m 
    \signalscheme_m\left(i + \frac{1}{\noiseScale}\right)
    \Quant\left(i + \frac{1}{\noiseScale}\right)
    \\
    &= 
    \sum_{i\in[m - 1]}
    K_2 K_1 \exp(\noiseScale i) \frac{1}{1 + \exp(\noiseScale(i + \frac{1}{\noiseScale}))} 
    \\
    &=
    ~\Theta(K_1 K_2 m)
    \qedhere
\end{align*}
\end{proof}

\subsection{Omitted Proof of \texorpdfstring{\Cref{lem:single signal payoff lowerbound large}}{}}
\label{apx:SDSU proof single signal payoff lowerbound large}
\singlesignalpayofflowerboundlarge*
\begin{proof}
We first prove the case $\Payoff{\lowerbounddecomposescheme} 
= O({K_1K_ 2})$ for any $\decomposesignal\in
[i, i + \sfrac{m\log(\noiseScale)}{\noiseScale}]$.
Recall that in signaling scheme $\lowerbounddecomposescheme$,
state $i$ and state $m$ are pooled on signal $\decomposesignal$,
and all other states are fully revealed,
i.e.,
\begin{align*}
    \lowerbounddecomposescheme_i(\decomposesignal) = 1;
    \quad
    \lowerbounddecomposescheme_m(\decomposesignal) = 
    \frac{\prior_i}{\prior_m}\frac{
    \decomposesignal - i}{m - \decomposesignal};
    \quad
    \lowerbounddecomposescheme_m(m) = 1 - 
    \frac{\prior_i}{\prior_m}\frac{
    \decomposesignal - i}{m - \decomposesignal};
    \quad
    \lowerbounddecomposescheme_j(j) = 1 \forall j\neq i,m~.
\end{align*}
The expected sender utility $\Payoff{\lowerbounddecomposescheme}$
of signaling scheme $\signalscheme$ is 
\begin{align*}
    \Payoff{\lowerbounddecomposescheme} &=
    \prior_m\senderU_m
    \left(
    \lowerbounddecomposescheme_m(\decomposesignal)
    \Quant(\decomposesignal)
    +
    \lowerbounddecomposescheme_m(m)
    \Quant(m)
    \right)
    \intertext{where}
    \prior_m\senderU_m
    \lowerbounddecomposescheme_m(m)
    \Quant(m)
    &\leq
    K_2 \frac{1}{1 + \exp(\noiseScale m)}
    \overset{(a)}{=} o(K_1 K_2)
    \\
    \prior_m\senderU_m
    \lowerbounddecomposescheme_m(\decomposesignal)
    \Quant(\decomposesignal)
    &=K_2 
    K_1
    \left(m - i - \frac{1}{\noiseScale}\right)
    \noiseScale\exp(\noiseScale i)
    \frac{
    \decomposesignal - i
    }{
    m - \decomposesignal
    }
    \frac{1}{1 + \exp(\noiseScale\decomposesignal)}
    \\
    &=
    K_1 K_2
    \frac{m - i - \frac{1}{\noiseScale}}{m - \decomposesignal}
    \frac{\exp(\noiseScale i)}{1 + \exp(\noiseScale\decomposesignal)}
    \noiseScale(\decomposesignal - i)
    \overset{(b)}{=}
    O(K_1 K_2)
\end{align*}
Here equality~(a) holds since 
$\sfrac{1}{(1 + \exp(\noiseScale m))} = o(K_1)$;
equality~(b) uses two facts that (i)
$\sfrac{(m - i - 1/\noiseScale)}{(m - \decomposesignal)} = O(1)$
since $\sfrac{\noiseScale}{\log(\noiseScale)} \geq 2m$
and thus
$\decomposesignal\leq i + \sfrac{m\log(\noiseScale)}{\noiseScale}
\leq i + \sfrac{1}{2}$;
and (ii) $\frac{\exp(\noiseScale i)}{1 + \exp(\noiseScale\decomposesignal)}
    \noiseScale(\decomposesignal - i) 
    \leq \frac{\noiseScale(\decomposesignal - i)}{\exp(\noiseScale(\decomposesignal - i))} = O(1)$.

We now prove the case $\Payoff{\lowerbounddecomposescheme} 
= o(\sfrac{K_1K_ 2}{m})$ for any $\decomposesignal\in
[i + \sfrac{m\log(\noiseScale)}{\noiseScale}, m]$.
It is clear that for every $\decomposesignal \geq \NOinforSignal$,
the expected sender utility $\Payoff{\lowerbounddecomposescheme}
\leq \Payoff{\lowerbounddecomposeschemeavg}$.
Thus, it is sufficient to show $\Payoff{\lowerbounddecomposescheme} = 
o(\sfrac{K_1 K_2}{m})$ for every 
$\decomposesignal\in [
i + \sfrac{m\log(\noiseScale)}
{\noiseScale}, \NOinforSignal]$.
By definition, 
\begin{align*}
     \Payoff{\lowerbounddecomposescheme} &=
    \prior_m\senderU_m
    \left(
    \lowerbounddecomposescheme_m(\decomposesignal)
    \Quant(\decomposesignal)
    +
    \lowerbounddecomposescheme_m(m)
    \Quant(m)
    \right)
    \intertext{where}
    \prior_m\senderU_m
    \lowerbounddecomposescheme_m(m)
    \Quant(m)
    &\leq
    K_2 \frac{1}{1 + \exp(\noiseScale m)}
    \overset{(a)}{=} o\left(\frac{K_1 K_2}{m}\right)
    \end{align*}
    Here equality~(a) holds since 
$\sfrac{1}{(1 + \exp(\noiseScale m))} = o(\sfrac{K_1}{m})$.
It remains to show
term 
$    \prior_m\senderU_m
    \lowerbounddecomposescheme_m(\decomposesignal)
    \Quant(\decomposesignal) = o(\sfrac{K_1 K_2}{m})$. 
    We show this in two cases based on 
    the value of $\decomposesignal$.
    \begin{itemize}[label={-}]
        \item Fix an arbitrary $\decomposesignal \in
        [i + \sfrac{m\log(\noiseScale)}
{\noiseScale},
        m - \sfrac{1}{2}]$.
        Note that 
        \begin{align*}
        \prior_m\senderU_m
    \lowerbounddecomposescheme_m(\decomposesignal)
    \Quant(\decomposesignal) &=
    K_2 
    K_1
    \left(m - i - \frac{1}{\noiseScale}\right)
    \noiseScale\exp(\noiseScale i)
    \frac{
    \decomposesignal - i
    }{
    m - \decomposesignal
    }
    \frac{1}{1 + \exp(\noiseScale\decomposesignal)}
    \\
    &\overset{(a)}{\leq}
    K_1 K_2 m \noiseScale \frac{m}{2} \exp(\noiseScale( i - \decomposesignal))
    \\
    &\overset{(b)}{\leq}
    K_1 K_2 m \noiseScale \frac{m}{2} \exp\left(\noiseScale\left(i - 
    \left(i + \frac{m\log(\noiseScale)}{\noiseScale}\right)\right)\right)
    \\
    &= o\left(\frac{K_1K_2}{m}\right)
        \end{align*}
    where inequality~(a) holds since 
    $m - i - \sfrac{1}{\noiseScale} \leq m$,
    $\sfrac{(\decomposesignal - i)}{(m - \decomposesignal))} \leq 
    \sfrac{m}{2}$;
    and inequality~(b) holds since 
    $\decomposesignal \geq i + \sfrac{m\log(\noiseScale)}
{\noiseScale}$.

\item Fix an arbitrary $\decomposesignal \in [m - \sfrac{1}{2}, \NOinforSignal]$.
Let $\lastNOinforSignal \triangleq \frac{\prior_{m-1}(m - 1)+\prior_m m}{\prior_{m - 1} + \prior_m}$
be the signal on which state $m - 1$ and state $m$
are fully pooled together.
It is clear that $\Payoff{\lowerbounddecomposescheme}
\leq \Payoff{\lowerbounddecomposeschemeavglast}$
if $\decomposesignal \geq \lastNOinforSignal$.
For $\decomposesignal \leq \lastNOinforSignal$,
note that 
\begin{align*}
            \prior_m\senderU_m
    \lowerbounddecomposescheme_m(\decomposesignal)
    \Quant(\decomposesignal) &=
    K_2 
    K_1
    \left(m - i - \frac{1}{\noiseScale}\right)
    \noiseScale\exp(\noiseScale i)
    \frac{
    \decomposesignal - i
    }{
    m - \decomposesignal
    }
    \frac{1}{1 + \exp(\noiseScale\decomposesignal)}
    \\
    &\overset{(a)}{\leq}
    K_1 K_2 m \noiseScale \frac{m}{m - \decomposesignal} \exp(\noiseScale(m - 1 - \decomposesignal))
    \\
    &= K_1 K_2 m^2 \noiseScale
    \frac{\exp(\noiseScale(m - 1 - \decomposesignal))}{m - \decomposesignal}
\end{align*}
where inequality~(a) holds since 
    $m - i - \sfrac{1}{\noiseScale} \leq m$,
    $\decomposesignal - i\leq 
    m$ and $\sfrac{\exp(\noiseScale i)}{(1 + \exp(\noiseScale \decomposesignal)} \leq \exp(\noiseScale(m - 1 - \decomposesignal))$.
    
    Denote function $f(\decomposesignal) \triangleq
    K_1 K_2 m^2 \noiseScale
    \frac{\exp(\noiseScale(m - 1- \decomposesignal))}{m - \decomposesignal}$.
    Notably, $f(\lastNOinforSignal)$
    also upperbounds $\Payoff{\lowerbounddecomposeschemeavglast}$.
    Hence, it is sufficient to show $f(\decomposesignal) = o(\sfrac{K_1 K_2}{m})$
    for all $\decomposesignal \in [m - \sfrac{1}{2}, \lastNOinforSignal]$.
    Now consider the derivative of function $f(\cdot)$,
    \begin{align*}
        \frac{d f(\decomposesignal)}{d\decomposesignal}
        =
        K_1 K_2 m^2 \noiseScale \frac{
        \exp(\noiseScale(m - 1 - \decomposesignal))
        (1 - \noiseScale(m - \decomposesignal))
        }{(m - \decomposesignal)^2}
    \end{align*}
    whose sign is determined by the term 
    $1 - \noiseScale(m - \decomposesignal))$.
    Since we are considering 
$\decomposesignal\in 
[m - \sfrac{1}{2}, \lastNOinforSignal]$, 
we conclude the proof by showing 
that $f(m - \sfrac{1}{2}) = o(\sfrac{K_1 K_2}{m})$,
and $1 - \noiseScale(m - \lastNOinforSignal) \leq 0$. 
By definition,
\begin{align*}
    f\left(m - \frac{1}{2}\right) &=
    K_1 K_2 m^2 \noiseScale \frac{
    \exp\left(\noiseScale\left(m - 1 - 
    \left(m - \frac{1}{2}\right)\right)\right)
    }{m - (m - \frac{1}{2})}
    =
    o\left(\frac{K_1 K_2}{m} \right)
    \intertext{and}
    \lastNOinforSignal 
    &= \frac{\prior_{m-1}(m - 1)+\prior_m m}{\prior_{m - 1} + \prior_m}
    \\
    &= \frac{
    K_1  K_2 \left(1 - \frac{1}{\noiseScale}\right) 
    \noiseScale\exp(\noiseScale(m - 1)) (m - 1)
    +
    K_2 m
    }{
    K_1  K_2 \left(1 - \frac{1}{\noiseScale}\right) 
    \noiseScale\exp(\noiseScale(m - 1))
    +
    K_2
    }
    \\
    &=\frac{
    \frac{1}{\sum_{j\in[m-1]}
    \exp(\noiseScale j)}
    \left(1 - \frac{1}{\noiseScale}\right) 
    \noiseScale\exp(\noiseScale(m - 1))(m - 1)
    +
    m
    }{
    \frac{1}{\sum_{j\in[m-1]}
    \exp(\noiseScale j)}
    \left(1 - \frac{1}{\noiseScale}\right) 
    \noiseScale\exp(\noiseScale(m - 1))
    +
    1
    }
    \\
    &
    \overset{(a)}{\leq}
    \frac{\frac{1}{2}
    \left(1 - \frac{1}{\noiseScale}\right) 
    (m - 1) + m
    }{
    \frac{1}{2}\left(1 - \frac{1}{\noiseScale}\right) 
     + 1
    }
    \leq m - \frac{1}{3} 
    \overset{(b)}{\leq} m - \frac{1}{\noiseScale}
\end{align*}
where inequalities~(a) and (b) hold for every $m \geq 3$.
\qedhere
    \end{itemize}
\end{proof}


\subsection{Omitted Proof of 
\texorpdfstring{\Cref{coro:SDSU censorship m approx}}{}}
\label{apx:proof SDSU censorship m approx}
\begin{proof}
[Remaining Proof of \Cref{coro:SDSU censorship m approx}]
We now prove that there always exists a direct signaling scheme that is $O(m)$-approximation. 
Similarly, let 
$\signalscheme\primed$
with signal space $\signalspace\primed$
be the signaling scheme
stated in
\Cref{thm:SDSU 4 approx}.
Suppose the signal 
$\signal_{ij}\in \signalspace\primed$ 
is induced from the pair of state $(i, j)$. 
We use $\widetilde{\totalsenderU}_{ij}$
denote the expected sender utility 
induced from the signal $\signal_{ij}$, 
i.e., 
$\widetilde{\totalsenderU}_{ij} 
\triangleq
(\prior_i\senderU_i\signalscheme\primed_i(\signal_{ij})
+ \prior_j\senderU_j\signalscheme\primed_j(\signal_{ij}))\Quant(\signal_{ij})$. 
Let $(\widetilde{i}, \widetilde{j}) \triangleq \argmax_{(i, j)} \widetilde{\totalsenderU}_{ij}$. 
Then, together with the properties (i)-(ii) and 
$|\signalspace\primed| \le 2m$, we know that 
$\Payoff{\signalscheme\primed}
\le 2m\cdot 
\widetilde{\totalsenderU}_{\widetilde{i}\widetilde{j}}$.
Now consider the following direct signaling scheme 
$\widetilde{\signalscheme}$:
\begin{align*}
    \widetilde{\signalscheme}_{\widetilde{i}}(\scaledNoiseDiff)
    & = \signalscheme\primed_i(\signal_{\widetilde{i}\widetilde{j}}) \indicator{\scaledNoiseDiff = \signal_{\widetilde{i}\widetilde{j}}}, \quad
    \widetilde{\signalscheme}_{\widetilde{i}}(\scaledNoiseDiff)
    = 
    \left(1- \signalscheme\primed_i(\signal_{\widetilde{i}\widetilde{j}})\right) 
    \indicator{\scaledNoiseDiff = \widetilde{\scaledNoiseDiff}}; \\
    \widetilde{\signalscheme}_{\widetilde{j}}(\scaledNoiseDiff)
    & = \signalscheme\primed_j(\signal_{\widetilde{i}\widetilde{j}}) \indicator{\scaledNoiseDiff = \signal_{\widetilde{i}\widetilde{j}}}, \quad
    \widetilde{\signalscheme}_{\widetilde{j}}(\scaledNoiseDiff)
    = 
    \left(1- \signalscheme\primed_j(\signal_{\widetilde{i}\widetilde{j}})\right) 
    \indicator{\scaledNoiseDiff = \widetilde{\scaledNoiseDiff}}; \\
    i\in [m]\setminus\{\widetilde{i}, \widetilde{j}\},
    \quad
    \widetilde{\signalscheme}_{i}(\scaledNoiseDiff)
    & = \indicator{\scaledNoiseDiff = \widetilde{\scaledNoiseDiff}}
\end{align*}
where $\widetilde{\scaledNoiseDiff} \triangleq \frac{\prior_{\widetilde{i}}(1 - \signalscheme\primed_i(\signal_{\widetilde{i}\widetilde{j}}))\receiverU_{\widetilde{i}}
+ \prior_{\widetilde{j}}(1 - \signalscheme\primed_j(\signal_{\widetilde{i}\widetilde{j}}))\receiverU_{\widetilde{j}}+
\sum_{i\in [m]\setminus\{\widetilde{i}, \widetilde{j}\}} \prior_i\receiverU_i}{
\prior_{\widetilde{i}}(1 - \signalscheme\primed_i(\signal_{\widetilde{i}\widetilde{j}}))
+ \prior_{\widetilde{j}}(1 - \signalscheme\primed_j(\signal_{\widetilde{i}\widetilde{j}}))+
\sum_{i\in [m]\setminus\{\widetilde{i}, \widetilde{j}\}} \prior_i }$.
Essentially, 
direct signaling scheme 
$\widetilde{\signalscheme}$ has the same signaling structure 
as the signaling scheme $\signalscheme\primed$
on inducing the signal $\signal_{\widetilde{i}\widetilde{j}}$, 
and then pools all remaining states at the same signal $\widetilde{\delta}$. 
By construction, 
it is easy to verify that 
$\Payoff{\widetilde{\signalscheme}} \ge \widetilde{\totalsenderU}_{\widetilde{i}\widetilde{j}}$, 
which gives an $O(m)$-approximation of the signaling scheme 
$\widetilde{\signalscheme}$. 
\end{proof}

\subsection{Omitted Proof of 
\texorpdfstring{\Cref{thm:SDSU 4 approx}}{}}
\label{apx:SDSU proof 4 approx}

We start with the first step --  
a characterization of an optimal signaling scheme
that has the same two properties as 
in \Cref{thm:SDSU 4 approx}. 
\paragraph{Step 1- a characterization of the structure 
of an optimal signaling scheme.}
\begin{lemma}
\label{thm:SDSU opt}
In SDSU environments, for a boundedly rational receiver,
there exists an optimal signaling scheme $\optscheme$
using at most $\sfrac{m(m+1)}{2}$ signals
and this optimal signaling scheme $\optscheme$
satisfies the two properties (i) (ii)
in \Cref{thm:SDSU 4 approx}.
\end{lemma}
We provide a graphical illustration for the 
structure of optimal signaling schemes
characterized in the above \Cref{thm:SDSU opt}.
\begin{figure}[H]
\centering
\subfloat[
\label{ex: opt censor}
]{
\scalebox{.7}{\input{plots/figure-SDSU-opt-five-state-b}}
}\qquad\qquad
\subfloat[
\label{ex: opt binary censor}
]{
\scalebox{.7}{\input{plots/figure-SDSU-opt-five-state-c}}
}
\caption{\label{ex: comparison}
Structure graphical illustration for
optimal signaling schemes characterized in \Cref{thm:SISU opt} (\Cref{ex: opt censor}) and \Cref{thm:SDSU opt} (\Cref{ex: opt binary censor})
in a SISU environment.
Each state is the black dot, and all the states in a gray shaded region
imply that there exists a signal induced from 
these states. 
In both figures, 
the receiver is fully rational, i.e., $\noiseScale = \infty$.
The SISU environment is specified as below:
$\receiverU_1 = -1.5, \receiverU_2 = 0.5,
\receiverU_3 = 1, \receiverU_4 = 1.5,
\receiverU_5 = 2$, 
and $\senderU_i = 1, \prior_i =  0.2, \forall i \in[5]$.
For this problem instance, it can be shown that
a censorship signaling scheme $\optschemeSISUcensor$  
(\Cref{ex: opt censor}) 
is optimal: $\optsignalspaceSISUnew = \{\signal_1, \signal_2, \signal_3\}$ , 
$\optschemeSISUcensor_i(\signal_1)=1, \forall i\in[3];
\optschemeSISUcensor_4(\signal_2) = 1;
\optschemeSISUcensor_5(\signal_3) = 1$.
Meanwhile, 
a signaling scheme $\optscheme$ (\Cref{ex: opt binary censor}) 
that satisfies the two properties in \Cref{thm:SDSU opt}
is also optimal:
$\optsignalspace = \{\signal_1, \signal_2, \signal_3, \signal_4\}$,
$\optscheme_1(\signal_1) = \sfrac{1}{3}, 
\optscheme_2(\signal_1) = 1;
\optscheme_1(\signal_2) = \sfrac{2}{3}, 
\optscheme_3(\signal_2) = 1;
\optscheme_4(\signal_3) = 1;
\optscheme_5(\signal_4) = 1.$
}
\end{figure}
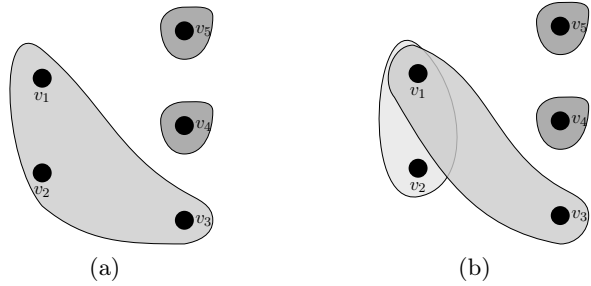
\begin{remark}
We would like to note that 
an application of the
Caratheodory's theorem 
shows that $m$ signals 
are sufficient for optimal signaling scheme.
However, such characterization does not shed much light 
on the structure of optimal signaling scheme. 
To prove \Cref{thm:SDSU 4 approx}, we resort to 
characterizing an optimal signaling scheme that uses
more signals but has more structural properties 
that we can leverage to study 
censorship/direct signaling schemes.
\end{remark}
\paragraph{Proof overview of \Cref{thm:SDSU opt}.}
At a high level, the proof of \Cref{thm:SDSU opt}
proceeds in two steps.
In step 1a, we present a reduction from arbitrary signaling schemes 
to signaling schemes that satisfy property (i).
Specifically,
given an arbitrary signaling scheme $\signalscheme$,
we can construct a new signaling scheme $\signalscheme\primed$
which satisfies property~(i) and 
achieves the same expected sender utility 
as the original signaling scheme.\footnote{In the remaining of 
this subsection,
we use superscript $\dagger$
to denote the constructed signaling schemes
satisfying property~(i),
and superscript $\ddagger$
to denote the constructed signaling schemes
satisfying properties~(i) (ii).}
In step 1b, we provide an approach 
to convert any signaling scheme $\signalscheme\primed$ 
which satisfies property~(i)
to a new signaling scheme $\signalscheme\doubleprimed$ 
which satisfies both properties (i) (ii),
and achieves weakly higher expected sender utility.
Informally,
given a signaling scheme $\signalscheme\primed$
from the first step, 
we can 
obtain the
signaling scheme 
$\signalscheme\doubleprimed$ by 
optimizing the pooling structure 
for each pair of states while holding signals from 
all other pairs fixed. 
Loosely speaking, this reduces our task
to identify optimal signaling scheme when the
state space is binary.
Hence, we introduce
a technical lemma 
(\Cref{lem:SDSU binary opt}),
showing that the optimal signaling schemes
are censorship signaling schemes when 
the sate space is binary,
which may be of independent interest.

Below we provide detailed discussion and related lemmas for the above mentioned two steps. 
In the end of this subsection, we combine all pieces together to conclude the proof of \Cref{thm:SDSU opt}.

\paragraph{Step 1a- reduction to signaling schemes with property~(i).}
In this step, we argue that it is without loss of generality to 
consider signaling schemes that satisfy property~(i) 
in \Cref{thm:SDSU opt}.
\begin{lemma}
\label{lem:SDSU opt step 1}
In SDSU environments, for a boundedly rational receiver,
for an arbitrary signaling scheme~$\signalscheme$,
there exists a signaling scheme $\signalscheme\primed$
with signal space $\signalspace\primed$
such that 
\begin{itemize}[label={-}]
\item each signal $\sigma\in\signalspace\primed$
is induced by at most two states,
\item signaling scheme $\signalscheme\primed$ 
achieves the same expected sender utility
as signaling scheme $\signalscheme$.
\end{itemize}
\end{lemma}

Informally, 
we can construct 
the signaling scheme $\signalscheme\primed$
in \Cref{lem:SDSU opt step 1} as follows.
For each signal $\signal$
in the original signaling scheme $\signalscheme$,
we decompose it into multiple signals, each of which 
is induced by at most two states,
and satisfies some other requirements.
The feasibility of this decomposition is guaranteed
by the following lemma.

\begin{lemma}[\citealp{FTX-22}]
\label{lem:multi label signal to binary label signal random variable}
Let $X$ be a random variable 
with discrete support 
$\supp(X)$.
There exists a 
positive integer $K$,
a finite set of $K$
random variables $\{X_k\}_{k\in[K]}$,
and 
convex combination coefficients
$\convexcombinbf\in [0, 1]^K$ 
with $\sum_{k\in[K]} \convexcombin_k = 1$
such that:
\begin{enumerate}
    \item[(i)] \underline{Bayesian-plausibility}: for each $k\in [K]$, $\expect{X_k} = \expect{X}$;
    \item[(ii)] \underline{Binary-support}:
    for each $k\in[K]$, the size of $X_k$'s support is at most 2,
    i.e., $|\supp(X_k)| \leq 2$
    \item[(iii)] \underline{Consistency}:
    for each $x\in \supp(X)$, 
$\prob{X=x} = \sum_{k\in[K]}
\convexcombin_k\cdot 
\prob{X_k = x}$
\end{enumerate}
\end{lemma}

\begin{proof}[Proof of \Cref{lem:SDSU opt step 1}]
Fix an arbitrary signaling scheme $\signalscheme$
with signal space $\signalspace$.
Recall that $\signalscheme_i(\signal)$
is the probability mass (or density) that signal $\signal$
is realized when the realized state is state $i$.

Now we describe the construction of $\signalscheme\primed$
and its signal space $\signalspace\primed$.
Initially, we set $\signalspace\primed \gets \emptyset$.
For each signal $\signal\in \signalspace$, 
let $\posterior_i(\signal) \triangleq
\frac{\prior_i \signalscheme_i(\signal)}
{\sum_{j\in[m]}\prior_j\signalscheme_j(\signal)}$
be its induced posterior belief for each state $i$.
Consider the following random variable $X$
where 
$\prob{X = \receiverU_i} = \posterior_i(\signal)$
for each $i\in[m]$.
Let integer $K$,
random variables $\{X_k\}_{k\in[K]}$,
and 
convex combination coefficients
$\convexcombinbf\in [0, 1]^K$ 
be the elements in \Cref{lem:multi label signal to binary label signal random variable} for the aforementioned random variable $X$. 
Add $K$ signals $\{\signal^{(1)}, \dots, \signal^{(K)}\}$
into the signal space $\signalspace\primed$,
i.e., 
$\signalspace\primed \gets 
\signalspace\primed \cup \{\signal^{(1)}, \dots, \signal^{(K)}\}$.
For each $k\in[K]$,
set $\signalscheme\primed_i(\signal\ked)
\gets 
\frac{1}{\prior_i}
{\convexcombin_k\cdot \prob{X_k = \receiverU_i}\cdot 
{(\sum_{j\in[m]} \prior_j\signalscheme_j(\signal))}}
$.
Note that this construction ensures that 
\begin{align*}
     \sum_{k\in[K]}
\signalscheme\primed_i(\signal\ked)
&=
\sum_{k\in[K]}
\frac{1}{\prior_i}
{\convexcombin_k\cdot \prob{X_k = \receiverU_i}\cdot 
{\left(\sum_{j\in[m]} \prior_j\signalscheme_j(\signal)\right)}}
\overset{(a)}{=}
\frac{1}{\prior_i}
\prob{X = \receiverU_i}\cdot 
{\left(\sum_{j\in[m]} \prior_j\signalscheme_j(\signal)\right)}
= \signalscheme_i(\signal)
\end{align*}
where equality~(a) holds due to the ``consistency'' property 
in \Cref{lem:multi label signal to binary label signal random variable}.
Hence, the constructed signaling scheme $\signalscheme\primed$
is feasible.

Additionally, the ``binary-support'' property in \Cref{lem:multi label signal to binary label signal random variable}
ensures that signaling scheme $\signalscheme\primed$ satisfies 
that each signal from $\signalspace\primed$
is induced by at most two states.

Finally, to see 
that signaling scheme $\signalscheme\primed$
achieves the same expected sender utility as signaling scheme $\signalscheme$,
consider the following coupling between these two signaling schemes:
whenever signal $\signal\in \signalspace$
is realized in signaling scheme $\signalscheme$,
sample
the corresponding signal $\signal\ked$ with probability $\convexcombin_k$
for each $k\in[K]$.
This coupling is well-defined due to the
``consistency'' property in 
\Cref{lem:multi label signal to binary label signal random variable}.
Invoking the ``Bayesian-plausibility'' property
in
\Cref{lem:multi label signal to binary label signal random variable},
from the receiver's perspective, 
her expected utility given the posterior belief 
$\posterior(\signal\ked)$
under signaling scheme $\signalscheme\primed$
is the same as 
her expected utility given the posterior belief 
$\posterior(\signal)$ 
under signaling scheme $\signalscheme$.
Thus the probabilities that the receiver takes 
action 1 
are the same in both signaling schemes,
yielding the same expected utility to the sender.
\end{proof}

\paragraph{Step 1b- reduction to signaling schemes with properties~(i) and (ii).}
In this step, we argue that it is without loss of generality to 
consider signaling schemes which satisfy properties~(i) and (ii) 
in \Cref{thm:SDSU opt}.

\begin{lemma}
\label{lem:SDSU opt step 2}
In SDSU environments, for a boundedly rational receiver,
given any signaling scheme~$\signalscheme\primed$
with signal space $\signalspace\primed$
where each signal $\sigma\in\signalspace\primed$
is induced by at most two states,
there exists a signaling scheme $\signalscheme\doubleprimed$
with signal space $\signalspace\doubleprimed$
such that 
\begin{itemize}[label={-}]
\item each signal $\sigma\in\signalspace\doubleprimed$
is induced by at most two states,
\item each pair of states is pooled at most one signal,
\item signaling scheme $\signalscheme\doubleprimed$ 
achieves weakly higher expected sender utility
as signaling scheme $\signalscheme\primed$.
\end{itemize}
\end{lemma}

\begin{proof}
Fix an arbitrary signaling scheme
$\signalscheme\primed$
with signal space $\signalspace\primed$
where each signal 
is induced by at most two states.
Below we describe the construction of $\signalscheme\doubleprimed$
and its signal space $\signalspace\doubleprimed$.
Initially, we set $\signalspace\doubleprimed \gets \emptyset$.

For each pair of states $(i, j)$,
let $\signalspace_{ij}\primed \subseteq \signalspace\primed$
be the subset of signals, each of which is induced by state $i$ 
and state $j$,
i.e., $\signalspace_{ij}\primed \triangleq
\{\signal\in \signalspace\primed:\supp(\posterior(\signal)) = 
\{i, j\}\}$.
For ease of presentation,
we introduce auxiliary notations
$\aggregateproba$
(resp.\ 
$\aggregateprobb$)
to denote 
the probability that 
the realized state is $i$
(resp.\ $j$),
and the realized signal is 
from $\signalspace\primed_{ij}$,
i.e., 
$\aggregateproba\triangleq
\int_{\signal\in\signalspace_{ij}\primed}
\prior_i\signalscheme\primed_i(\signal)\,d\signal$
and 
$\aggregateprobb\triangleq
\int_{\signal\in\signalspace_{ij}\primed}
\prior_j\signalscheme\primed_j(\signal)\,d\signal$.
Consider the program~\ref{eq:opt lp}
on the following binary-state instance $\instance_{ij}=
(\hat m, 
\{\hat \prior_k\},
\{\hat \receiverU_k\},
\{\hat \senderU_k\})$:
\begin{align*}
    \hat m \gets 2,\qquad
    \hat \receiverU_1 &\gets \receiverU_i,\qquad
    \hat \receiverU_2 \gets \receiverU_j,\qquad
    \hat \senderU_1 \gets \senderU_i,\qquad
    \hat \senderU_2 \gets \senderU_j,
    \\
    &\hat\prior_1\gets
    \frac{\aggregateproba}
    {\aggregateproba+\aggregateprobb},
    \qquad
    \hat\prior_2 \gets
    \frac{\aggregateproba}
    {\aggregateproba+\aggregateprobb}
\end{align*}
Notably, 
$\left\{
\frac{\prior_i\signalscheme\primed_i(\signal)}{
\aggregateproba
}
\cdot 
\indicator{\signal\in \signalspace_{ij}\primed},
\frac{\prior_j\signalscheme\primed_j(\signal)
}{
\aggregateprobb
}
\cdot \indicator{\signal\in \signalspace_{ij}\primed}
\right\}$
is a feasible solution 
of program~\ref{eq:opt lp}
on the binary-state instance $\instance_{ij}$.
Now, let 
$\{\hat\signalscheme^*_1(\signal),\hat\signalscheme^*_2(\signal)\}$
with signal space $\hat\signalspace_{ij}^*$
be the optimal solution of program~\ref{eq:opt lp}
on the binary-state instance $\instance_{ij}$.
We add signals from $\hat\signalspace_{ij}^*$
into the signal space $\signalspace\doubleprimed$,
i.e., $\signalspace\doubleprimed\gets \signalspace\doubleprimed
\cup \hat\signalspace_{ij}^*$,
and
set $\signalscheme\doubleprimed_i(\signal) \gets 
\frac{\aggregateproba}{\prior_i}
    \cdot
    \hat\signalscheme_1^*(\signal)$,
    $\signalscheme\doubleprimed_j(\signal) \gets 
\frac{\aggregateprobb}{\prior_j}\cdot 
    \hat\signalscheme_2^*(\signal)$
    for each signal $\signal \in \hat\signalspace_{ij}^*$.
It is straightforward
to verify that the constructed signaling scheme 
$\signalscheme\doubleprimed$ 
with signal space $\signalspace\doubleprimed$
is feasible, and 
each signal $\signal\in\signalspace\primed$
is induced by at most two states by construction.

Now we argue that 
each pair of states in signaling scheme $\signalscheme\doubleprimed$
is pooled at most one signal.
By our construction of 
signaling scheme $\signalscheme\doubleprimed$,
it is sufficient to show
that for each pair of states $(i, j)$,
the optimal solution 
$\{\hat \signalscheme_i^*(\signal),
\hat \signalscheme_j^*(\signal)\}$
in program~\ref{eq:opt lp}
on binary-state instance $\instance_{ij}$
is a censorship signaling scheme (and thus
it is pooled at most one signal).
We prove this statement by leveraging 
the following lemma (\Cref{lem:SDSU binary opt})
that characterizes the optimal 
signaling scheme of any binary state instance 
is indeed a censorship signaling scheme. 
The proof, deferred to \Cref{apx:SDSU proof},
is based on a primal-dual 
analysis similar to the one for \Cref{thm:opt state independent}.
\begin{lemma}
\label{lem:SDSU binary opt}
In SDSU environments with binary state space (i.e., $m=2$),
there exists a censorship signaling scheme
that is an optimal signaling scheme.
\end{lemma}
The proof of \Cref{lem:SDSU binary opt}
follows similar primal-dual analysis
of the one for \Cref{thm:SISU opt}, we thus defer to proof
to \Cref{apx:SDSU proof binary opt}.

Finally, we verify that expected sender utility $\Payoff{\signalscheme\doubleprimed}$
is weakly higher than 
the expected sender utility $\Payoff{\signalscheme\primed}$.
Note that 
\begin{align*}
    \Payoff{\signalscheme\primed} 
    &\overset{(a)}{=}
    \sum_{(i,j)}
    \displaystyle\int_{\signal \in \signalspace_{ij}\primed}
    \left(\prior_i\senderU_i\signalscheme\primed_i(\signal)
    +
    \prior_j\senderU_j\signalscheme\primed_j(\signal)
    \right)
    \Quant(\signal)
    \,d\signal
    \\
    &=
    \sum_{(i,j)}
    \left(
    \aggregateproba + \aggregateprobb
    \right)
    \displaystyle\int_{\signal \in \signalspace_{ij}\primed}
    \left(
    \frac{\aggregateproba}
    {\aggregateproba+\aggregateprobb}
    \senderU_i
    \frac{\prior_i\signalscheme\primed_i(\signal)}
    {\aggregateproba}
    +
    \frac{\aggregateprobb}
    {\aggregateproba+\aggregateprobb}
    \senderU_i
    \frac{\prior_j\signalscheme\primed_i(\signal)}
    {\aggregateprobb}
    \right)
    \Quant(\signal)
    \,d\signal
    \\
    &\overset{(b)}{\leq}
    \sum_{(i,j)}
    \left(
    \aggregateproba + \aggregateprobb
    \right)
    \displaystyle\int_{\signal \in \signalspace_{ij}\primed}
    \left(
    \frac{\aggregateproba}
    {\aggregateproba+\aggregateprobb}
    \senderU_i
    \frac{\prior_i\signalscheme\doubleprimed_i(\signal)}
    {\aggregateproba}
    +
    \frac{\aggregateprobb}
    {\aggregateproba+\aggregateprobb}
    \senderU_i
    \frac{\prior_j\signalscheme\doubleprimed_i(\signal)}
    {\aggregateprobb}
    \right)
    \Quant(\signal)
    \,d\signal
    \\
    &\overset{}{=}
    \sum_{(i,j)}
    \displaystyle\int_{\signal \in \signalspace_{ij}\primed}
    \left(\prior_i\senderU_i\signalscheme\doubleprimed_i(\signal)
    +
    \prior_j\senderU_j\signalscheme\doubleprimed_j(\signal)
    \right)
    \Quant(\signal)
    \,d\signal
    \\
    &\overset{(c)}{=}
    \Payoff{\signalscheme\doubleprimed} 
\end{align*}
where equalities~(a) (c) use \Cref{prop:opt lp}.
To see why inequality~(b) holds,
note that the left-hand side of 
inequality~(b) is the objective value 
of solution 
$\left\{
\frac{\prior_i\signalscheme\primed_i(\signal)}{
\aggregateproba
}
\cdot 
\indicator{\signal\in \signalspace_{ij}\primed},
\frac{\prior_j\signalscheme\primed_j(\signal)
}{
\aggregateprobb
}
\cdot \indicator{\signal\in \signalspace_{ij}\primed}
\right\}$
in program~\ref{eq:opt lp}
on binary-state instance $\instance_{ij}$,
while the right-hand side of inequality~(b),
by the construction of $\signalscheme\doubleprimed$,
is the optimal objective value in this program.
\end{proof}

Now we are ready to prove \Cref{thm:SDSU opt}.
\begin{proof}[Proof of \Cref{thm:SDSU opt}]
Invoking \Cref{lem:SDSU opt step 1}
and \Cref{lem:SDSU opt step 2},
we know that 
there exists an optimal signaling scheme 
where 
(i) each signal is induced by at most two states,
and 
(ii) each pair of states pools on at most one signal.
Note that property~(i) and property~(ii) together 
imply that its signal space has $\frac{m(m+1)}{2}$ signals.
\end{proof}

In below, we provide the analysis of the second
step for the proof of \Cref{thm:SDSU 4 approx}.
\paragraph{Step 2- a connection to fractional 
generalized assignment problem.}
Due to properties (i) and (ii) of 
the optimal signaling scheme $\optscheme$
stated in \Cref{thm:SDSU opt},
there is at most one signal realized by 
each pair of states $(i, j)$,
i.e.,  $|\{\signal:
\optscheme_i(\signal) > 0
\land 
\optscheme_j(\signal)>0\}| \leq 1$.
For ease of presentation, we assume 
$|\{\signal:
\optscheme_i(\signal) > 0
\land 
\optscheme_j(\signal)>0\}| = 1$
for each pair $(i, j)$,
and denote it as $\signal_{ij}$.\footnote{The analysis 
in this subsection extends trivially if 
$|\{\signal:
\optscheme_i(\signal) > 0
\land 
\optscheme_j(\signal)>0\}| = 0$ for some pair $(i, j)$.}
Furthermore, we define set of pairs
$\Edge \triangleq \{(i, j):
\optscheme_i(\signal_{ij}) \geq 
\optscheme_j(\signal_{ij})\}$.
Note that
the expected sender utility 
$\Payoff{\optscheme}
=
\sum_{(i,j)\in\Edge}
(\prior_i\senderU_i\optscheme_i(\signal_{ij}) + 
\prior_j\senderU_j\optscheme_j(\signal_{ij})
)\Quant(\signal_{ij})
$
can be upper bounded by the optimal value
of the following linear program,
\begin{align}
    \label{eq:SDSU opt matching lp}
    \arraycolsep=5.4pt\def\arraystretch{1}
    \tag{$\mathcal{P}_\texttt{SDSU-OPT}$}
    &\begin{array}{lll}
     \max\limits_{\xbf\geq \zerobf} ~ &
    \displaystyle\sum_{(i,j)\in\Edge}
       \left(\prior_i \senderU_i + 
       \prior_j \senderU_j \frac{\optscheme_j(\signal_{ij})}{
       \optscheme_i(\signal_{ij})}
       \right)
       \Quant(\signal_{ij})
       \edgealloci
     \quad& \text{s.t.} 
     \vspace{1mm}
     \\
       & 
       \displaystyle\sum_{j:(i,j)\in\Edge}
       \edgealloci
       \leq 1
       & i\in[m] 
     \vspace{1mm}
     \\
       & 
       \displaystyle\sum_{i:(i,j)\in\Edge}
       \frac{\optscheme_j(\signal_{ij})}
       {\optscheme_i(\signal_{ij})}\cdot\edgealloci
       \leq 1
       \quad& j\in[m] 
    \end{array}
\end{align}

\begin{lemma}
\label{lem:SDSU 8 approx step 1}
The expected sender utility $\Payoff{\optscheme}$
of the optimal signaling scheme $\optscheme$
is at most the optimal objective value 
of program~\ref{eq:SDSU opt matching lp}.
\end{lemma}
\begin{proof}
Consider the following assignment $\xbf$
of program~\ref{eq:SDSU opt matching lp},
\begin{align*}
    i\in[m],~~ j\in[m]:\qquad
    \edgealloci\gets \optscheme_i(\signal_{ij})
\end{align*}
By construction, the objective value of the constructed assignment 
equals $\Payoff{\optscheme}$.
Now, we show the feasibility of the constructed assignment.
Note the feasibility of optimal signaling scheme $\optscheme$
implies that for each state $i\in[m]$,
    $\sum_{j\in[m]}\optscheme_i(\signal_{ij}) = 1$.
Thus,
\begin{align*}
    \sum_{j\in[m]} \edgealloci = 
    \sum_{j\in[m]}\optscheme_i(\signal_{ij})
    \leq 1; \quad 
    \sum_{i\in[m]} 
    \frac{\optscheme_j(\signal_{ij})}
    {\optscheme_i(\signal_{ij})}
    \cdot 
    \edgealloci
    =
    \sum_{i\in[m]} 
    \frac{\optscheme_j(\signal_{ij})}
    {\optscheme_i(\signal_{ij})}
    \cdot 
    \optscheme_i(\signal_{ij})
    =
     \sum_{i\in[m]} 
    {\optscheme_j(\signal_{ij}})
    \leq 1
\end{align*}
which finishes the proof.
\end{proof}

We remark that the program~\ref{eq:SDSU opt matching lp}
has the same formulation as the fractional generalized assignment problem: 
there are $m$ items and $m$ bins.
Each bin has a unit budget.
Each pair of item $i$ and bin $j$ such that $(i,j)\in\Edge$
has value 
$ (\prior_i \senderU_i + 
       \sfrac{
       \prior_j \senderU_j
       \optscheme_j(\signal_{ij})}{
       \optscheme_i(\signal_{ij})}
       )
       \Quant(\signal_{ij})$
and cost $\sfrac{\optscheme_j(\signal_{ij})}{
\optscheme_i(\signal_{ij})}$.
With this connection to the generalized assignment problem,
we use 
the following established result 
about the optimal integral solution of program~\ref{eq:SDSU opt matching lp}. 

\begin{lemma}[Theorem 2.1 and its proof in \citealp{ST-93}]
\label{lem:SDSU 8 approx step 2}
The optimal integral solution of program~\ref{eq:SDSU opt matching lp}
is a 2-approximation to the optimal fraction solution 
of program~\ref{eq:SDSU opt matching lp}.
\end{lemma}

Now we are ready to prove \Cref{thm:SDSU 4 approx}.
\begin{proof}[Proof of \Cref{thm:SDSU 4 approx}]
Let $\xbf\primed$ be the optimal integral solution of 
program~\ref{eq:SDSU opt matching lp}.
Consider a signaling scheme $\signalscheme\primed$
constructed as follows.
First, initialize the signal space
$\signalspace\primed \gets \emptyset$.
Second, for each pair of state $(i, j)\in\Edge$,
if $\edgealloci\primed > 0$,
update $\signalspace\primed\gets 
\signalspace\primed \cup \{\signal_{ij}\}$,
$\signalscheme\primed_i(\signal_{ij}) \gets 
\sfrac{\edgealloci\primed}{2}$,
and 
$\signalscheme\primed_j(\signal_{ij}) \gets 
\sfrac{(\edgealloci\primed\optscheme_j(\signal_{ij}))}
{(2\optscheme_i(\signal_{ij}))}$.
Third, for each state $i\in[m]$, 
if $\sum_{\signal\in\signalspace\primed}\signalscheme_i\primed(\signal) < 1$,
update $\signalspace\primed\gets 
\signalspace\primed \cup \{\receiverU_i\}$,
$\signalscheme\primed_i(\receiverU_i) \gets 
1 - \sum_{\signal\in\signalspace\primed}\signalscheme_i\primed(\signal)$.

Now we verify that the constructed signaling scheme
$\signalscheme\primed$ is feasible,
i.e., for each state $i\in[m]$,
$\sum_{\signal\in\signalspace\primed} \signalscheme\primed_i(\signal) = 1$.
By construction, the feasibility is guaranteed since
that 
for each state $i\in[m]$,
\begin{align*}
    \sum_{j:(i, j)\in \Edge} \signalscheme\primed_i(\signal_{ij})
    + \sum_{j:(j, i)\in \Edge} \signalscheme\primed_i(\signal_{ji})
     =
    \sum_{j:(i,j)\in\Edge} 
    \frac{1}{2}\edgealloci\primed
    +
    \sum_{j:(j,i)\in\Edge}
    \frac{1}{2}\frac{\optscheme_i(\signal_{ji})}
    {\optscheme_j(\signal_{ji})}
    \edgealloc_{ji}\primed
    \leq \frac{1}{2} + \frac{1}{2} = 1
\end{align*}
where the inequality holds due to the feasibility of solution 
$\xbf\primed$.

Next, we verify that the constructed signaling scheme
$\signalscheme\primed$
satisfies properties 
stated in \Cref{thm:SDSU 4 approx}.
Note the two properties same as in \Cref{thm:SDSU opt} 
are guaranteed by construction
straightforwardly. 
By construction, the expected sender utility 
$\Payoff{\signalscheme\primed}$ is a 2-approximation
to the objective value of the optimal integral solution $\xbf\primed$.
Invoking \Cref{lem:SDSU 8 approx step 1}
and \Cref{lem:SDSU 8 approx step 2},
we conclude that 
signaling scheme $\signalscheme\primed$
is 4-approximation to the optimal signaling scheme.

Finally, since the optimal integral solution $\xbf\primed$
has at most $m$ non-zero entries, i.e., 
$|\{\edgealloci\primed:\edgealloci\primed > 0\}| \leq m$,
the constructed signal space $\signalspace\primed$
has at most $|\{\edgealloci\primed:\edgealloci\primed > 0\}| + m \leq 2m$ signals.
\end{proof}

\subsection{Proof of 
\texorpdfstring{\Cref{lem:SDSU binary opt}}{}}
\label{apx:SDSU proof binary opt}

\input{Paper/apx-proof-SDSU-binary-opt}

%% file: plots/figure-SDSU-opt-five-state-b.tex
\begin{tikzpicture}[scale=0.9, transform shape]
    \node (v1) at (0,3) {};
    \node (v2) at (0,1) {};
    \node (v3) at (3,0) {};
    \node (v4) at (3,2) {};
    \node (v5) at (3,4) {};

    \begin{scope}[fill opacity=0.8]
    \filldraw[fill=gray!40!white] 
        ($(v1)+(0,0.6)$) 
        to[out=-40,in=150] ($(v3) + (0, 0.6)$) 
        to[out=-30,in=90] ($(v3) + (0.6, 0)$) 
        to[out=-90,in=10] ($(v3) + (0, -0.5)$)
        to[out=180,in=-40] ($(v2) + (0,-0.7)$)
        to[out=130,in=140] ($(v1)+(0,0.6)$);
    \filldraw[fill=gray!80!white] 
        ($(v4)+(-0.5,0.1)$)
        to[out=90,in=180] ($(v4)+(0,0.5)$)
        to[out=0,in=90] ($(v4)+(0.6,0.3)$)
        to[out=270,in=0] ($(v4)+(0,-0.6)$)
        to[out=180,in=270] ($(v4)+(-0.5,0.1)$);
    \filldraw[fill=gray!80!white] 
        ($(v5)+(-0.5,0.1)$)
        to[out=90,in=180] ($(v5)+(0,0.5)$)
        to[out=0,in=90] ($(v5)+(0.6,0.3)$)
        to[out=270,in=0] ($(v5)+(0,-0.6)$)
        to[out=180,in=270] ($(v5)+(-0.5,0.1)$);
    \end{scope}


    \fill (v1) circle (0.2) 
    node [below, yshift=-0.15 cm] {$\receiverU_1$};
    \fill (v2) circle (0.2) 
    node [below, yshift=-0.15 cm] {$\receiverU_2$};
    \fill (v3) circle (0.2) 
    node [right, xshift=0.1 cm] {$\receiverU_3$};
    \fill (v4) circle (0.2) 
    node [right, xshift=0.1 cm] {$\receiverU_4$};
    \fill (v5) circle (0.2) 
    node [right, xshift=0.1 cm] {$\receiverU_5$};

\end{tikzpicture}

%% file: plots/figure-SDSU-opt-five-state-c.tex
\begin{tikzpicture}[scale=0.9, transform shape]
    \node (v1) at (0,3) {};
    \node (v2) at (0,1) {};
    \node (v3) at (3,0) {};
    \node (v4) at (3,2) {};
    \node (v5) at (3,4) {};

    \begin{scope}[fill opacity=0.8]
    \filldraw[fill=gray!20!white] 
        ($(v1)+(0,0.6)$) 
        to[out=-40,in=10] ($(v2) + (0, -0.6)$) 
        to[out=190,in=140] ($(v1) + (0,0.6)$);
    \filldraw[fill=gray!40!white] 
        ($(v1)+(-0.1,0.6)$) 
        to[out=-20,in=150] ($(v3) + (0, 0.6)$) 
        to[out=-30,in=90] ($(v3) + (0.6, 0)$) 
        to[out=-90,in=0] ($(v3) + (0, -0.6)$)
        to[out=170,in=-60] ($(v1) + (-0.5,-0.5)$)
        to[out=130,in=180] ($(v1)+(-0.1,0.6)$);
    \filldraw[fill=gray!80!white] 
        ($(v4)+(-0.5,0.1)$)
        to[out=90,in=180] ($(v4)+(0,0.5)$)
        to[out=0,in=90] ($(v4)+(0.6,0.3)$)
        to[out=270,in=0] ($(v4)+(0,-0.6)$)
        to[out=180,in=270] ($(v4)+(-0.5,0.1)$);
    \filldraw[fill=gray!80!white] 
        ($(v5)+(-0.5,0.1)$)
        to[out=90,in=180] ($(v5)+(0,0.5)$)
        to[out=0,in=90] ($(v5)+(0.6,0.3)$)
        to[out=270,in=0] ($(v5)+(0,-0.6)$)
        to[out=180,in=270] ($(v5)+(-0.5,0.1)$);
    \end{scope}


    \fill (v1) circle (0.2) 
    node [below, yshift=-0.15 cm] {$\receiverU_1$};
    \fill (v2) circle (0.2) 
    node [below, yshift=-0.15 cm] {$\receiverU_2$};
    \fill (v3) circle (0.2) 
    node [right, xshift=0.1 cm] {$\receiverU_3$};
    \fill (v4) circle (0.2) 
    node [right, xshift=0.1 cm] {$\receiverU_4$};
    \fill (v5) circle (0.2) 
    node [right, xshift=0.1 cm] {$\receiverU_5$};

\end{tikzpicture}

%% file: Paper/apx-proof-SDSU-binary-opt.tex
We now
present a more detailed statement 
for \Cref{lem:SDSU binary opt} and then present its
associated proof.
\begin{lemma}
\label{lem:SDSU binary opt tmp}
In SDSU environments with binary state space (i.e., $m=2$),
there exists an optimal signaling scheme $\optscheme$
for a boundedly rational receiver
that
is a censorship signaling scheme.
In particular, 
define $\gamma(\scaledNoiseDiff) \triangleq \frac{\receiverU_1-\scaledNoiseDiff}{ \receiverU_2-\scaledNoiseDiff} + \frac{\Quant(\receiverU_2) - \Quant(\scaledNoiseDiff)}{ \receiverU_2 - \scaledNoiseDiff}\cdot\frac{1}{\Quant'(\scaledNoiseDiff)} \cdot\left(1 - \frac{\receiverU_1-\scaledNoiseDiff}{\receiverU_2-\scaledNoiseDiff} \right)$. 
Let $\widehat{\scaledNoiseDiff}$ satisfy $\gamma(\widehat{\scaledNoiseDiff}) = \sfrac{\senderU_1}{\senderU_2}$, 
and define 
$\censorshipsignal \triangleq \min\left\{\max\left\{\receiverU_1, \widehat{\scaledNoiseDiff}\right\}, \prior_1 \receiverU_1 + \prior_2 \receiverU_2\right\} ; ~~
\censorshipprob \triangleq \frac{\prior_1  \left(\censorshipsignal - \receiverU_1 \right)}{\prior_2\left(\receiverU_2 - \censorshipsignal\right)}~.$
Then the optimal signaling $\optscheme$ is
\begin{equation}
\begin{aligned}
    \label{eq: opt signaling SDSU binary}
    \optscheme_1\left(\censorshipsignal \right) & = 1; ~ 
    \optscheme_2\left(\censorshipsignal \right) 
    = \censorshipprob, ~
    \optscheme_2\left(\receiverU_2\right) 
    = 1 - \censorshipprob~.
\end{aligned}
\end{equation}
\end{lemma}
A few useful observations of the above result are
as follows. 
First, by inspecting the first-order derivative, 
we know that the function $\gamma(\cdot)$ is monotone decreasing. 
Second, we always have 
$\censorshipsignal \in \left[\receiverU_1, \prior_1\receiverU_1 + \prior_2\receiverU_2\right]$ and thus
$\censorshipprob \in [0, 1]$. 
Third, 
(a) when $\widehat{\scaledNoiseDiff} \le  \receiverU_1$, we
have $\censorshipsignal = \receiverU_1$ and $\censorshipprob = 0$, 
and thus full information revealing is optimal;
(b) when $\receiverU_1 < \widehat{\scaledNoiseDiff} < \prior_1  \receiverU_1 + \prior_2\receiverU_2$, 
we have $\censorshipsignal = \widehat{\scaledNoiseDiff}$ and $\censorshipprob \in (0, 1)$,
and thus partial information revealing is optimal;
(c) when $\widehat{\scaledNoiseDiff} > \prior_1 \receiverU_1 + \prior_2  \receiverU_2$, 
we have $\censorshipsignal = \prior_1\receiverU_1 + \prior_2\receiverU_2$
and $\censorshipprob = 1$, and thus no information revealing is optimal.

\begin{proof}[Proof of \Cref{lem:SDSU binary opt tmp}]
We prove the optimality of the signaling scheme 
\eqref{eq: opt signaling SDSU binary} 
by constructing a feasible dual solution to 
the dual program \ref{eq:opt lp dual} that satisfies 
the complementary slackness.

Based on the signaling scheme \eqref{eq: opt signaling SDSU binary}, 
we give our dual solution to the program \ref{eq:opt lp dual}
as follows:
\begin{equation}
\begin{aligned}
    \label{dual solution SDSU binary}
    \scaledNoiseDiff\in(\infty, \receiverU_1]:
    & ~
    \noiseDualVar(\scaledNoiseDiff)
    =\max_{i\in[2]} 
    - \frac{\senderU_i\cdot(\Quant(\scaledNoiseDiff) - \Quant(\censorshipsignal)) + \noiseDualVar(\censorshipsignal) \cdot(\censorshipsignal-\receiverU_i)}{\scaledNoiseDiff - \receiverU_i} \\
    \scaledNoiseDiff\in(\receiverU_1, \receiverU_2]:
    & ~
    \noiseDualVar(\scaledNoiseDiff)
    =\frac{-\senderU_2(\Quant(\scaledNoiseDiff) - \Quant(\censorshipsignal))
    + \max\left\{\frac{-\senderU_2(\Quant(\receiverU_2)-\Quant(\censorshipsignal))}{\receiverU_2-\censorshipsignal},
    -(\prior_1\senderU_1+\prior_2\senderU_2) \Quant'(\censorshipsignal)\right\} (\censorshipsignal-\receiverU_2)}{\scaledNoiseDiff - \receiverU_2}  \\
    \scaledNoiseDiff\in(\receiverU_2, \infty]:
    & ~ 
    \noiseDualVar(\scaledNoiseDiff)=0\\
    i = 1:
    & ~ 
    \distDualVar(1) = \prior_1 \senderU_1 \cdot
    \left(\Quant(\censorshipsignal) + \noiseDualVar(\censorshipsignal)\cdot \frac{\censorshipsignal -\receiverU_1 }{\senderU_1}\right) \\
    i = 2:
    & ~ 
    \distDualVar(2) = \prior_2 \senderU_2 \cdot \left(\Quant(\censorshipsignal) + \noiseDualVar(\censorshipsignal)\cdot \frac{\censorshipsignal -\receiverU_2 }{\senderU_2}\right);
\end{aligned}
\end{equation}
Given the above constructed dual assignment, 
we first argue that when No information revealing is optimal,
namely, $\censorshipprob = 1$, we have 
$\noiseDualVar(\censorshipsignal) 
= -(\prior_1\senderU_1 + \prior_2\senderU_2) \cdot \Quant'(\censorshipsignal)$
where $\censorshipsignal = \prior_1\receiverU_2+\prior_2\receiverU_2$, 
otherwise we have $\noiseDualVar(\censorshipsignal) 
= -\frac{\senderU_2\cdot (\Quant(\censorshipsignal) - \Quant(\receiverU_2))}{\censorshipsignal - \receiverU_2}$ where 
$\censorshipsignal \in [\receiverU_1, \prior_1\receiverU_2+\prior_2\receiverU_2)$. 
For notation simplicity, let $\NOinforSignal \triangleq \prior_1\receiverU_2+\prior_2\receiverU_2$. 
To see this, note that when $\censorshipprob = 1$, 
it must be the case $\gamma(\censorshipsignal)
= \gamma(\NOinforSignal) \ge \frac{\senderU_1}{\senderU_2}$.
Recall that 
\begin{align*}
    \gamma(\NOinforSignal)
    & = \frac{\receiverU_1 - \NOinforSignal}{\receiverU_2 - \NOinforSignal}
    + \frac{\Quant(\receiverU_2) - \Quant(\NOinforSignal)}{\receiverU_2 - \NOinforSignal}\cdot\frac{1}{\Quant'(\NOinforSignal)} \cdot\left(1 - \frac{\receiverU_1 - \NOinforSignal}{\receiverU_2 - \NOinforSignal}\right) \nonumber\\
    & = -\frac{\prior_2}{\prior_1} + \frac{\Quant(\receiverU_2) - \Quant(\NOinforSignal)}{\receiverU_2 - \NOinforSignal}\cdot\frac{1}{\Quant'(\NOinforSignal)} \cdot\left(1 + \frac{\prior_2}{\prior_1}\right)
\end{align*}
Hence, we have 
\begin{align*}
    -\frac{\prior_2}{\prior_1} + \frac{\Quant(\receiverU_2) - \Quant(\NOinforSignal)}{\receiverU_2 - \NOinforSignal}\cdot\frac{1}{\Quant'(\NOinforSignal)} \cdot\left(1 + \frac{\prior_2}{\prior_1}\right) \ge \frac{\senderU_1}{\senderU_2}~.
\end{align*}
Rearranging the above inequality gives us
\begin{align*}
    - (\prior_1 \senderU_1+\prior_2 \senderU_2) \cdot \Quant'(\NOinforSignal) \ge 
    - \frac{\senderU_2\cdot (\Quant(\NOinforSignal) - \Quant(\receiverU_2) )}{\NOinforSignal - \receiverU_2}~,
\end{align*}
which implies the dual assignment of $\noiseDualVar(\NOinforSignal)$
when No information revealing is optimal.
As a consequence, we have 
$\distDualVar(2) = \prior_2\senderU_2 \Quant(\receiverU_2)$
when No information revealing is not optimal, and 
$\distDualVar(2) = \prior_2 \senderU_2 \left(\Quant(\NOinforSignal) 
-(\prior_1\senderU_1 + \prior_2\senderU_2) \cdot \Quant'(\NOinforSignal) \cdot \frac{\NOinforSignal -\receiverU_2 }{\senderU_2}\right)$
when No information revealing is optimal.

In below, we show that the above constructed dual solution 
\eqref{dual solution SDSU binary}
is indeed a feasible solution to the dual program \ref{eq:opt lp dual}
(i.e., the following constraint \eqref{dual feasibility SDSU binary} holds), 
and also, complementary slackness holds between 
\eqref{eq: opt signaling SDSU binary} and \eqref{dual solution SDSU binary} 
(i.e., the following constraint \eqref{CL non-zero primal SDSU binary} holds).
\begin{align}
    \Quant(\scaledNoiseDiff) + \noiseDualVar(\scaledNoiseDiff) \cdot \frac{\scaledNoiseDiff - \receiverU_i}{\senderU_i}
    & = \frac{\distDualVar(i)}{\prior_i \senderU_i} , 
    ~~ \text{ if } \optscheme_i(\scaledNoiseDiff) > 0, ~ \forall i\in[2],
    \quad\quad(\text{complementary-slackness})
    \label{CL non-zero primal SDSU binary}\\
    \Quant(\scaledNoiseDiff) + \noiseDualVar(\scaledNoiseDiff) \cdot \frac{\scaledNoiseDiff - \receiverU_i}{\senderU_i}
    & \le \frac{\distDualVar(i)}{\prior_i \senderU_i} , 
    ~~ \text{ if } \optscheme_i(\scaledNoiseDiff) = 0, ~ \forall i\in[2],
    \quad\quad(\text{dual-feasibility})
    \label{dual feasibility SDSU binary}
\end{align}

\paragraph{Complementary slackness.} 
We now argue the complementary slackness 
of the constructed assignment.
Namely, for each state $i\in[2]$
and $\scaledNoiseDiff\in(-\infty,\infty)$
such that $\optscheme_i(\scaledNoiseDiff) > 0$,
its corresponding dual constraint holds with equality,
i.e., the above equality \eqref{CL non-zero primal SDSU binary}.
We verify this for each state $i\in[2]$ separately.
\begin{itemize}[label={-}]
    \item Fix state $1$,
    note that $\optscheme_1(\scaledNoiseDiff) > 0$
    for $\scaledNoiseDiff = \censorshipsignal$ only.
    Here equality~\eqref{CL non-zero primal SDSU binary} holds 
    by construction. 
    
    \item Fix state $2$, 
    note that $\optscheme_2(\scaledNoiseDiff) > 0$
    for $\scaledNoiseDiff = \censorshipsignal$ and
    $\scaledNoiseDiff = \receiverU_2$ only.
    When $\scaledNoiseDiff = \censorshipsignal$, the
    equality~\eqref{CL non-zero primal SDSU binary}
    holds for $\optscheme_{2}(\censorshipsignal) > 0$.
    To see this, 
    when No information revealing is not optimal, 
    we have 
    \begin{align*}
        \Quant(\censorshipsignal) + \noiseDualVar(\censorshipsignal)
        \cdot \frac{\censorshipsignal - \senderU_2}{\senderU_2}
        \overset{(a)}{=} 
        \Quant(\censorshipsignal) -\frac{\senderU_2\cdot (\Quant(\receiverU_2) - \Quant(\censorshipsignal))}{\receiverU_2 - \censorshipsignal}
        \cdot \frac{\censorshipsignal - \receiverU_2}{\senderU_2}
        = \Quant(\receiverU_2)
        \overset{(b)}{=}  \frac{\distDualVar(2)}{\prior_2 \senderU_2}~,
    \end{align*}
    where the equality (a) holds 
    due to the assignment $\noiseDualVar(\censorshipsignal)$,
    and the equality (b) holds due to the assignment $\distDualVar(2)$.
    When No information revealing is optimal, 
    we have 
    \begin{align*}
        \Quant(\NOinforSignal) + \noiseDualVar(\NOinforSignal)
        \cdot \frac{\NOinforSignal - \senderU_2}{\senderU_2}
        & \overset{(a)}{=}  \frac{\distDualVar(2)}{\prior_2 \senderU_2}~,
    \end{align*}
    where the equality (a) directly follows from the assignment of $\distDualVar(2)$.
    Now it is remaining verify equality \eqref{CL non-zero primal SDSU binary} for $\optscheme_2(\receiverU_2)$.
    To see this, note that this must be the case where $\distDualVar(2) = \prior_2 \senderU_2 \Quant(\receiverU_2)$: 
    \begin{align*}
        \Quant(\receiverU_2) + \noiseDualVar(\receiverU_2)
        \cdot \frac{\receiverU_2 - \senderU_2}{\senderU_2}
        \overset{(a)}{=} 
        \Quant(\receiverU_2)
        \overset{(b)}{=}  \frac{\distDualVar(2)}{\prior_2 \senderU_2}~,
    \end{align*}
    where the equalities (a) (b) hold 
    due to the assignment $\noiseDualVar(\censorshipsignal)$
    and the assignment $\distDualVar(2)$.
\end{itemize}
\paragraph{Dual feasibility.}
when No information revealing is not optimal.
Note that 
\begin{align*}
    \frac{\distDualVar(1)}{\prior_1\senderU_1}
    \overset{(a)}{=}
    \Quant(\censorshipsignal) + \noiseDualVar(\censorshipsignal)\cdot \frac{\censorshipsignal -\receiverU_1 }{\senderU_1},
    ~~ 
    \frac{\distDualVar(2)}{\prior_2\senderU_2}
    = \Quant(\censorshipsignal) + \noiseDualVar(\censorshipsignal)\cdot \frac{\censorshipsignal -\receiverU_2 }{\senderU_2}~.
\end{align*}
Thus, we can rewrite those dual constraints associated with $\optscheme_1(\scaledNoiseDiff)$ for state $1$
and those dual constraints associated with $\optscheme_2(\scaledNoiseDiff)$ for state $2$
as follows
\begin{align}
    \label{eq:SDSU dual feasibility state 1}
    \Quant(\scaledNoiseDiff) + \noiseDualVar(\scaledNoiseDiff) \cdot \frac{\scaledNoiseDiff - \receiverU_1}{\senderU_1}
    & \le
    \Quant(\censorshipsignal) + \noiseDualVar(\censorshipsignal)\cdot \frac{\censorshipsignal -\receiverU_1 }{\senderU_1}; \\
    \Quant(\scaledNoiseDiff) + \noiseDualVar(\scaledNoiseDiff) \cdot \frac{\scaledNoiseDiff - \receiverU_2}{\senderU_2}
    & \le
    \Quant(\censorshipsignal) + \noiseDualVar(\censorshipsignal)\cdot \frac{\censorshipsignal -\receiverU_2 }{\senderU_2}~.
    \label{eq:SDSU dual feasibility state 2}
\end{align}
We verify the above inequalities for different values of $\scaledNoiseDiff$ in three cases separately.
\begin{itemize}[label={-}]
    \item 
    Fix an arbitrary $\scaledNoiseDiff \in (-\infty, \receiverU_1]$. 
    Note by the construction
    of the dual assignment in \eqref{dual solution SDSU binary}, we have 
    \begin{align*}
        \noiseDualVar(\scaledNoiseDiff)
        \ge - \frac{\senderU_1(\Quant(\scaledNoiseDiff) - \Quant(\censorshipsignal))}{\scaledNoiseDiff - \receiverU_1}
        + \noiseDualVar(\censorshipsignal)  \frac{\censorshipsignal-\receiverU_1}{\scaledNoiseDiff - \receiverU_1}, ~~
        \noiseDualVar(\scaledNoiseDiff)
        \ge - \frac{\senderU_2(\Quant(\scaledNoiseDiff) - \Quant(\censorshipsignal))}{\scaledNoiseDiff - \receiverU_2}
        + \noiseDualVar(\censorshipsignal)  \frac{\censorshipsignal-\receiverU_2}{\scaledNoiseDiff - \receiverU_2}
    \end{align*}
    after rearranging the terms, the above inequalities 
    directly imply the inequality \eqref{eq:SDSU dual feasibility state 1} and 
    and the inequality \eqref{eq:SDSU dual feasibility state 2}.

    \item 
    Fix an arbitrary $\scaledNoiseDiff \in (\receiverU_1, \receiverU_2)$. 
    
    In this case, 
    we first argue the feasibility of 
    dual assignment \eqref{eq: opt signaling SDSU binary}
    when No information revealing is not optimal. 
    By construction, we have
    $ \noiseDualVar(\scaledNoiseDiff)
        = -\frac{\senderU_2\cdot(\Quant(\scaledNoiseDiff) - \Quant(\receiverU_2))}{\scaledNoiseDiff - \receiverU_2}~,$
    which directly implies the inequality \eqref{eq:SDSU dual feasibility state 2}. 
    To ensure the inequality \eqref{eq:SDSU dual feasibility state 1}, 
    it is remaining to show that 
    the following holds for all $\scaledNoiseDiff \in (\receiverU_1, \receiverU_2)$
    \begin{align*}
        -\frac{\senderU_2\cdot(\Quant(\scaledNoiseDiff) - \Quant(\receiverU_2))}{\scaledNoiseDiff - \receiverU_2}
        & \le 
        \frac{-\senderU_1 \cdot(\Quant(\scaledNoiseDiff) - \Quant(\censorshipsignal)) + \noiseDualVar(\censorshipsignal)\cdot(\censorshipsignal-\receiverU_1)}{\scaledNoiseDiff - \receiverU_1}~,
    \end{align*}
    which is equivalent to show
    the following holds for all $\scaledNoiseDiff \in (\receiverU_1, \receiverU_2)$
    \begin{align}
        \label{ineq: SDSU dual feasibility middle}
        \frac{\senderU_2}{\senderU_1} 
        \cdot 
        \left(
        -\frac{\Quant(\scaledNoiseDiff) - \Quant(\receiverU_2)}{\scaledNoiseDiff - \receiverU_2}
        + \frac{\Quant(\censorshipsignal) - \Quant(\receiverU_2)}{\censorshipsignal-\receiverU_2}
        \cdot \frac{\censorshipsignal-\receiverU_1}{\scaledNoiseDiff - \receiverU_1}
        \right)
        \le 
        \frac{\Quant(\censorshipsignal) - \Quant(\scaledNoiseDiff)}{\scaledNoiseDiff - \receiverU_1}~.
    \end{align}
    We define following function 
    $f(\scaledNoiseDiff) \triangleq
    \left(
    -\frac{\Quant(\scaledNoiseDiff) - \Quant(\receiverU_2)}{\scaledNoiseDiff - \receiverU_2}
    + \frac{\Quant(\censorshipsignal) - \Quant(\receiverU_2)}{\censorshipsignal-\receiverU_2}
    \cdot \frac{\censorshipsignal-\receiverU_1}{\scaledNoiseDiff - \receiverU_1}
    \right) \cdot \frac{\scaledNoiseDiff - \receiverU_1}{\Quant(\censorshipsignal) - \Quant(\scaledNoiseDiff)}~.$
    Then the inequality \eqref{ineq: SDSU dual feasibility middle}
    is equivalent to show that 
    \begin{align}
        \forall \scaledNoiseDiff\in(\receiverU_1, \censorshipsignal], 
        f(\scaledNoiseDiff) \ge \frac{\senderU_1}{\senderU_2}; 
        \text{ and }
        \forall \scaledNoiseDiff\in(\censorshipsignal, \receiverU_2), 
        f(\scaledNoiseDiff) \le \frac{\senderU_1}{\senderU_2}~.
        \label{ineq: SDSU dual feasibility middl 2}
    \end{align}
    Inspecting the first-order derivative of the function
    $f(\cdot)$, we know that 
    $f'(\scaledNoiseDiff) < 0, \forall \scaledNoiseDiff\in(\receiverU_1, \receiverU_2)$.
    Moreover, observe that 
    \begin{align*}
        \lim_{\scaledNoiseDiff \rightarrow (\censorshipsignal)^+}
        f(\scaledNoiseDiff)
        & = \frac{ \receiverU_1-\censorshipsignal}{\receiverU_2-\censorshipsignal} + \frac{\Quant(\receiverU_2) - \Quant(\censorshipsignal)}{\receiverU_2 - \censorshipsignal}\cdot\frac{1}{\Quant'(\censorshipsignal)} \cdot\left(1 - \frac{ \receiverU_1-\censorshipsignal}{\receiverU_2-\censorshipsignal} \right)
        \overset{(a)}{=}
        \gamma(\censorshipsignal)~,
    \end{align*}
    where the equality (a) follows from the definition of 
    the function $\gamma(\cdot)$.
    Recall that by definition, 
    when $\censorshipsignal \in 
    (\receiverU_1, \prior_1\receiverU_1 + \prior_2\receiverU_2)$,
    we must have $\gamma(\censorshipsignal) = \frac{\senderU_1}{\senderU_2}$, 
    which proves the inequality \eqref{ineq: SDSU dual feasibility middl 2}.
    When $\censorshipsignal = \receiverU_1$, 
    it suffices to argue that 
    $\forall \scaledNoiseDiff\in(\receiverU_1, \receiverU_2), 
    f(\scaledNoiseDiff) \le \frac{\senderU_1}{\senderU_2}$.
    To see this, note that 
    for all $\scaledNoiseDiff \in(\receiverU_1, \receiverU_2)$, 
    we have $f(\scaledNoiseDiff) \le f(\scaledNoiseDiff_1) = \lim_{\scaledNoiseDiff \rightarrow (\censorshipsignal)^+}f(\scaledNoiseDiff) = \gamma(\censorshipsignal)
    \le \gamma(\widehat{\scaledNoiseDiff}) = \frac{\senderU_1}{\senderU_2}$,
    where we have used the definition of $\widehat{\scaledNoiseDiff}$,
    and the function $\gamma(\cdot)$ is decreasing. 
    
    We now argue the feasibility of dual assignment 
    \eqref{eq: opt signaling SDSU binary}
    when No information revealing is optimal. 
    Note that to ensure that
    the inequalities \eqref{eq:SDSU dual feasibility state 1} and 
    and \eqref{eq:SDSU dual feasibility state 2} 
    hold for dual assignment 
    \eqref{eq: opt signaling SDSU binary},
    it is remaining to show that 
    \begin{align*}
        \frac{\senderU_2 (\Quant(\NOinforSignal) - \Quant(\scaledNoiseDiff))
        + \noiseDualVar(\NOinforSignal) (\NOinforSignal - \receiverU_2)}{\scaledNoiseDiff - \receiverU_2}
        \le 
        \frac{\senderU_1 (\Quant(\NOinforSignal) - \Quant(\scaledNoiseDiff) )
        + \noiseDualVar(\NOinforSignal) (\NOinforSignal - \receiverU_1)}{\scaledNoiseDiff - \receiverU_1}
    \end{align*}
    Rearranging the terms, it suffices to show that
    \begin{align*}
        \noiseDualVar(\NOinforSignal)
        & \ge 
        \frac{\senderU_2(\scaledNoiseDiff - \receiverU_1) - \senderU_1(\scaledNoiseDiff - \receiverU_2)}{(\receiverU_2-\receiverU_1)(\scaledNoiseDiff - \NOinforSignal)}\cdot (\Quant(\NOinforSignal) - \Quant(\scaledNoiseDiff)), ~ \forall \scaledNoiseDiff \in(\receiverU_1, \NOinforSignal]; \\ 
        \noiseDualVar(\NOinforSignal)
        & \le 
        \frac{\senderU_2(\scaledNoiseDiff - \receiverU_1) - \senderU_1(\scaledNoiseDiff - \receiverU_2)}{(\receiverU_2-\receiverU_1)(\scaledNoiseDiff - \NOinforSignal)}\cdot (\Quant(\NOinforSignal) - \Quant(\scaledNoiseDiff)), ~ \forall \scaledNoiseDiff \in(\NOinforSignal, \receiverU_2]~.
    \end{align*}
    Consider the function 
    $f(\scaledNoiseDiff) \triangleq \frac{\senderU_2(\scaledNoiseDiff - \receiverU_1) - \senderU_1(\scaledNoiseDiff - \receiverU_2)}{(\receiverU_2-\receiverU_1)(\scaledNoiseDiff - \NOinforSignal)}\cdot (\Quant(\NOinforSignal) - \Quant(\scaledNoiseDiff))$, 
    then we have 
    $\lim_{\scaledNoiseDiff \rightarrow \NOinforSignal} f(\scaledNoiseDiff)= -(\prior_1\senderU_1 + \prior_2\senderU_2) \cdot \Quant'(\NOinforSignal) = \noiseDualVar(\NOinforSignal)$.
    Furthermore, it can be shown that 
    $f(\scaledNoiseDiff) \le \noiseDualVar(\NOinforSignal), \forall \scaledNoiseDiff \in(\receiverU_1, \NOinforSignal]$, 
    and 
    $f(\scaledNoiseDiff) \ge \noiseDualVar(\NOinforSignal), \forall \scaledNoiseDiff \in(\NOinforSignal, \receiverU_2]$.

    \item 
    Fix an arbitrary $\scaledNoiseDiff \ge \receiverU_2$.
    Then by construction, 
    we have 
    \begin{align*}
        \noiseDualVar(\scaledNoiseDiff) = 0 
        \le
        \frac{\senderU_1(\Quant(\censorshipsignal)- \Quant(\scaledNoiseDiff)) + \noiseDualVar(\censorshipsignal)(\censorshipsignal -\receiverU_1) }{\scaledNoiseDiff - \receiverU_1}; ~
        \noiseDualVar(\scaledNoiseDiff) = 0 
        \le \frac{\senderU_2(\Quant(\censorshipsignal)- \Quant(\scaledNoiseDiff)) + \noiseDualVar(\censorshipsignal)(\censorshipsignal -\receiverU_2) }{\scaledNoiseDiff - \receiverU_2}~,
    \end{align*}
    which directly imply the inequalities \eqref{eq:SDSU dual feasibility state 1} and 
    \eqref{eq:SDSU dual feasibility state 2}.
\end{itemize}

\end{proof}

%% file: Paper/apx-robustness.tex

\subsection{Omitted Proof of 
\texorpdfstring{\Cref{prop:2 robust approx SISU lower bound}}{}}
\label{apx:proof 2 robust approx SISU lower bound}
\begin{proof}
Consider following problem instance with binary state (i.e., $m = 2$),
\begin{align*}
    \prior_1 = 1 - \frac{\varepsilon}{4-\varepsilon}, \qquad
    \prior_2 = \frac{\varepsilon}{4-\varepsilon}, \qquad
    \receiverU_1 = \log \left(\frac{\varepsilon}{4-\varepsilon}\right), \qquad
    \receiverU_2 = -\frac{\prior_1}{\prior_2}\receiverU_1.
\end{align*}
In this problem instance, the optimal censorship
$\optschemeInfty$ for a fully rational receiver is the
no-information revealing signaling scheme, 
in which both states are pooled at $\censorshipstateInfty = 0$.
Now, consider a receiver with bounded rationality level $\noiseScale = 1$.
Note that 
\begin{align*}
    \Payoff[1]{\optschemeInfty}
    = \Quant(0) = \frac{1}{2}~.
\end{align*}
On the other hand, the optimal expected sender utility 
$\Payoff[1]{\texttt{OPT}(1)}$
can be lower bounded by the full-information revealing signaling scheme,
i.e.,
\begin{align*}
    \Payoff[1]{\texttt{OPT}(1)}
    \ge \prior_1 \Quant(\receiverU_1) + \prior_2 \Quant(\receiverU_2)
    > \prior_1 \Quant(\receiverU_1)
    = \left(1 - \frac{\varepsilon}{4-\varepsilon}\right) \frac{1}{1+\frac{\varepsilon}{4-\varepsilon}}
    = 1-\frac{\varepsilon}{2}
\end{align*}
which completes the proof.
\end{proof}

\subsection{Omitted Proof of 
\texorpdfstring{\Cref{prop:unbounded robust direct approx SISU lower bound}}{}}
\label{apx:unbounded robust proof direct approx SISU lower bound}
\begin{proof}
Given a sufficiently small $\varepsilon > 0$,
consider following problem instance 
with three states (i.e., $m = 3$),
\begin{align*}
    \prior_1 = \varepsilon, \qquad
    \prior_2 = \prior_3 = \frac{1 - \varepsilon}{2}, \qquad
    \receiverU_1 = -0.01, \qquad
    \receiverU_2 = 0.01, \qquad
    \receiverU_3 = 3~.
\end{align*}
In this problem instance, the optimal 
direct signaling scheme
$\optdirectInfy$ for a fully rational receiver is 
characterized as follows
\begin{align*}
    \optdirectInfy_1(\delta) 
    & = \indicator{\delta = 0} \\
    \optdirectInfy_2(\delta)
    & = \frac{2\varepsilon}{1-\varepsilon}
    \indicator{\delta = 0}, \quad
    \optdirectInfy_2(\delta)
    = \left(1 - \frac{2\varepsilon}{1-\varepsilon}\right)
    \indicator{\delta = \frac{ \frac{0.01(1-3\varepsilon)}{2} + 
    \frac{3(1-\varepsilon)}{2}}{1-2\varepsilon}} \\
    \optdirectInfy_3(\delta)
    & = \indicator{\delta = \frac{ \frac{0.01(1-3\varepsilon)}{2} + 
    \frac{3(1-\varepsilon)}{2}}{1-2\varepsilon}}~.
\end{align*}
On the other hand, 
for a sufficiently small $\varepsilon$,
it can be shown that from \Cref{thm:SISU opt},
the optimal signaling scheme $\texttt{OPT}(\noiseScale)$ 
for any bounded rationality level $\noiseScale < \infty$
is full-information revealing, 
yielding 
$\Payoff[\noiseScale]{\texttt{OPT}(\noiseScale)}
= \varepsilon \Quant(-0.01) + \sfrac{\Quant(0.01)(1-\varepsilon)}{2}  + \sfrac{\Quant(3)(1-\varepsilon)}{2}$. 
Numerically, 
it can be verified that 
$\lim_{\noiseScale \rightarrow \infty} \lim_{\varepsilon \rightarrow 0}
\frac{\Payoff[\noiseScale]{\texttt{OPT}(\noiseScale)}}{\Payoff[\noiseScale]{\optdirectInfy}}
= \infty$,
which completes the proof.
\end{proof}

\subsection{Omitted Proof of 
\texorpdfstring{\Cref{thm:imposs robust approx SDSU}}{}}
\label{apx:imposs robust proof CHECK}
\impossrobustapproxSDSU*

\begin{proof}[Proof of \Cref{thm:imposs robust approx SDSU}]
For the ease of the presentation, 
let $\robustRatio' \triangleq \frac{1}{\robustRatio}$.
To analyze the rationality-robust approximation ratio $\robustRatio$ of a certain signaling scheme $\signalscheme$ over all possible $\noiseScale\in\rationalityClass$, 
we consider following factor-revealing program
\begin{align*}
    \label{eq:SDSU approx lp CHECK}
    \arraycolsep=5.4pt\def\arraystretch{1}
    \tag{$\mathcal{P}_\texttt{Factor-Revealing}$}
    &\begin{array}{lll}
     \max\limits_{\boldsymbol{\signalscheme} \geq \zerobf, \robustRatio' \geq 0} & \robustRatio'
       & \text{s.t.} 
     \\
       & 
       \prior_1\signalscheme_1(\scaledNoiseDiff)\cdot \left(\scaledNoiseDiff - \receiverU_1\right)
        +\prior_2 \signalscheme_2(\scaledNoiseDiff)\cdot \left(\scaledNoiseDiff - \receiverU_2\right) \ge 0
       & \scaledNoiseDiff\in(-\infty, \infty)
     \\
       & 
       \displaystyle\int_{-\infty}^{\infty} \signalscheme_i(\scaledNoiseDiff)d\scaledNoiseDiff = 1
        &  i\in[2]
     \\
       & 
       \signalscheme_i(\scaledNoiseDiff) \ge 0
         & \scaledNoiseDiff\in(-\infty, \infty), ~  i\in[2]
     \\
       & 
        \Payoff[\noiseScale_\ell]{\signalscheme}
        \ge \robustRatio' \Payoff[\noiseScale_\ell]{\optschemebetaell}, 
          &  \ell\in[L]
    \end{array}
\end{align*}


We first lower bound the optimal expected sender 
utility under the bounded rationality level $\noiseScale_\ell$:
$\Payoff[\noiseScale_\ell]{\optschemebetaell} = \Omega\left(\frac{1}{\noiseScale_\ell \cdot \exp(\noiseScale_\ell)}\right)$.
To see this, consider
following signaling scheme $\signalscheme'$:
\begin{align*}
    \signalscheme_1'\left(\frac{\noiseScale_\ell+2}{\noiseScale_\ell+1}\right) = 1; 
    \quad 
    \signalscheme_2'\left(\frac{\noiseScale_\ell+2}{\noiseScale_\ell+1}\right) = \frac{1}{\noiseScale_\ell},
    ~
    \signalscheme_2'\left(2\right) = 1 - \frac{1}{\noiseScale_\ell}~.
\end{align*}
Clearly, the above signaling scheme is a feasible solution
to the program~\ref{eq:opt lp} 
with the above constructed problem instance $\instance$.
Thus, we have
\begin{align*}
    \Payoff[\noiseScale_\ell]{\optschemebetaell} 
    \ge 
    \Payoff[\noiseScale_\ell]{\signalscheme'} 
    \ge \prior_2 \cdot \frac{1}{\noiseScale_\ell} \Quant\left(\frac{\noiseScale_\ell+2}{\noiseScale_\ell+1}\right)
    = \Omega\left(\frac{1}{\noiseScale_\ell\exp(\noiseScale_\ell)}\right)~.
\end{align*}

Recall that $\senderU_1 = 0$, and thus
$\Payoff[\noiseScale_\ell]{\signalscheme} = 
\prior_2\senderU_2 \int_{-\infty}^\infty \signalscheme_i(\scaledNoiseDiff) \Quant^{(\noiseScale_\ell)}(\scaledNoiseDiff)d\scaledNoiseDiff$.
To analyze the objective of the program \ref{eq:SDSU approx lp CHECK},
we consider following dual program~\ref{eq:SDSU approx dual lp CHECK}
of the program \ref{eq:SDSU approx lp CHECK}
with relaxing its fourth constraint to 
$\Payoff[\noiseScale_\ell]{\signalscheme}
\ge \robustRatio' \frac{1}{\noiseScale_\ell\exp(\noiseScale_\ell)}, 
\forall\ell\in[L]$:\footnote{The the duality of our infinite-dimensional LP can be obtained formally from Theorem 3.12 in \cite{AN-87}.}
\begin{align}
    \label{eq:SDSU approx dual lp CHECK}
    \arraycolsep=5.4pt\def\arraystretch{1}
    \tag{$\mathcal{P}_\texttt{FR-Dual}$}
    &\begin{array}{lll}
     \min\limits_{\boldsymbol\noiseDualVar,
     \boldsymbol\distDualVar,\boldsymbol\approxDualVar} & \distDualVar(1) + \distDualVar(2)
       & \text{s.t.} 
     \\
       & 
       \noiseDualVar(\scaledNoiseDiff) \cdot(1 - \scaledNoiseDiff) + 2\distDualVar(1) \ge 0
       &  \scaledNoiseDiff\in(-\infty, \infty)
     \\
       & 
       \noiseDualVar(\scaledNoiseDiff) \cdot(2 - \scaledNoiseDiff ) + 2\distDualVar(2) \ge \displaystyle\sum_{\ell=1}^L \approxDualVar(\ell) \cdot  \frac{1}{1+\exp(\noiseScale_\ell \scaledNoiseDiff)}
        &  \scaledNoiseDiff\in(-\infty, \infty)
     \\
       & 
       \displaystyle\sum_{\ell=1}^L \approxDualVar(\ell) \cdot  \frac{1}{\noiseScale_\ell\exp(\noiseScale_\ell)} \ge 1
         & 
     \\
       & 
       \noiseDualVar(\scaledNoiseDiff) \ge 0, 
          & \scaledNoiseDiff\in(-\infty, \infty)
     \\
       & 
       \approxDualVar(\ell)\ge 0
        &  \ell\in[L]
    \end{array}
\end{align}

Below we construct an assignment 
for the dual variables $\distDualVar(1), \distDualVar(2)$
and $\{\approxDualVar(\ell)\}$
for the dual program \ref{eq:SDSU approx dual lp CHECK} and show that
together with an assignment 
of $\{\noiseDualVar(\scaledNoiseDiff)\}$,
our constructed assignment is feasible
for sufficiently large $L$.
Consider the following dual assignment
of $\distDualVar(1), \distDualVar(2)$
and $\{\approxDualVar(\ell)\}$,
\begin{align*}
    \distDualVar(1)  \gets \frac{3}{L}
    \quad
    \distDualVar(2)  \gets \frac{2}{L}; 
    \qquad
    \approxDualVar(\ell) \gets \frac{1}{L}
    (\noiseScale_\ell \exp(\noiseScale_\ell))
    \quad\forall \ell\in[L]
\end{align*}
Note that
the dual constraint for primal variable $\robustRatio$
is satisfied by construction.
Next, we discuss
how to construct
the assignment for $\{\noiseDualVar(\scaledNoiseDiff)\}$.
We consider the three cases separately:
$\scaledNoiseDiff \leq 1$, $\scaledNoiseDiff \geq 2$
and $\scaledNoiseDiff\in(1, 2)$.
\begin{itemize}[label={-}]

    \item 
    For every $\scaledNoiseDiff \leq 1$, 
    let $\noiseDualVar(\scaledNoiseDiff) = \infty$
    is sufficient to satisfies the dual constraints
    for $\signalProb_1(\scaledNoiseDiff)$
    and $\signalProb_2(\scaledNoiseDiff)$.
    
    \item 
    For every $\scaledNoiseDiff \geq 2$,
    let $\noiseDualVar(\scaledNoiseDiff) = 0$
    is sufficient to satisfies the dual constraints
    for $\signalProb_1(\scaledNoiseDiff)$
    and $\signalProb_2(\scaledNoiseDiff)$.
    To see this, note that the dual constraints for $\signalProb_1(\scaledNoiseDiff)$ holds straightforwardly
    as $\distDualVar(1) = \sfrac{3}{L} \ge 0$.
    To satisfy the dual constraints for $\signalProb_2(\scaledNoiseDiff)$,
    it is sufficient to show 
    $\distDualVar(2) \ge  \frac{1}{2}\cdot\sum_{\ell=1}^L  \frac{\approxDualVar(\ell)  }{1+\exp(\noiseScale_\ell \scaledNoiseDiff)}$,
    which holds for sufficiently large $L$.
    To see this, note that
    \begin{align*}
      \frac{1}{2}\cdot \sum_{\ell=1}^L  \frac{\approxDualVar(\ell) }{1+\exp(\noiseScale_\ell \scaledNoiseDiff)}
      =
      \frac{1}{2}\cdot \sum_{\ell=1}^L  
      \frac{1}{L}\frac{
      \noiseScale_\ell \exp(\noiseScale_\ell) 
      }{1+\exp(\noiseScale_\ell \scaledNoiseDiff)}
      & \overset{(a)}{\leq}
      \frac{1}{2}\cdot \sum_{\ell=1}^L  
      \frac{1}{L}
      \frac{
      \noiseScale_\ell \exp(\noiseScale_\ell) 
      }{1+\exp(2\noiseScale_\ell)}
      \overset{(b)}{\leq} 
      \frac{1}{2}\cdot \sum_{\ell=1}^L  
      \frac{1}{L}
      \frac{
      1
      }{L}
      \\
      &
      = \frac{1}{2L}
      \leq \distDualVar(2)
    \end{align*}
    
    \item 
Now we consider $\scaledNoiseDiff \in (1, 2)$.
To satisfies the dual constraints 
for $\signalProb_1(\scaledNoiseDiff)$
and $\signalProb_2(\scaledNoiseDiff)$,
it is sufficient to show 
    $\frac{2\distDualVar(1)}{\scaledNoiseDiff-1} \ge 0$,
which holds straightforwardly as $\scaledNoiseDiff > 1$
as $\distDualVar(1) 
=\sfrac{3}{L} > 0$,
and 
\begin{align*}
    \frac{1}{2-\scaledNoiseDiff} \cdot \left(\displaystyle\sum_{\ell=1}^L \approxDualVar(\ell) \cdot  \frac{1}{1+ \exp(\noiseScale_\ell\scaledNoiseDiff)} -2\distDualVar(2)\right)
    \le \frac{2\distDualVar(1)}{\scaledNoiseDiff-1}~.
\end{align*}
Rearranging the above inequality, we have 
\begin{align}
\label{eq:existence feasible alpha between 1 2 CHECK}
    \distDualVar(1) + \frac{\scaledNoiseDiff-1}{2-\scaledNoiseDiff }\cdot \distDualVar(2)  
    & \ge \frac{1}{2}\cdot  \frac{\scaledNoiseDiff-1}{2-\scaledNoiseDiff} \cdot \sum_{\ell=1}^L   \frac{\approxDualVar(\ell) }{1+\exp(\noiseScale_\ell \scaledNoiseDiff)}~.
\end{align}
To see why the above inequality \eqref{eq:existence feasible alpha between 1 2 CHECK} 
holds true, we consider two subcases
$\scaledNoiseDiff\in[\sfrac{3}{2}, 2)$ 
and $\scaledNoiseDiff\in(1, \sfrac{3}{2}]$ separately.

Fix an arbitrary $\scaledNoiseDiff \in [\sfrac{3}{2},2)$.
Note that
for sufficiently large $L$
\begin{align*}
    \text{Left-hand side of \eqref{eq:existence feasible alpha between 1 2 CHECK}}
    &\geq \frac{\scaledNoiseDiff - 1}{2-\scaledNoiseDiff}\distDualVar(2) 
    \geq \frac{1}{2-\scaledNoiseDiff}\frac{1}{L}
    \\
    \text{Right-hand side of \eqref{eq:existence feasible alpha between 1 2 CHECK}}
    &
    \overset{(a)}{\leq} 
    \frac{1}{2}
    \frac{1}{2-\scaledNoiseDiff}
    \sum_{\ell=1}^L
    \frac{1}{L}
    \frac{\noiseScale_\ell\exp(\noiseScale_\ell)}{1 + \exp(\frac{3}{2}\noiseScale_\ell)}
    \\
    &
    \overset{(b)}{\leq} 
    \frac{1}{2}
    \frac{1}{2-\scaledNoiseDiff}
    \sum_{\ell=1}^L
    \frac{1}{L}
    \frac{1}{L}
    =
    \frac{1}{2-\scaledNoiseDiff}\frac{1}{2L}
\end{align*}
where
inequality~(a) holds since $\scaledNoiseDiff \in [\sfrac{3}{2}, 2)$;
and inequality~(b) holds since
$\frac{\noiseScale_\ell\exp(\noiseScale_\ell)}{1 + \exp(\frac{3}{2}\noiseScale_\ell)} \leq \frac{1}{L}$
for sufficiently large $L$.

Fix an arbitrary $\scaledNoiseDiff \in (1, \sfrac{3}{2}]$.
Let $k\in \naturals$ be the 
index such that $\scaledNoiseDiff \in 
\left[1 + \frac{1}{L^{k + \frac{1}{2}}},
1+ \frac{1}{L^{k - \frac{1}{2}}}\right]$.
Note that
for sufficiently large $L$,
\begin{align*}
    \text{Left-hand side of \eqref{eq:existence feasible alpha between 1 2 CHECK}}
    &\geq \distDualVar(1)
    = \frac{3}{L}
    \\
    \text{Right-hand side of \eqref{eq:existence feasible alpha between 1 2 CHECK}}
    &
    \overset{(a)}{\leq} 
    (\scaledNoiseDiff - 1)
    \sum_{\ell=1}^L
    \frac{1}{L}
    \frac{\noiseScale_\ell\exp(\noiseScale_\ell)}{1 + \exp(\scaledNoiseDiff\noiseScale_\ell)}
    \\
    & = 
    (\scaledNoiseDiff - 1)
    \sum_{\ell\in[L]: \ell\not= k}
    \frac{1}{L}
    \frac{\noiseScale_\ell\exp(\noiseScale_\ell)}{1 + \exp(\scaledNoiseDiff\noiseScale_\ell)}
    +
    (\scaledNoiseDiff - 1)
    \frac{1}{L}
    \frac{L^k\exp(L^k)}{1 + \exp(\scaledNoiseDiff L^k)}
    \\
    &
    \overset{(b)}{\leq}
      \sum_{\ell\in[L]: \ell\not= k}
    \frac{1}{L}
    \frac{1}{L}
    +
    \frac{1}{e}
    \frac{1}{L}
    \leq \frac{2}{L}
\end{align*}
where inequality~(a) holds since
$\scaledNoiseDiff \in (1, \sfrac{3}{2}]$,
and 
inequality~(b) holds 
since
$(\scaledNoiseDiff - 1)
\frac{\noiseScale_\ell\exp(\noiseScale_\ell)}{1 + \exp(\scaledNoiseDiff\noiseScale_\ell)}
\leq \frac{1}{L}$
for every $\ell\not= k$
when $L$ is sufficiently large,
and 
$(\scaledNoiseDiff - 1)
\frac{L^k\exp(L^k)}{1 + \exp(\scaledNoiseDiff L^k)}
\leq \frac{1}{e}$.
\end{itemize}
Hence, we have the optimal objective value of the program \ref{eq:SDSU approx lp CHECK}
is at most $\distDualVar(1) + \distDualVar(2) = \sfrac{5}{L}$.
As a result,
the rationality-robust approximation ratio $\robustRatio(\signalscheme, \rationalityClass)$ is at least $\sfrac{L}{5}$ for any signaling scheme $\signalscheme$.
The proof now completes. 
\end{proof}

\ifEC
\subsection{Positive Results for Binary-state Instances}
\label{apx:SDSU binary robust proof}

\input{Paper/proof-prop-5-5.tex}

\fi

\subsection{Omitted Proof of 
\texorpdfstring{\Cref{lem:SDSU binary approx with K}}{}}
\label{apx:SDSU binary approx with K proof}
\SDSUbinaryapproxwithK*
\begin{proof}
Since the expected sender utility generated from state $1$
in the optimal signaling scheme $\optscheme$
is at most $\Payoff{\signalscheme\doubleprimed}$,
it is sufficient to show 
that the expected sender utility generated from state $2$
in the optimal signaling scheme $\optscheme$
is at most $16eK\Payoff{\signalscheme\primed}$.
This can be further reformulated as showing 
$\Payoff{\optscheme} \leq 16eK \Payoff{\signalscheme\primed}$
for all problem instances with sender utility $\senderU_1 = 0$.
We further assume $\receiverU_2 > 0$.\footnote{When
sender utility $\senderU_1 = 0$
and receiver utility difference $\receiverU_2 \leq 0$,
signaling scheme $\signalscheme\primed$
becomes the no-information revealing signaling scheme,
which is indeed optimal due to the concavity
of function $\Quant(\cdot)$ on $(-\infty, 0]$.}
We show this inequality using the weakly duality of linear program~\ref{eq:opt lp}
with its dual program~\ref{eq:opt lp dual}.

First, consider following assignment 
for dual variables $\{\distDualVar(1), \distDualVar(2)\}$,
\begin{align*}
    \distDualVar(1) 
    = 
    \distDualVar(2) 
    =  
    8eK\,
    \Quant(\censorshipsignal)
\end{align*}
where $\censorshipprob = 
\frac{\prior_1(\censorshipsignal -
\receiverU_1)}{\prior_2(\receiverU_2 - \censorshipsignal)}$
is the threshold state probability 
in signaling scheme $\signalscheme\primed$.

Below we argue that there exists an assignment 
for dual variables 
$\{\noiseDualVar(\scaledNoiseDiff)\}$,
which together with the constructed dual assignment 
for $\{\distDualVar(1),\distDualVar(2)\}$ above
is feasible.
For every $\scaledNoiseDiff\in(-\infty,\infty)$,
there are two dual constraints related to 
$\{\noiseDualVar(\scaledNoiseDiff)\}$,
\begin{align*}
    \prior_1
      \left(\receiverU_1 - \scaledNoiseDiff \right)
      \noiseDualVar(\scaledNoiseDiff)
      +
      \distDualVar(1)
      \geq 
      \prior_1 \senderU_1 \Quant(\scaledNoiseDiff);
      \quad 
    \prior_2
      \left(\receiverU_2 - \scaledNoiseDiff \right)
      \noiseDualVar(\scaledNoiseDiff)
      +
      \distDualVar(2)
      \geq 
      \prior_2 \senderU_2 \Quant(\scaledNoiseDiff)
\end{align*}
If $\scaledNoiseDiff\in(-\infty,\receiverU_1]$
(resp.\ $\scaledNoiseDiff\in[\receiverU_2,\infty)$),
setting $\noiseDualVar(\scaledNoiseDiff) = -\infty$
(resp.\ $\noiseDualVar(\scaledNoiseDiff) = \infty$)
satisfies the dual constraints.
If $\scaledNoiseDiff\in(\receiverU_1,\receiverU_2)$,
plugging the assignment 
for dual variables 
$\{\distDualVar(1), \distDualVar(2)\}$
constructed above as well as $\senderU_1 = 0$,
the two dual constraints are equivalent to
\begin{align}
    \label{eq: approx verify with K}
    \frac{\prior_2 (\receiverU_2 - \scaledNoiseDiff)}
    {\prior_1(\scaledNoiseDiff - \receiverU_1)}
    \censorshipprob
    \Quant(\censorshipsignal)
    + \censorshipprob\Quant(\censorshipsignal)
    \geq
    \frac{1}{8eK}
    \Quant(\scaledNoiseDiff)
\end{align}
Let $\NOinforSignal \triangleq (\prior_1 \receiverU_1 + \prior_2 \receiverU_2)$.
To establish the above inequality \eqref{eq: approx verify with K}
with the dual assignment of $\{\distDualVar(1), \distDualVar(2)\}$, 
we analyze two cases 
(i) $\NOinforSignal >
\max\{\receiverU_1, 0\} + \sfrac{1}{(K\noiseScale)}$
and 
(ii) $\NOinforSignal \leq
\max\{\receiverU_1, 0\} + \sfrac{1}{(K\noiseScale)}$
separately.

Suppose $\NOinforSignal > 
\max\{\receiverU_1, 0\} + \sfrac{1}{(K\noiseScale)}$,
and thus $\censorshipsignal = \max\{\receiverU_1, 0\} + \sfrac{1}{(K\noiseScale)}$.
Here we consider
following three subcases based on different values of $\delta$:
\begin{itemize}[label={-}]
    \item 
    Fix an arbitrary
    $\scaledNoiseDiff \in(\receiverU_1, \censorshipsignal]$.
    Note that 
    \begin{align*}
    \frac{\prior_2 (\receiverU_2 - \scaledNoiseDiff)}
    {\prior_1(\scaledNoiseDiff - \receiverU_1)}
    \censorshipprob
    \Quant(\censorshipsignal)
    + \censorshipprob\Quant(\censorshipsignal)
    &\geq 
    \frac{\prior_2 (\receiverU_2 - \scaledNoiseDiff)}
    {\prior_1(\scaledNoiseDiff - \receiverU_1)}
    \censorshipprob
    \Quant(\censorshipsignal)
    \overset{(a)}{\geq}
    \frac{1}{\censorshipprob}\censorshipprob\Quant(\censorshipsignal)
    \overset{(b)}{\geq}
    \frac{1}{1+e}\Quant(\scaledNoiseDiff)
    \end{align*}
    where inequality~(a) holds since $\sfrac{\prior_2 (\receiverU_2 - \scaledNoiseDiff)}
    {\prior_1(\scaledNoiseDiff - \receiverU_1)} 
    \geq \sfrac{1}{\censorshipprob}$
    due to the construction of $\censorshipprob$;
    and inequality~(b) holds since $\Quant(\censorshipsignal)\geq 
    \sfrac{\Quant(\scaledNoiseDiff)}{(1+\exp(\adjustedConst))}$
    due to the construction of $\censorshipsignal$
    and $K\geq 1$

    \item 
    Fix an arbitrary $\scaledNoiseDiff \in(\censorshipsignal, 
    \sfrac{(\max\{\receiverU_1, 0\}  + \receiverU_2)}
    {2}]$.
    Note that 
    \begin{align*}
    \frac{\prior_2 (\receiverU_2 - \scaledNoiseDiff)}
    {\prior_1(\scaledNoiseDiff - \receiverU_1)}
    \censorshipprob
    \Quant(\censorshipsignal)
    + \censorshipprob\Quant(\censorshipsignal)
    \geq&
    \frac{\prior_2 (\receiverU_2 - \scaledNoiseDiff)}
    {\prior_1(\scaledNoiseDiff - \receiverU_1)}
    \censorshipprob
    \Quant(\censorshipsignal) 
    \\
    =&
    \frac{\prior_2 (\receiverU_2 - \scaledNoiseDiff)}
    {\prior_1(\scaledNoiseDiff - \receiverU_1)}
    \frac{\prior_1(\censorshipsignal - \receiverU_1)}
    {\prior_2 (\receiverU_2 - \censorshipsignal)}
    \Quant(\censorshipsignal)
    \\
    =&
    \frac{(\receiverU_2 - \scaledNoiseDiff)}
    {(\receiverU_2 - \censorshipsignal)}
    \frac{(\censorshipsignal - \receiverU_1)}
    {(\scaledNoiseDiff - \receiverU_1)}
    \Quant(\censorshipsignal)
    \\
    \overset{(a)}{\geq}&
    \frac{1}{2}
    \frac{1}{2e}
    \frac{1}{K\noiseScale}
    \frac{1}{\scaledNoiseDiff - \receiverU_1}
    \frac{1}{\exp(\noiseScale\receiverU_1)}
    \\
    =&
    \frac{\exp(\noiseScale\scaledNoiseDiff)}
    {2\noiseScale(\scaledNoiseDiff - \receiverU_1)\exp(\noiseScale\receiverU_1)
    }
    \cdot 
    \frac{1}{2eK} \frac{1}{\exp(\noiseScale\scaledNoiseDiff)}
    \overset{(b)}{\geq}
    \frac{1}{2eK} \Quant(\scaledNoiseDiff)
    \end{align*}
    where inequality~(a) holds since
    $\sfrac{(\receiverU_2 - \scaledNoiseDiff)}{(\receiverU_2 - \censorshipsignal)} \geq \sfrac{1}{2}$,
    $\censorshipsignal-\receiverU_1 \geq  \sfrac{1}{(K\noiseScale)}$,
    and
    $2\exp(\adjustedConst)\Quant(\censorshipsignal) \geq  \sfrac{1}{\exp(\noiseScale\receiverU_1)}$.
    To see why inequality~(b) holds,
    first note that 
    $\sfrac{1}{\exp(\noiseScale\scaledNoiseDiff)}
    \geq \Quant(\scaledNoiseDiff)$.
    Additionally, 
    $\noiseScale(\scaledNoiseDiff-\receiverU_1) \geq 1$,
    and
    thus
    $\sfrac{\exp(\noiseScale\scaledNoiseDiff)}
    {(2\noiseScale(\scaledNoiseDiff - \receiverU_1)\exp(\noiseScale\receiverU_1))
    } \geq \sfrac{\exp(1)}{2} \geq 1$.
    
    \item Fix an arbitrary $\scaledNoiseDiff\in
    (\sfrac{(\max\{\receiverU_1, 0\}  + \receiverU_2)}{2}, \receiverU_2)$.
    Note that 
    \begin{align*}
    \frac{
    \frac{\prior_2 (\receiverU_2 - \scaledNoiseDiff)}
    {\prior_1(\scaledNoiseDiff - \receiverU_1)}
    \censorshipprob
    \Quant(\censorshipsignal)
    + \censorshipprob\Quant(\censorshipsignal)}
    {
    \frac{1}{8eK}\Quant(\scaledNoiseDiff)}
    & \geq 
    8eK\censorshipprob
    \frac{\Quant(\censorshipsignal)}
    {\Quant(\scaledNoiseDiff)}\\
    & \overset{(a)}{\geq}
    4eK\,
    \frac{\prior_1(\censorshipsignal - \receiverU_1)}
    {\prior_2 (\receiverU_2 - \censorshipsignal)}
    \frac{\exp(\noiseScale(\scaledNoiseDiff-\max\{\receiverU_1, 0\} ))}{ \exp(1/K)}
    \\
    &\overset{(b)}{\geq}
    4K
    \frac{\prior_1\frac{1}{K\noiseScale}}{\prior_2
    (\receiverU_2 - \max\{\receiverU_1, 0\} )}
    {\left(\noiseScale(\scaledNoiseDiff-\max\{\receiverU_1, 0\} )\right)^2}
    \\
    &
    \overset{(c)}{\geq}
    \frac{\prior_1}{\prior_2}
    \noiseScale{(\receiverU_2 - \max\{\receiverU_1, 0\} )}
    \overset{(d)}{\geq} 1
    \end{align*}
    where inequality~(a) holds since 
    $2\sfrac{\Quant(\censorshipsignal)}
    {\Quant(\scaledNoiseDiff)} \geq \exp(\noiseScale(\scaledNoiseDiff-
    \max\{\receiverU_1, 0\} - \sfrac{1}{(K\noiseScale)}))$
    and the definition of $\censorshipprob$;
    inequality~(b) holds since
    $\censorshipsignal-\receiverU_1 \geq \sfrac{1}{(K\noiseScale)}$,
    $\receiverU_2 - \censorshipsignal \leq 
    \receiverU_2 - \max\{\receiverU_1, 0\} $,
    $\exp(\noiseScale(\scaledNoiseDiff-\max\{\receiverU_1, 0\} )) \geq 
    \left(\noiseScale(\scaledNoiseDiff-\max\{\receiverU_1, 0\} )\right)^2
    $;
    inequality~(c) holds since $\scaledNoiseDiff \geq 
    \sfrac{(\max\{\receiverU_1, 0\}  + \receiverU_2)}{2}$;
    and inequality~(d) holds since we assume that 
    $\noiseScale \geq \sfrac{\prior_2}{(\prior_1(\receiverU_2 - \max\{\receiverU_1, 0\} ))}$
    and $\sfrac{\prior_2}{(\prior_1(\receiverU_2 - \max\{\receiverU_1, 0\} ))}\geq 0$
    due to $\receiverU_2 > 0$.
\end{itemize}

Suppose $\NOinforSignal \leq  
\max\{\receiverU_1, 0\} + \sfrac{1}{(K\noiseScale)}$,
and thus $\censorshipsignal = \NOinforSignal$.
Note that in this case, the threshold state probability
$\censorshipprob = 1$, and 
inequality~\eqref{eq: approx verify with K}
can be further simplified and relaxed as 
$
\Quant(\censorshipsignal) \geq
\sfrac{\Quant(\scaledNoiseDiff)}{(1+e)} $
which holds for every $\scaledNoiseDiff\in(\receiverU_1,\censorshipsignal]$
due to the same argument as the previous case;
and for every 
$\scaledNoiseDiff\in[\censorshipsignal,\receiverU_2)$
due to the monotonicity of function $\Quant(\cdot)$.

Finally, recall that the expected sender utility 
in signaling scheme $\Payoff{\signalscheme\primed}$ is $
\prior_2\senderU_2(\censorshipprob\Quant(\censorshipsignal)
+
(1 - \censorshipprob)\Quant(\receiverU_2))$,
which is a 
$(16eK)$-approximation
to the objective value of the constructed dual assignment,
$i.e., \distDualVar(1) + \distDualVar(2) = 
16eK
\prior_2\senderU_2\censorshipprob\Quant(\censorshipsignal)$.
Invoking the weak duality of linear program finishes the proof.
\end{proof}

%% file: Paper/proof-prop-5-5.tex
\begin{proposition}
\label{thm:SDSU binary robust}
In SDSU environments, for any 
problem instance $\instance 
= (m, \{\prior_i\}, \{\receiverU_i\}, \{\senderU_i\})$
with binary state (i.e., $m = 2$),
for any $K \geq 1$ and $\noiseScaleLB \geq \frac{\prior_2}{\prior_1} \cdot \frac{1}{\receiverU_2 - \max\{\receiverU_1, 0\}}
\cdot \indicator{\receiverU_2 > 0}$,
there exists a signaling scheme $\signalscheme$
such that 
$\robustRatio(\signalscheme, [\noiseScaleLB, K\noiseScaleLB])   
\leq 
\left(4\sqrt{eK} + 1\right)^2$.
\end{proposition}

Here we sketch the high-level idea for the proof 
of \Cref{thm:SDSU binary robust}.
Note that by definition, a rationality-robust signaling scheme 
with rationality-robust approximation ratio $\robustRatio$
must imply that it is also $\robustRatio$-approximately optimal
to the optimal signaling scheme $\texttt{OPT}(\noiseScale)$ under every possible 
rationality level $\noiseScale\in\rationalityClass$.
Hence, to identify a robust signaling scheme,
ideally, one needs to understand
how does the optimal sender expected utility 
(or a non-trivial upper bound of it)
change when receiver's bounded rationality level changes. 
However, it is difficult to characterize the optimal 
sender expected utility\footnote{Although in 
\Cref{lem:SDSU binary opt}, we show that optimal signaling scheme $\texttt{OPT}(\noiseScale)$ for binary-state instances 
in SDSU environments is a censorship signaling scheme, 
the pooling signal no longer admits a simple structure as in SISU environments (see \Cref{lem:SDSU binary opt tmp}). 
Moreover, unlike in SISU environments where the pooling signal  is monotone with respect to the rationality level, 
here, the pooling signal of optimal signaling scheme $\texttt{OPT}(\noiseScale)$ is no longer monotone.}, 
let alone to understand its behavior over different rationality levels.
We tackle this challenge 
by first showing that a censorship signaling scheme 
whose pooling signal (i.e., roughly at $\receiverU_1 + \Theta(\sfrac{1}{\noiseScale})$)
depends on the value of receiver's rationality level  
is approximately optimal to the optimal signaling scheme $\texttt{OPT}(\noiseScale)$.
With this structure of censorship signaling scheme,
a robust signaling scheme can be constructed 
by fine-tuning the location of the pooling signal.



\begin{restatable}{lemma}{SDSUbinaryapproxwithK}
\label{lem:SDSU binary approx with K}
In SDSU environments,
for any problem instance $\instance 
= (m, \{\prior_i\}, \{\receiverU_i\}, \{\senderU_i\})$
with binary state (i.e., $m = 2$),
for any receiver with bounded rationality level 
$\noiseScale$,
if $\noiseScale\geq \frac{\prior_2}{\prior_1} \cdot \frac{1}{\receiverU_2 - \max\{\receiverU_1, 0\}}
\cdot \indicator{\receiverU_2 > 0}$,
for any $\adjustedConst \geq 1$,
the optimal expected sender utility $\Payoff{\texttt{OPT}(\noiseScale)}$
is at most 
\begin{align*}
    \Payoff{\texttt{OPT}(\noiseScale)}
    \leq 
    16e\adjustedConst
    \Payoff{\signalscheme\primed} + 
    \Payoff{\signalscheme\doubleprimed}~,
\end{align*}
where 
signaling scheme
$\signalscheme\primed$ is a censorship signaling scheme
with threshold state $\censorshipstate = 2$
and pooling signal $\censorshipsignal = 
\min\{
\prior_1 \receiverU_1 + \prior_2 \receiverU_2, \max\{\receiverU_1, 0\}+ \sfrac{1}{(\adjustedConst\noiseScale})\}$,
and signaling scheme $\signalscheme\doubleprimed$ is 
the full-information revealing signaling scheme
(i.e., also a censorship signaling scheme).
\end{restatable}
\Cref{thm:SDSU binary robust} 
is a direct implication of \Cref{lem:SDSU binary approx with K},
whose proof is deferred to 
\Cref{apx:SDSU binary approx with K proof}.
\begin{proof}[Proof of \Cref{thm:SDSU binary robust}]
Let $\signalscheme\primed$ be the signaling scheme
with threshold state 
$\censorshipstate = 2$
and pooling signal $\censorshipsignal = 
\min\{\prior_1 \receiverU_1 + \prior_2 \receiverU_2, \max\{\receiverU_1, 0\}+ \sfrac{1}{(K\noiseScaleLB)}  \}$, 
and $\signalscheme\doubleprimed$ be the full-information 
revealing signaling scheme. 
Now construct 
signaling scheme 
$\signalscheme\triangleq
q\signalscheme\primed + 
(1- q)\signalscheme\doubleprimed$ 
as the convex combination 
of signaling scheme $\signalscheme\primed$
and $\signalscheme\doubleprimed$.
We specify the convex combination factor $q\in(0,1)$ in the end of the analysis.
By construction,
for any bounded rationality level $\noiseScale$,
$\Payoff[\noiseScale]{\signalscheme}
=
{q\Payoff[\noiseScale]{\signalscheme\primed}} + 
{(1-q)\Payoff[\noiseScale]{\signalscheme\doubleprimed}}$.
To see the rationality-robustness of signaling scheme $\signalscheme$, 
consider a receiver with an arbitrary bounded
rationality level $\noiseScale \in [\noiseScaleLB, K\noiseScaleLB]$, 
and note that 
\begin{align*}
    \Payoff[\noiseScale]{\texttt{OPT}(\noiseScale)}
    & \overset{(a)}{\le} 
    16e\frac{K\noiseScaleLB}{\noiseScale}\Payoff[\noiseScale]{\signalscheme\primed} +
    \Payoff[\noiseScale]{\signalscheme\doubleprimed}
    \\
    & \overset{(b)}{\le} 
    16eK\Payoff[\noiseScale]{\signalscheme\primed} +
    \Payoff[\noiseScale]{\signalscheme\doubleprimed}
    \leq 
    \left(\frac{16eK}{q} + \frac{1}{1-q}
    \right)
    \Payoff[\noiseScale]{\signalscheme}
\end{align*}
where inequality (a) holds from \Cref{lem:SDSU binary approx with K}, 
and inequality (b) holds since $\sfrac{\noiseScaleLB}{\noiseScale} \leq 1$.
We finishes the proof by
letting the convex combination factor $q$
minimizes $\sfrac{(16eK)}{q} + \sfrac{1}{(1-q)}$.
\end{proof}

%% file: Paper/apx-complexity.tex
\newcommand{\poly}{\cc{poly}}
\newcommand{\discretizedSupp}{\mathcal{S}}
\newcommand{\discreOptScheme}{\widehat{\pi}^*}
\newcommand{\newscheme}{\widehat{\pi}}
\newcommand{\maxSenderU}{\bar{u}}

\newcommand{\leftDelta}{\scaledNoiseDiff_L}
\newcommand{\rightDelta}{\scaledNoiseDiff_R}

\newcommand{\approxFactor}{\varepsilon}
\newcommand{\discrePara}{\varepsilon}

In this section, we discuss the complexity
on computing an approximately optimal signaling scheme 
in both SISU and SDSU environments. 

\begin{proposition}
\label{prop:complexity SISU opt}
In SISU environments, there exists a $\poly(m)$ time
algorithm that can find the optimal signaling scheme.
\end{proposition}
\begin{proof}
Recall that by \Cref{thm:opt state independent}, the optimal 
signaling scheme in SISU environments is a censorship
signaling scheme. Thus, to find the optimal signaling scheme
in SISU environments, it suffices to find the 
threshold state $\censorshipstate$ and the threshold state probability 
$\censorshipprob$.
To identify $\censorshipstate, \censorshipprob$, consider the following procedure:
for every state $i\in[m]$ where $\receiverU_i \ge 0$, compute the corresponding 
$p_i$ where $p_i$ is defined as in \eqref{eq: pooled prob defn}. 
Then the threshold state $\censorshipstate = \argmax_{i:\receiverU_i > 0}\{p_i\}$.
If $p_{\censorshipstate} > 1$, then 
the threshold state probability 
$\censorshipprob = 1$, otherwise $\censorshipprob = p_{\censorshipstate}$.
It is easy to see that the above procedure has 
the complexity at most $O(m)$. 
\end{proof}
\wtr{
Unlike in SISU environments, determining the computational complexity of finding the optimal signaling scheme in SDSU environments is much more challenging. One reason for this is the lack of a clear structure for the optimal signaling scheme in SDSU environments. 
Nonetheless, we present two complexity characterizations for finding the optimal signaling scheme in SDSU environments. The first one applies to special instances with binary states, 
while the second one applies to general problem instances.
}
\begin{corollary}
In SDSU environments with binary states, 
there exists a $O(1)$ time algorithm that can find 
the optimal signaling scheme.
\end{corollary}
\begin{proof}
By \Cref{lem:SDSU binary opt}, we know that 
the optimal signaling scheme in SDSU environments with binary states
is also a censorship signaling scheme, then the above 
result immediately follows by the similar analysis
of \Cref{prop:complexity SISU opt}.
\end{proof}

\begin{proposition}
\label{prop:compelxity of optimal SDSU}
In SDSU environments, there exists a 
$\poly(m, \sfrac{\noiseScale(\receiverU_m - \receiverU_1)}{\varepsilon})$ time
algorithm that can find a $(1+\varepsilon)$-approximate signaling scheme.
\end{proposition}
\begin{proof}
Given an arbitrary small $\varepsilon_0 > 0$, we discuss how to solve a $(1+\varepsilon_0)$-approximate signaling scheme with running time 
$\poly(m, \sfrac{\noiseScale(\receiverU_m - \receiverU_1)}{\varepsilon_0})$.

Let $\discrePara = \sfrac{\approxFactor_0}{\noiseScale}$.
Define the set $\discretizedSupp \triangleq \{\receiverU_0, \receiverU_0 + \discrePara, \receiverU_0 + 2\discrePara, \ldots, \receiverU_m\} \cup\{\receiverU_i\}_{i\in[m]}$.
We consider the following program and its optimal solution $\discreOptScheme$:
\begin{align}
    \label{eq:opt lp - approx}
    \arraycolsep=5.4pt\def\arraystretch{1}
    \tag{$\mathcal{P}_\texttt{OPT-Primal}(\discrePara)$}
    &\begin{array}{lll}
     \discreOptScheme = \argmax\limits_{\boldsymbol\signalscheme\geq \zerobf} ~ &
     \displaystyle\sum\nolimits_{i\in[m]} \prior_i \senderU_i 
     \displaystyle \sum_{\scaledNoiseDiff\in\discretizedSupp} \signalscheme_i(\scaledNoiseDiff)
     \Quant(\scaledNoiseDiff) 
     \quad& \text{s.t.} 
     \vspace{1mm}
     \\
       & 
       \displaystyle\sum\nolimits_{i\in[m]}  \prior_i 
      \left(\receiverU_i - \scaledNoiseDiff \right)
       \signalscheme_i(\scaledNoiseDiff)
       = 0
       & \scaledNoiseDiff\in \discretizedSupp
     \vspace{1mm}
       \\
       &
       \displaystyle\sum_{\scaledNoiseDiff\in\discretizedSupp} \signalscheme_i(\scaledNoiseDiff)  = 1
       &
       i\in[m]
     \vspace{1mm}
     \\
    \end{array}
\end{align}
Essentially, the above program restricts the support of each 
conditional distribution $\pi_i, i\in[m]$ to be the subset of the set $\discretizedSupp$. 
Recall that from \Cref{thm:SDSU 4 approx}, we know that 
there exists an optimal signaling scheme $\optscheme$ such that 
for every $\scaledNoiseDiff\in (-\infty, \infty)$, we have 
$|\{i: \optscheme_i(\scaledNoiseDiff) > 0, i\in [m]\}| \le 2$. 
Based on the signaling scheme $\optscheme$, we 
below construct a new signaling scheme $\newscheme$
that is also a feasible solution to the program \ref{eq:opt lp - approx}. In particular, for every $\scaledNoiseDiff$
where $|\{i: \optscheme_i(\scaledNoiseDiff) > 0, i\in [m]\}| \ge 1$: 
\begin{enumerate}
    \item if 
    $|\{i': \optscheme_{i'}(\scaledNoiseDiff) > 0, i'\in [m]\}| = \{i\}$, 
    let $\newscheme_{i}(\scaledNoiseDiff) = \optscheme_i(\scaledNoiseDiff)$;

    \item if 
    $|\{i': \optscheme_{i'}(\scaledNoiseDiff) > 0, i'\in [m]\}| = \{i, j\}$ where $i < j$, 
    let $\leftDelta \triangleq \max\{x \in \discretizedSupp: x\le \scaledNoiseDiff\}$, and let 
    $\rightDelta \triangleq \min\{x \in \discretizedSupp: x\ge \scaledNoiseDiff\}$.
    Let 
    $\newscheme_i(\leftDelta)
    = 
    \frac{\receiverU_j - \leftDelta}{\receiverU_j - \receiverU_i}  \frac{1}{\prior_i}  \frac{\optscheme(\scaledNoiseDiff) (\rightDelta-\scaledNoiseDiff)}{\rightDelta-\leftDelta}$ and 
    $\newscheme_i(\rightDelta) 
    = 
    \optscheme_i(\scaledNoiseDiff) - \newscheme_i(\leftDelta)$;
    $\newscheme_j(\leftDelta) 
    = 
    \frac{\leftDelta- \receiverU_i}{\receiverU_j - \receiverU_i}  \frac{1}{\prior_j}  \frac{\optscheme(\scaledNoiseDiff) (\rightDelta-\scaledNoiseDiff)}{\rightDelta-\leftDelta}$ 
    and 
    $\newscheme_j(\rightDelta) 
    = 
    \optscheme_j(\scaledNoiseDiff) - \newscheme_j(\leftDelta)$
    where $\optscheme(\scaledNoiseDiff) = \prior_i \optscheme_i(\scaledNoiseDiff) + \prior_j\optscheme_j(\scaledNoiseDiff)$.
\end{enumerate}
By construction, it is easy to verify that the signaling scheme $\newscheme$ is a feasible solution to the program \ref{eq:opt lp - approx}. 
Furthermore, when $|\{i': \optscheme_{i'}(\scaledNoiseDiff) > 0, i'\in [m]\}| = \{i\}$, the expected payoff contributed from 
the induced $\scaledNoiseDiff$ in both signaling scheme 
$\newscheme, \optscheme$ equals to 
$\prior_i\senderU_i\newscheme_i(\scaledNoiseDiff) \Quant(\scaledNoiseDiff)$; 
when $|\{i': \optscheme_{i'}(\scaledNoiseDiff) > 0, i'\in [m]\}| = \{i, j\}$, the expected payoff contributed from 
the induced $\scaledNoiseDiff$ in both signaling scheme 
$\newscheme, \optscheme$ satisfy that
\begin{align*}
    & \prior_i\senderU_i\newscheme_i(\leftDelta)\Quant(\leftDelta)
    + \prior_j\senderU_j\newscheme_j(\leftDelta)\Quant(\leftDelta)
    + 
    \prior_i\senderU_i\newscheme_i(\rightDelta)\Quant(\rightDelta)
    + \prior_j\senderU_j\newscheme_j(\rightDelta)\Quant(\rightDelta) \\
    \overset{(a)}{\ge} ~~ &  
    \prior_i\senderU_i\optscheme_i(\scaledNoiseDiff)\cdot \Quant(\scaledNoiseDiff + \discrePara) 
    + 
    \prior_j\senderU_j\optscheme_j(\scaledNoiseDiff)\cdot \Quant(\scaledNoiseDiff + \discrePara) j
\end{align*}
where in inequality (a), we have used the fact that 
$\optscheme_i(\scaledNoiseDiff) = \newscheme_i(\leftDelta)+\newscheme_i(\rightDelta), 
\optscheme_j(\scaledNoiseDiff) = \newscheme_j(\leftDelta)+\newscheme_j(\rightDelta)$, 
$\leftDelta\le \scaledNoiseDiff\le\rightDelta$, 
$\rightDelta \le \scaledNoiseDiff+\discrePara$, and the fact that 
$\Quant(\cdot)$ is monotone non-increasing.
Summing over all $\scaledNoiseDiff$ and rearranging the terms, we know that
\begin{align*}
    \sum_{i\in[m]} \prior_i\senderU_i \int_\scaledNoiseDiff \optscheme_i(\scaledNoiseDiff) \Quant(\scaledNoiseDiff+\discrePara) d\scaledNoiseDiff
    \le \Payoff[\noiseScale]{\newscheme}
\end{align*}
Now observe that for any $\scaledNoiseDiff$, we have 
$\frac{\Quant(\scaledNoiseDiff)}{\Quant(\scaledNoiseDiff+\discrePara)} 
\le \exp(\noiseScale\discrePara)$.
Thus, with the above inequality, we have 
$\Payoff[\noiseScale]{\optscheme} 
\le\exp(\noiseScale\discrePara) \Payoff[\noiseScale]{\newscheme}
\le (1+ \noiseScale\discrePara)  \Payoff[\noiseScale]{\newscheme}
= (1+ \varepsilon_0)  \Payoff[\noiseScale]{\newscheme}$.
\end{proof}

%% file: 0_main.bbl
\begin{thebibliography}{72}
\providecommand{\natexlab}[1]{#1}
\providecommand{\url}[1]{\texttt{#1}}
\expandafter\ifx\csname urlstyle\endcsname\relax
  \providecommand{\doi}[1]{doi: #1}\else
  \providecommand{\doi}{doi: \begingroup \urlstyle{rm}\Url}\fi

\bibitem[Akbarpour and Li(2020)]{AL-20}
Mohammad Akbarpour and Shengwu Li.
\newblock Credible auctions: A trilemma.
\newblock \emph{Econometrica}, 88\penalty0 (2):\penalty0 425--467, 2020.

\bibitem[Allouah and Besbes(2020)]{AB-20}
Amine Allouah and Omar Besbes.
\newblock Prior-independent optimal auctions.
\newblock \emph{Management Science}, 66\penalty0 (10):\penalty0 4417--4432,
  2020.

\bibitem[Alonso and C{\^a}mara(2016{\natexlab{a}})]{AC-16}
Ricardo Alonso and Odilon C{\^a}mara.
\newblock Political disagreement and information in elections.
\newblock \emph{Games and Economic Behavior}, 100:\penalty0 390--412,
  2016{\natexlab{a}}.

\bibitem[Alonso and C{\^a}mara(2016{\natexlab{b}})]{AC-16b}
Ricardo Alonso and Odilon C{\^a}mara.
\newblock Persuading voters.
\newblock \emph{American Economic Review}, 106\penalty0 (11):\penalty0
  3590--3605, 2016{\natexlab{b}}.

\bibitem[Anderson and Nash(1987)]{AN-87}
Edward~J Anderson and Peter Nash.
\newblock \emph{Linear programming in infinite-dimensional spaces: theory and
  applications}.
\newblock John Wiley \& Sons, 1987.

\bibitem[Anunrojwong et~al.(2020)Anunrojwong, Iyer, and Lingenbrink]{AIL-20}
Jerry Anunrojwong, Krishnamurthy Iyer, and David Lingenbrink.
\newblock Persuading risk-conscious agents: A geometric approach.
\newblock \emph{Available at SSRN 3386273}, 2020.

\bibitem[Arieli and Babichenko(2019)]{AB-19}
Itai Arieli and Yakov Babichenko.
\newblock Private bayesian persuasion.
\newblock \emph{Journal of Economic Theory}, 182:\penalty0 185--217, 2019.

\bibitem[Assadi et~al.(2022)Assadi, Khandeparkar, Saxena, and
  Weinberg]{AKSW-22}
Sepehr Assadi, Hrishikesh Khandeparkar, Raghuvansh~R Saxena, and S~Matthew
  Weinberg.
\newblock Separating the communication complexity of truthful and nontruthful
  algorithms for combinatorial auctions.
\newblock \emph{SIAM Journal on Computing}, \penalty0 (0):\penalty0 STOC20--75,
  2022.

\bibitem[Aybas and Turkel(2019)]{AT-19}
Yunus~C Aybas and Eray Turkel.
\newblock Persuasion with coarse communication.
\newblock \emph{arXiv preprint arXiv:1910.13547}, 2019.

\bibitem[Babichenko and Barman(2017)]{BB-17}
Yakov Babichenko and Siddharth Barman.
\newblock Algorithmic aspects of private bayesian persuasion.
\newblock In \emph{8th Innovations in Theoretical Computer Science Conference},
  2017.

\bibitem[Babichenko et~al.(2021)Babichenko, Talgam-Cohen, Xu, and
  Zabarnyi]{BTXZ-21}
Yakov Babichenko, Inbal Talgam-Cohen, Haifeng Xu, and Konstantin Zabarnyi.
\newblock Regret-minimizing bayesian persuasion.
\newblock In \emph{Proceedings of the 22nd ACM Conference on Economics and
  Computation}, pages 128--128, 2021.

\bibitem[Badanidiyuru et~al.(2018)Badanidiyuru, Bhawalkar, and Xu]{BBX-18}
Ashwinkumar Badanidiyuru, Kshipra Bhawalkar, and Haifeng Xu.
\newblock Targeting and signaling in ad auctions.
\newblock In \emph{Proceedings of the 29th Annual ACM-SIAM Symposium on
  Discrete Algorithms}, pages 2545--2563, 2018.

\bibitem[Bergemann and Morris(2019)]{BM-19}
Dirk Bergemann and Stephen Morris.
\newblock Information design: A unified perspective.
\newblock \emph{Journal of Economic Literature}, 57\penalty0 (1):\penalty0
  44--95, 2019.

\bibitem[Bergemann et~al.(2022{\natexlab{a}})Bergemann, Cai, Velegkas, and
  Zhao]{BCGZ-22}
Dirk Bergemann, Yang Cai, Grigoris Velegkas, and Mingfei Zhao.
\newblock Is selling complete information (approximately) optimal?
\newblock In \emph{Proceedings of the 23rd ACM Conference on Economics and
  Computation}, pages 608--663, 2022{\natexlab{a}}.

\bibitem[Bergemann et~al.(2022{\natexlab{b}})Bergemann, Heumann, Morris,
  Sorokin, and Winter]{BHMSW-22}
Dirk Bergemann, Tibor Heumann, Stephen Morris, Constantine Sorokin, and Eyal
  Winter.
\newblock Optimal information disclosure in auctions.
\newblock \emph{American Economic Review: Insights}, 2022{\natexlab{b}}.

\bibitem[Blackwell(1953)]{bla-53}
David Blackwell.
\newblock Equivalent comparisons of experiments.
\newblock \emph{The annals of mathematical statistics}, pages 265--272, 1953.

\bibitem[Bloedel and Segal(2020)]{BS-20}
Alexander~W Bloedel and Ilya Segal.
\newblock Persuading a rationally inattentive agent.
\newblock Technical report, Working paper, 2020.

\bibitem[Braverman et~al.(2018)Braverman, Mao, Schneider, and
  Weinberg]{BMSW-18}
Mark Braverman, Jieming Mao, Jon Schneider, and Matt Weinberg.
\newblock Selling to a no-regret buyer.
\newblock In \emph{Proceedings of the 19th ACM Conference on Economics and
  Computation}, pages 523--538, 2018.

\bibitem[Cai et~al.(2019)Cai, Thomas, and Weinberg]{CTW-19}
Linda Cai, Clayton Thomas, and S~Matthew Weinberg.
\newblock Implementation in advised strategies: Welfare guarantees from
  posted-price mechanisms when demand queries are np-hard.
\newblock \emph{arXiv preprint arXiv:1910.04342}, 2019.

\bibitem[Camara et~al.(2020)Camara, Hartline, and Johnsen]{CHJ-20}
Modibo~K Camara, Jason~D Hartline, and Aleck Johnsen.
\newblock Mechanisms for a no-regret agent: Beyond the common prior.
\newblock In \emph{2020 IEEE 61st annual symposium on foundations of computer
  science}, pages 259--270, 2020.

\bibitem[Candogan(2019)]{C-19}
Ozan Candogan.
\newblock Optimality of double intervals in persuasion: A convex programming
  framework.
\newblock \emph{Available at SSRN 3452145}, 2019.

\bibitem[Castiglioni et~al.(2020)Castiglioni, Celli, Marchesi, and
  Gatti]{CCMG-20}
Matteo Castiglioni, Andrea Celli, Alberto Marchesi, and Nicola Gatti.
\newblock Online bayesian persuasion.
\newblock \emph{Advances in Neural Information Processing Systems},
  33:\penalty0 16188--16198, 2020.

\bibitem[Castiglioni et~al.(2021)Castiglioni, Marchesi, Celli, and
  Gatti]{CMCG-21}
Matteo Castiglioni, Alberto Marchesi, Andrea Celli, and Nicola Gatti.
\newblock Multi-receiver online bayesian persuasion.
\newblock In \emph{International Conference on Machine Learning}, pages
  1314--1323. PMLR, 2021.

\bibitem[Chawla et~al.(2018)Chawla, Goldner, Miller, and Pountourakis]{CGMP-18}
Shuchi Chawla, Kira Goldner, J~Benjamin Miller, and Emmanouil Pountourakis.
\newblock Revenue maximization with an uncertainty-averse buyer.
\newblock In \emph{Proceedings of the 29th Annual ACM-SIAM Symposium on
  Discrete Algorithms}, pages 2050--2068, 2018.

\bibitem[Chawla et~al.(2022)Chawla, Devanur, Karlin, and Sivan]{CDKS-22}
Shuchi Chawla, Nikhil~R Devanur, Anna~R Karlin, and Balasubramanian Sivan.
\newblock Simple pricing schemes for consumers with evolving values.
\newblock \emph{Games and Economic Behavior}, 2022.

\bibitem[Chen and Lin(2023)]{CL-23}
Yiling Chen and Tao Lin.
\newblock Persuading a behavioral agent: Approximately best responding and
  learning.
\newblock \emph{arXiv preprint arXiv:2302.03719}, 2023.

\bibitem[Clippel and Zhang(2022)]{CZ-22}
Geoffroy~de Clippel and Xu~Zhang.
\newblock Non-bayesian persuasion.
\newblock 2022.

\bibitem[Cristian(1991)]{cri-91}
Flavin Cristian.
\newblock Understanding fault-tolerant distributed systems.
\newblock \emph{Communications of the ACM}, 34\penalty0 (2):\penalty0 56--78,
  1991.

\bibitem[Cummings et~al.(2020)Cummings, Devanur, Huang, and Wang]{CDHW-20}
Rachel Cummings, Nikhil~R Devanur, Zhiyi Huang, and Xiangning Wang.
\newblock Algorithmic price discrimination.
\newblock In \emph{Proceedings of the 14th Annual ACM-SIAM Symposium on
  Discrete Algorithms}, pages 2432--2451, 2020.

\bibitem[Daskalakis et~al.(2020)Daskalakis, Fishelson, Lucier, Syrgkanis, and
  Velusamy]{DFLSV-20}
Constantinos Daskalakis, Maxwell Fishelson, Brendan Lucier, Vasilis Syrgkanis,
  and Santhoshini Velusamy.
\newblock Simple, credible, and approximately-optimal auctions.
\newblock In \emph{Proceedings of the 21st ACM Conference on Economics and
  Computation}, pages 713--713, 2020.

\bibitem[Deng et~al.(2019{\natexlab{a}})Deng, Schneider, and Sivan]{DSS-19a}
Yuan Deng, Jon Schneider, and Balasubramanian Sivan.
\newblock Strategizing against no-regret learners.
\newblock \emph{Advances in neural information processing systems}, 32,
  2019{\natexlab{a}}.

\bibitem[Deng et~al.(2019{\natexlab{b}})Deng, Schneider, and Sivan]{DSS-19b}
Yuan Deng, Jon Schneider, and Balasubramanian Sivan.
\newblock Prior-free dynamic auctions with low regret buyers.
\newblock \emph{Advances in Neural Information Processing Systems}, 32,
  2019{\natexlab{b}}.

\bibitem[Dhangwatnotai et~al.(2015)Dhangwatnotai, Roughgarden, and Yan]{DRY-15}
Peerapong Dhangwatnotai, Tim Roughgarden, and Qiqi Yan.
\newblock Revenue maximization with a single sample.
\newblock \emph{Games and Economic Behavior}, 91:\penalty0 318--333, 2015.

\bibitem[Dughmi(2017)]{D-17}
Shaddin Dughmi.
\newblock Algorithmic information structure design: a survey.
\newblock \emph{ACM SIGecom Exchanges}, 15\penalty0 (2):\penalty0 2--24, 2017.

\bibitem[Dughmi and Peres(2012)]{DP-12}
Shaddin Dughmi and Yuval Peres.
\newblock Mechanisms for risk averse agents, without loss.
\newblock \emph{arXiv preprint arXiv:1206.2957}, 2012.

\bibitem[Dughmi and Xu(2016)]{DX-16}
Shaddin Dughmi and Haifeng Xu.
\newblock Algorithmic bayesian persuasion.
\newblock In \emph{Proceedings of the forty-eighth annual ACM symposium on
  Theory of Computing}, pages 412--425, 2016.

\bibitem[Dughmi and Xu(2017)]{DX-17}
Shaddin Dughmi and Haifeng Xu.
\newblock Algorithmic persuasion with no externalities.
\newblock In \emph{Proceedings of the 18th ACM Conference on Economics and
  Computation}, pages 351--368, 2017.

\bibitem[Dughmi et~al.(2016)Dughmi, Kempe, and Qiang]{DKQ-16}
Shaddin Dughmi, David Kempe, and Ruixin Qiang.
\newblock Persuasion with limited communication.
\newblock In \emph{Proceedings of the 17th ACM Conference on Economics and
  Computation}, pages 663--680, 2016.

\bibitem[Dworczak and Martini(2019)]{DM-19}
Piotr Dworczak and Giorgio Martini.
\newblock The simple economics of optimal persuasion.
\newblock \emph{Journal of Political Economy}, 127\penalty0 (5):\penalty0
  1993--2048, 2019.

\bibitem[Dworczak and Pavan(2022)]{DP-22}
Piotr Dworczak and Alessandro Pavan.
\newblock Preparing for the worst but hoping for the best: Robust (bayesian)
  persuasion.
\newblock \emph{Econometrica}, 90\penalty0 (5):\penalty0 2017--2051, 2022.

\bibitem[Easley and Ghosh(2015)]{EG-15}
David Easley and Arpita Ghosh.
\newblock Behavioral mechanism design: Optimal crowdsourcing contracts and
  prospect theory.
\newblock In \emph{Proceedings of the 16th ACM Conference on Economics and
  Computation}, pages 679--696, 2015.

\bibitem[Emek et~al.(2014)Emek, Feldman, Gamzu, PaesLeme, and
  Tennenholtz]{EFGPT-14}
Yuval Emek, Michal Feldman, Iftah Gamzu, Renato PaesLeme, and Moshe
  Tennenholtz.
\newblock Signaling schemes for revenue maximization.
\newblock \emph{ACM Transactions on Economics and Computation}, 2\penalty0
  (2):\penalty0 1--19, 2014.

\bibitem[Feng and Hartline(2018)]{FH-18}
Yiding Feng and Jason~D Hartline.
\newblock An end-to-end argument in mechanism design (prior-independent
  auctions for budgeted agents).
\newblock In \emph{2018 IEEE 59th Annual Symposium on Foundations of Computer
  Science}, pages 404--415, 2018.

\bibitem[Feng et~al.(2022)Feng, Tang, and Xu]{FTX-22}
Yiding Feng, Wei Tang, and Haifeng Xu.
\newblock Online bayesian recommendation with no regret.
\newblock In \emph{Proceedings of the 23rd ACM Conference on Economics and
  Computation}, page 818–819, 2022.

\bibitem[Fu et~al.(2013)Fu, Hartline, and Hoy]{FHH-13}
Hu~Fu, Jason Hartline, and Darrell Hoy.
\newblock Prior-independent auctions for risk-averse agents.
\newblock In \emph{Proceedings of the 14th ACM conference on Electronic
  commerce}, pages 471--488, 2013.

\bibitem[Fu et~al.(2015)Fu, Immorlica, Lucier, and Strack]{FILS-15}
Hu~Fu, Nicole Immorlica, Brendan Lucier, and Philipp Strack.
\newblock Randomization beats second price as a prior-independent auction.
\newblock In \emph{Proceedings of the 16th ACM Conference on Economics and
  Computation}, pages 323--323, 2015.

\bibitem[Gan et~al.(2023)Gan, Han, Wu, and Xu]{GHWX-23}
Jiarui Gan, Minbiao Han, Jibang Wu, and Haifeng Xu.
\newblock Robust stackelberg equilibria.
\newblock \emph{arXiv preprint arXiv:2304.14990}, 2023.

\bibitem[Gradwohl et~al.(2022)Gradwohl, Hahn, Hoefer, and Smorodinsky]{GHHS-22}
Ronen Gradwohl, Niklas Hahn, Martin Hoefer, and Rann Smorodinsky.
\newblock Algorithms for persuasion with limited communication.
\newblock \emph{Mathematics of Operations Research}, 2022.

\bibitem[Guo and Shmaya(2019)]{GS-19}
Yingni Guo and Eran Shmaya.
\newblock The interval structure of optimal disclosure.
\newblock \emph{Econometrica}, 87\penalty0 (2):\penalty0 653--675, 2019.

\bibitem[Hartline et~al.(2020)Hartline, Johnsen, and Li]{HJL-20}
Jason Hartline, Aleck Johnsen, and Yingkai Li.
\newblock Benchmark design and prior-independent optimization.
\newblock In \emph{2020 IEEE 61st Annual Symposium on Foundations of Computer
  Science}, pages 294--305, 2020.

\bibitem[Hu and Weng(2021)]{HW-21}
Ju~Hu and Xi~Weng.
\newblock Robust persuasion of a privately informed receiver.
\newblock \emph{Economic Theory}, 72:\penalty0 909--953, 2021.

\bibitem[Kamenica(2019)]{Kamenica-19}
Emir Kamenica.
\newblock Bayesian persuasion and information design.
\newblock \emph{Annual Review of Economics}, 11:\penalty0 249--272, 2019.

\bibitem[Kamenica and Gentzkow(2011)]{KG-11}
Emir Kamenica and Matthew Gentzkow.
\newblock Bayesian persuasion.
\newblock \emph{American Economic Review}, 101\penalty0 (6):\penalty0
  2590--2615, 2011.

\bibitem[Kolotilin(2018)]{K-18}
Anton Kolotilin.
\newblock Optimal information disclosure: A linear programming approach.
\newblock \emph{Theoretical Economics}, 13\penalty0 (2):\penalty0 607--635,
  2018.

\bibitem[Kolotilin et~al.(2017)Kolotilin, Mylovanov, Zapechelnyuk, and
  Li]{KMZL-17}
Anton Kolotilin, Tymofiy Mylovanov, Andriy Zapechelnyuk, and Ming Li.
\newblock Persuasion of a privately informed receiver.
\newblock \emph{Econometrica}, 85\penalty0 (6):\penalty0 1949--1964, 2017.

\bibitem[Kolotilin et~al.(2022)Kolotilin, Mylovanov, and
  Zapechelnyuk]{KMZ-2022}
Anton Kolotilin, Timofiy Mylovanov, and Andriy Zapechelnyuk.
\newblock Censorship as optimal persuasion.
\newblock \emph{Theoretical Economics}, 17\penalty0 (2):\penalty0 561--585,
  2022.

\bibitem[Koren and Krishna(2020)]{KK-20}
Israel Koren and C~Mani Krishna.
\newblock \emph{Fault-tolerant systems}.
\newblock Morgan Kaufmann, 2020.

\bibitem[Kosterina(2022)]{K-22}
Svetlana Kosterina.
\newblock Persuasion with unknown beliefs.
\newblock \emph{Theoretical Economics}, 17\penalty0 (3):\penalty0 1075--1107,
  2022.

\bibitem[Laprie(1992)]{lap-92}
Jean-Claude Laprie.
\newblock Dependability: Basic concepts and terminology.
\newblock In \emph{Dependability: Basic Concepts and Terminology}, pages
  3--245. Springer, 1992.

\bibitem[Le~Treust and Tomala(2019)]{LT-19}
Ma{\"e}l Le~Treust and Tristan Tomala.
\newblock Persuasion with limited communication capacity.
\newblock \emph{Journal of Economic Theory}, 184:\penalty0 104940, 2019.

\bibitem[Lingenbrink and Iyer(2019)]{LI-19}
David Lingenbrink and Krishnamurthy Iyer.
\newblock Optimal signaling mechanisms in unobservable queues.
\newblock \emph{Operations research}, 67\penalty0 (5):\penalty0 1397--1416,
  2019.

\bibitem[Lipnowski et~al.(2020)Lipnowski, Mathevet, and Wei]{LMW-20}
Elliot Lipnowski, Laurent Mathevet, and Dong Wei.
\newblock Attention management.
\newblock \emph{American Economic Review: Insights}, 2\penalty0 (1):\penalty0
  17--32, 2020.

\bibitem[Mansour et~al.(2022)Mansour, Slivkins, Syrgkanis, and Wu]{MSSW-22}
Yishay Mansour, Alex Slivkins, Vasilis Syrgkanis, and Zhiwei~Steven Wu.
\newblock Bayesian exploration: Incentivizing exploration in bayesian games.
\newblock \emph{Operations Research}, 70\penalty0 (2):\penalty0 1105--1127,
  2022.

\bibitem[Mat{\v{e}}jka and McKay(2015)]{MM-15}
Filip Mat{\v{e}}jka and Alisdair McKay.
\newblock Rational inattention to discrete choices: A new foundation for the
  multinomial logit model.
\newblock \emph{American Economic Review}, 105\penalty0 (1):\penalty0 272--98,
  2015.

\bibitem[McKelvey and Palfrey(1995)]{MP-95}
Richard~D McKelvey and Thomas~R Palfrey.
\newblock Quantal response equilibria for normal form games.
\newblock \emph{Games and economic behavior}, 10\penalty0 (1):\penalty0 6--38,
  1995.

\bibitem[Renault et~al.(2017)Renault, Solan, and Vieille]{RSV-17}
J{\'e}r{\^o}me Renault, Eilon Solan, and Nicolas Vieille.
\newblock Optimal dynamic information provision.
\newblock \emph{Games and Economic Behavior}, 104:\penalty0 329--349, 2017.

\bibitem[Rust(1987)]{rus-87}
John Rust.
\newblock Optimal replacement of gmc bus engines: An empirical model of harold
  zurcher.
\newblock \emph{Econometrica: Journal of the Econometric Society}, pages
  999--1033, 1987.

\bibitem[Shmoys and Tardos(1993)]{ST-93}
David~B Shmoys and {\'E}va Tardos.
\newblock An approximation algorithm for the generalized assignment problem.
\newblock \emph{Mathematical programming}, 62\penalty0 (1):\penalty0 461--474,
  1993.

\bibitem[Talluri and Van~Ryzin(2004)]{TV-04}
Kalyan Talluri and Garrett Van~Ryzin.
\newblock Revenue management under a general discrete choice model of consumer
  behavior.
\newblock \emph{Management Science}, 50\penalty0 (1):\penalty0 15--33, 2004.

\bibitem[Tang and Ho(2021)]{TH-21}
Wei Tang and Chien-Ju Ho.
\newblock On the bayesian rational assumption in information design.
\newblock In \emph{Proceedings of the AAAI Conference on Human Computation and
  Crowdsourcing}, volume~9, pages 120--130, 2021.

\bibitem[Xu(2020)]{X-20}
Haifeng Xu.
\newblock On the tractability of public persuasion with no externalities.
\newblock In \emph{Proceedings of the 14th Annual ACM-SIAM Symposium on
  Discrete Algorithms}, pages 2708--2727. SIAM, 2020.

\bibitem[Yu et~al.(2023)Yu, Tang, Narayanan, and Ho]{YTNH-23}
Guanghui Yu, Wei Tang, Saumik Narayanan, and Chien-Ju Ho.
\newblock Encoding human behavior in information design through deep learning.
\newblock \emph{Advances in neural information processing systems}, 2023.

\end{thebibliography}
